\documentclass[11pt,leqno,fleqn]{article}

\newcommand{\Comment}[1]{\relax}
\pdfoutput=1  

\usepackage{amsthm}
\usepackage{latexsym}
\usepackage{amsmath}
\usepackage{amssymb}
\usepackage{amsfonts}
\usepackage{algpseudocode}
\usepackage{algorithm}
\usepackage{algorithmicx}
\usepackage{amsbsy}

\usepackage{enumitem} 

\usepackage[colorlinks=true,
            citecolor=blue,filecolor=blue,linkcolor=blue,urlcolor=blue,  
            linktoc=page,
            pagebackref=true]{hyperref} 
\usepackage{url}      
\usepackage{fullpage} 
\usepackage{afterpage}  
\usepackage{float} 
\floatstyle{boxed} 

\usepackage{times}
\usepackage{bm}   
\usepackage{mathrsfs}  
\usepackage{graphicx}
\usepackage{wrapfig}
\usepackage{fancybox}
\usepackage[usenames]{color}
\usepackage{alltt}

\newcommand{\Hide}[1]{}

\newif
\ifnote
\notetrue

\newif
\ifTR
 \TRtrue

\newtheorem{theorem}{Theorem}
\newtheorem{lemma}[theorem]{Lemma}

\newtheorem{proposition}[theorem]{Proposition}

\newtheorem{fact}[theorem]{Fact}
\newtheorem{restxxx}[theorem]{Restriction}

\newtheorem{agreexxx}[theorem]{Agreement}
\newenvironment{agreement}{\begin{agreexxx}\rm}{\hfill\QED\end{agreexxx}}
\newtheorem{termxxx}[theorem]{Terminology}
\newenvironment{terminology}{\begin{termxxx}\rm}{\hfill\QED\end{termxxx}}
\newtheorem{notxxx}[theorem]{Notation}

\newtheorem{assumxxx}[theorem]{Assumption}
\newenvironment{assumption}{\begin{assumxxx}\rm}{\hfill\QED\end{assumxxx}}
\newtheorem{convenxxx}[theorem]{Convention}

\newtheorem{exaxxx}[theorem]{Example}
\newenvironment{example}{\begin{exaxxx}\rm}{\hfill\QED\end{exaxxx}}
\newtheorem{exexxx}[theorem]{Exercise}

\newtheorem{remxxx}[theorem]{Remark}
\newenvironment{remark}{\begin{remxxx}\rm}{\hfill\QED\end{remxxx}}
\newtheorem{openxxx}[theorem]{Open Problem}
\newenvironment{problem}{\begin{openxxx}\rm}{\end{openxxx}}%
\newtheorem{conjxxx}[theorem]{Conjecture}

\newtheorem{defxxx}[theorem]{Definition}
\newenvironment{definition}[1]{\begin{defxxx}[\emph{#1}]\rm}%
{\hfill\QED\end{defxxx}}
\newtheorem{procxxx}[theorem]{Procedure}
{\hfill\QED\end{procxxx}}

\newtheorem{Prxxx}[theorem]{Proof}
\newenvironment{Proof}[1]{\begin{Prxxx}[\emph{#1}]\rm}%
{\end{Prxxx}} 

\newenvironment{custommargins}[2]%
  {\addtolength{\leftskip}{#1}\addtolength{\rightskip}{#2}}{\par}

\usepackage[margin=0pt,justification=centerlast,font=small,format=hang,%
labelfont=bf,up,textfont=rm,up]{caption}

\newcommand{\Angles}[1]{\langle #1 \rangle}
\newcommand{\Set}[1]{\{ #1 \}}
\newcommand{\SET}[1]{\bigl\{ #1 \bigr\}}
\newcommand{\constOf}[1]{\mathsf{Constraints}(#1)}
\newcommand{\tmin}[2]{T^{\text{\rm min}}_{#1}(#2)}
\newcommand{\tmax}[2]{T^{\text{\rm max}}_{#1}(#2)}

\newcommand{\aaa}{\mathbf{A}}
\newcommand{\bbb}{\mathbf{B}}

\newcommand{\mm}{\mathbf{M}}
\newcommand{\nn}{\mathbf{N}}

\newcommand{\M}{{\cal M}}
\newcommand{\N}{{\cal N}}
\newcommand{\PP}{{\cal P}}

\newcommand{\bigO}[1]{{\cal O}\bigl(#1\bigr)} 
\newcommand{\bigOO}[1]{{\cal O}(#1)} 
\newcommand{\ConnP}[2]{(#1\,\bm{\|}\,#2)}
\newcommand{\ConnPP}[2]{#1\,\bm{\|}\,#2}
\newcommand{\MergeSym}{\mathsf{Merge}} 

\newcommand{\bind}[2]{\mathsf{Bind}\bigl(#1,#2\bigr)} 
  
\newcommand{\BindT}{\mathsf{BindPT}} 
\newcommand{\Bind}{\mathsf{Bind}} 

\newcommand{\BindTone}{\mathsf{BindPT}_1} 
\newcommand{\Break}[1]{\mathsf{BreakUp}(#1)} 
  
\newcommand{\Let}[3]%
    {\textbf{\textsf{let}}\ {#1}\,{#2}\ \textbf{\textsf{in}}\;{#3}\,}
\newcommand{\Try}[3]%
    {\textbf{\textsf{try}}\ {#1} {#2}\ \textbf{\textsf{in}}\;{#3}\;}
\newcommand{\Mix}[3]%
    {\textbf{\textsf{mix}}\ {#1} {#2}\ \textbf{\textsf{in}}\;{#3}\;}
\newcommand{\LET}[3]%
    {\textbf{\textsf{let}}^{\bm{*}}\ {#1} {#2}\ \textbf{\textsf{in}}\;{#3}\;}
\newcommand{\Letrec}[3]%
    {\textbf{\textsf{letrec}}\ {#1} {#2}\ \textbf{\textsf{in}}\;{#3}\;}

\newcommand{\subTsym}{<:}
\newcommand{\subT}[2]{#1\subTsym #2} 
\newcommand{\SubTsym}{\ll:}
\newcommand{\SubT}[2]{#1\SubTsym #2} 
\newcommand{\NotSubTsym}{\not{\ll}:}
\newcommand{\NotSubT}[2]{#1\NotSubTsym #2} 

\newcommand{\WPT}{\textsf{WholePT}}
\newcommand{\CompPT}{\mathsf{CompPT}} 
\newcommand{\CompMF}{\textsf{CompMaxFlow}}
\newcommand{\OnePT}{\mathsf{OneNodePT}} 
\newcommand{\Schedule}{\mathsf{BindSchedule}}

\newcommand{\degreeSym}{\mathit{deg}} 
\newcommand{\degree}[1]{{\degreeSym}(#1)}

\newcommand{\hook}[2]{\mathsf{hook}_{#1}(#2)} 
\newcommand{\bindingS}[1]{\mathsf{binding\text{-}strength}(#1)}
\newcommand{\commonA}[1]{\mathsf{joint\text{-}arcs}(#1)} 
\newcommand{\indegree}[1]{\mathsf{in\text{-}degree}(#1)} 
\newcommand{\outdegree}[1]{\mathsf{out\text{-}degree}(#1)}

\newcommand{\ie}{\textit{i.e.}}
\newcommand{\eg}{\textit{e.g.}}
\newcommand{\QED}{{\Large $\square$}} 

\newcommand{\posnats}{{\mathbb{N}}_{+}}
\newcommand{\nreals}{\mathbb{R}_{+}}

\newcommand{\reals}{\mathbb{R}}

\newcommand{\intervals}[1]{{\cal I}(#1)}

\newcommand{\poly}[1]{\mathsf{Poly}(#1)} 
\newcommand{\Poly}[1]{\mathsf{Poly}\bigl(#1\bigr)} 
\newcommand{\IndexSym}{\mathit{index}} 
\newcommand{\Index}[1]{{\IndexSym}(#1)} 
\newcommand{\size}[1]{|\,#1\,|}

\newcommand{\BOT}[1]{\mathsf{Bot}[#1]} 
\newcommand{\TOP}[1]{\mathsf{Top}[#1]} 

\newcommand{\power}[1]{\mathscr{P}(#1)}
\newcommand{\rest}[2]{[#1]_{#2}} 

\newcommand{\head}[1]{\text{\em head}(#1)} 
\newcommand{\tail}[1]{\text{\em tail}(#1)} 
\newcommand{\ParAdd}[2]{#1 \bm{\oplus} #2} 
\newcommand{\TotAdd}[2]{#1 \bm{\oplus_{\text{\rm t}}} #2} 

\newcommand{\dime}[1]{\mathsf{dim}(#1)} %
\newcommand{\exDim}[1]{\mathsf{exDim}(#1)} %

\newcommand{\lc}{\underline{\bm{c}}} 
\newcommand{\uc}{\overline{\bm{c}}} 
\newcommand{\OutF}[1]{\mathsf{OuterFace}(#1)}

\newcommand{\abs}[1]{\lvert #1\rvert}

\newcommand{\merg}[2]{#1 {\otimes} #2} 

\newcommand{\AAA}{\mathscr{A}}
\newcommand{\CC}{\mathscr{C}}

\newcommand{\EE}{\mathscr{E}}

\newcommand{\Valid}[1]{\mathsf{Valid}(#1)}

\newcommand{\spacing}[2]{
  \renewcommand{\baselinestretch}{#2}
  \small\normalsize #1
  \setlength{\parskip}{0.1\baselineskip}
  \settowidth{\parindent}{xxxx}
  \setlength{\parindent}{#2\parindent}
  \setlength{\leftmargini}{\parindent}
  \setlength{\leftmarginii}{\parindent}
  \setlength{\leftmarginiii}{\parindent}
  \setlength{\footnotesep}{#2\footnotesep}
}

\begin{document}

\spacing{\normalsize}{0.98}
\setcounter{page}{1}     
\setcounter{tocdepth}{1} 
\ifTR
  \pagenumbering{roman} 
\else
\fi

\title{A Compositional Approach to Network Algorithms} 
\author{Assaf Kfoury%
           \thanks{Partially supported by NSF awards CCF-0820138
       and CNS-1135722. This is an updated version of a report that
       was largely completed and circulated since 2014~\cite{kfourySTOC2014}.} \\
        Boston University \\
        \ifTR Boston, Massachusetts \\ \else \fi
        \href{mailto:kfoury@bu.edu}{kfoury{@}bu.edu}
}

\ifTR
   \date{\today}
\else
   \date{} %
\fi
\maketitle
  \ifTR
     \thispagestyle{empty} 
  \else
  \fi
\vspace{-.4in}
  \begin{abstract}

\noindent
We present elements of a typing theory for flow networks, where
``types'', ``typings'', and ``type inference'' are formulated in terms
of familiar notions from polyhedral analysis and convex
optimization. Based on this typing theory, we develop an alternative
approach to the design and analysis of network algorithms,
which we illustrate by applying it to the max-flow problem in
multiple-source, multiple-sink, capacited directed planar graphs.

  \end{abstract}
\ifTR
    \tableofcontents
    \newpage
    \pagenumbering{arabic}
\else
    \vspace{-.2in}
\fi

\section{Introduction}
\label{sect:intro}


\vspace{-.1in}
\paragraph{Background and motivation.}

The work reported herein stems from a group effort to develop an
integrated enviroment for system modeling and system 
analysis that are simultaneously: \emph{modular} (``distributed in
space''), \emph{incremental} (``distributed in time''),
and \emph{order-oblivious} (``components can be analyzed and
assembled in \emph{any} order''). These are the three defining properties 
of what we call a \emph{compositional} approach to system development.%
\Hide{
    \footnote{ 
    This is one of the projects currently in progress under the
    umbrella of the 
    {\it iBench Initiative} 
    at Boston University,
    co-directed by Azer Bestavros and Assaf Kfoury.  The
    website 
    {\footnotesize \url{https://sites.google.com/site/ibenchbu/}} 
    gives further details on this and other research activities.  }
}    
Several papers
explain how this environment is defined and used, as well as its
current state of development and implementation
\cite{BestavrosKfouryLapetsOcean:crts09,%
BestavrosKfouryLapetsOcean:hscc10,%
BestKfoury:dsl11,Kfoury:sblp11,kfouryDSL:2011,%
SouleBestKfouryLapets:eoolt11}.
An extra fortuitous benefit of our work has been a
fresh perspective on the design and analysis of network algorithms.

For this approach to succeed at all, we need to appropriately
encapsulate a system's components as they are modeled and become
available for analysis:%
       \footnote{In this Introduction, a ``component'' is not taken in
       the graph-theoretic sense of ``maximal connected subgraph''.
       It here means a ``subnetwork'' (or, if there is
       an underlying graph, a ``subgraph'') with input and output
       ports to connect it with other ``subnetworks''. } 
We hide their internal workings, but also
infer enough information to safely connect them at their boundaries
and to later guarantee the safe operation of the system as a whole.
The inferred information has to be formally encoded, and somehow
composable at the interfaces, to enforce safety invariants throughout
the process of assembling components and later during system
operation. This is precisely the traditional role assigned
to \emph{types} and \emph{typings} in a different context -- namely, for
a strongly-typed programming language, their purpose is to enforce
safety invariants across program modules and abstractions.  Naturally,
our types and typings will be formalized differently here, depending
on how we specify systems and on the choice of invariant properties.%
     \footnote{Example of a \emph{safety invariant} for programs: ``A
     boolean value is never divided by $5$.'' Example of
     a \emph{safety invariant} for networks: ``Conservation of flow is
     never violated at a network node.''}

To illustrate our methodology, we consider the classical 
\emph{max-flow problem} in capacited directed graphs.%
    \footnote{\label{foot:survey}%
    A comprehensive survey of algorithms 
    for the max-flow problem is nearly impossible, as it is one of the most
    studied optimization problems over several decades. A broad
    classification is still useful, depending on concepts, proof
    techniques and/or graph restrictions. There is the family of
    algorithms based on the concept of \emph{augmenting path},
    starting with the Ford-Fulkerson algorithm in the
    1950's~\cite{FordFulkerson1956}.  A refinement of the
    augmenting-path method is the \emph{blocking flow}
    method~\cite{Dinic1980}.  A later family of max-flow algorithms
    uses the \emph{preflow push} (or \emph{push relabel})
    method~\cite{Goldberg1985,GoldbergTarjan1986}. A survey of these
    families of max-flow algorithms to the end of the 1990's is
    in~\cite{Goldberg1998,Asano1999}.  Later papers combine
    variants of augmenting-path algorithms and related blocking-flow
    algorithms, variants of preflow-push algorithms, and algorithms
    combining different parts of all of these methodologies 
    \cite{MazzoniPallottinoScutella1991,Goldberg-Rao1998,Goldberg2009,Orlin2013}.
    Another late entry in this plethora of approaches uses
    the notion of \emph{pseudoflow} \cite{Hochbaum2008,Chandran-Hochbaum2009}.
    The most recent research includes max-flow algorithms restricted to
    planar graphs~\cite{borradaileKlein2009,10.1109/FOCS.2012.66,
    borradaileKlein2011,EisenstatK13}, approximate max-flow algorithms 
    restricted to undirected graphs~\cite{DBLP:journals/corr/abs-1304-2077,
    DBLP:journals/corr/abs-1304-2338}, and approximate and exact
    max-flow algorithms restricted to uncapacited undirected graphs
    using concepts of \emph{electrical flow}~\cite{christiano:2011,LeeRS13}.
    } 
Since it comes at no extra cost for us, we simultaneously consider
the \emph{min-flow problem} as well as the presence of \emph{multiple
sources} and \emph{multiple sinks}. The \emph{min-flow problem} is
meaningful only if arcs are assigned lower-bound capacities (or
thresholds) which feasible flows are not allowed to go
under. Every arc in our networks is therefore assigned two
capacities, one lower bound and one upper bound.

For favorable comparison with other approaches, as far as run-time
complexities are concerned, we limit our attention to \emph{planar
networks}, a sufficiently large class with many practical
applications. It is also a class that has been studied extensively,
often with further restrictions on the topology (\eg, undirected
graphs \emph{vs.} directed graphs) and/or the capacities (\eg,
integral \emph{vs.} rational). None of the latter restrictions are
necessary for our approach to work. However, as of now, if we lift the
planarity restriction, our run-time complexities exceed those of other
approaches.

We stress that our methodology has applicability beyond 
the max-flow problem: It can be applied to tackle other
network-related algorithmic problems, with different or additional
measures of what qualify as desirable solutions, even if the
associated run-time complexities are not linear or nearly linear.

\vspace{-.1in}
\paragraph{Overview of our methodology.}

The central concept of our approach is what we call a \emph{network
typing}. To make this work, a network (or network component) $\N$ is
allowed to have ``dangling'' arcs; in effect, $\N$ is allowed to have
multiple sources or \emph{input arcs} (\ie, arcs whose tails are
not incident to any node) and multiple sinks or \emph{output arcs}
(\ie, arcs whose heads are not incident to any node). Given a
network $\N$, now with multiple input arcs and multiple output arcs, a
typing for $\N$ is an algebraic characterization of all the
feasible flows in $\N$ -- including, in particular, all \emph{maximum}
feasible flows and all \emph{minimum} feasible flows.%

More precisely, a \emph{sound} typing $T$ for network $\N$ specifies
constraints on the latter's inputs and outputs, such that every
assignment $f$ of values to its input/output arcs satisfying these
constraints can be extended to a feasible flow $f'$ in $\N$. Moreover,
if the input/output constraints specified by $T$ are satisfied
by \emph{every} input/output assignment $f$ extendable to a feasible
flow $f'$, then we say that $T$ is \emph{complete} for $\N$. 
In analogy with a similar concept in strongly-typed programming
languages, we call \emph{principal} a typing which is both sound and
complete -- and satisfying a few additional syntactic requirements
for easier inference of types and typings.

\Hide{
Every sound typing is less general than a principal typing for the
same network -- technically, the latter is a ``subtyping'' of every
sound typing for the same network.
}

\Hide{
\newcommand{\mI}{m_{\text{\rm in}}}
\newcommand{\mO}{m_{\text{\rm out}}}
\newcommand{\mIO}{m_{\text{\rm \#}}}
}

In our formulation, a typing $T$ for network $\N$ defines a compact
convex polyhedral set (or \emph{polytope}), which we denote
$\poly{T}$, in the vector space ${\reals}^{p+q}$, where $\reals$ is
the set of reals, and $p$ and $q$ are the numbers of input arcs and output
arcs in $\N$. An input/output assignment $f$ satisfies
$T$ if $f$, viewed as a point in the space ${\reals}^{p+q}$, is inside
$\poly{T}$. Hence, $T$ is a \emph{sound} typing (resp.
\emph{sound+complete} or \emph{principal} typing) 
if $\poly{T}$ is \emph{contained in} (resp. \emph{equal to})
the set of all input/output assignments extendable to
feasible flows in $\N$.

Let $T_1$ and $T_2$ be principal typings for networks $\N_1$ and
$\N_2$.  If we connect $\N_1$ and $\N_2$ by linking some of their
output arcs to some of their input arcs, we obtain a new network which
we denote (only in this introduction) $\N_1\oplus\N_2$.  One of our 
results shows that a principal typing of $\N_1\oplus\N_2$ can be obtained
by direct (and relatively easy) algebraic operations on $T_1$ and $T_2$, 
without any need to re-examine the internal details of the two components
$\N_1$ and $\N_2$. Put differently, an analysis (to produce a principal typing)
for the assembled network $\N_1\oplus\N_2$ can be directly and easily
obtained from the analysis of $\N_1$ and the analysis of $\N_2$.

What we have just described is the counterpart of what
programming-language theorists call a \emph{modular}
(or \emph{syntax-directed}) \emph{analysis} (or \emph{type inference}), 
which infers a type for the whole program from the types of its
subprograms, and the latter from the types of their respective
subprograms, and so on recursively, down to the types of the smallest
program fragments.%
   \footnote{We will make a distinction between a ``type'' and a ``typing'',
    similar to a distinction made by programming-language theorists.
    }

Because our network typings denote polytopes, we can in fact make our
approach not only modular but also \emph{compositional}, now
mathematically stated as follows: If $T_1$ and $T_2$ are 
sound and complete typings for networks $\N_1$ and $\N_2$, then the calculation
of $T_1$ and the calculation of $T_2$ can be done 
independently of each other; that is, the analysis (to produce $T_1$) for
$\N_1$ and the analysis (to produce $T_2$) for $\N_2$ can be carried
out separately without prior knowledge that the two will be
subsequently assembled together.%
    \footnote{In the study of programming languages, there are type
     systems that support modular but not compositional analysis.
     What is compositional is modular, but not the other way around. A
     case in point is the so-called Hindley-Milner type system for
     ML-like functional languages, where the order matters in which
     types are inferred: Hindley-Milner type-inference is \emph{not}
     order-oblivious.}

Given a network $\N$ partitioned into finitely many components
$\N_1,\N_2,\N_3,\ldots$ with respective principal typings
$T_1,T_2,T_3,\ldots$, we can then assemble these typings in \emph{any}
order to obtain a principal typing $T$ for the whole of
$\N$. Efficiency in computing the final principal typing $T$ depends
on a judicious partitioning of $\N$, which is to decrease as much as
possible the number of arcs running between separate components, and
again recursively when assembling larger components from smaller
components. At the end of this procedure, every input/output assignment
$f$ extendable to a maximum feasible flow $f'$ in $\N$, and every
input/output assignment $g$ extendable to a minimum feasible flow $g'$,
can be directly read off the final typing $T$ -- but observe:
not $f'$ and $g'$ themselves.

\vspace{-.1in}
\paragraph{Whole-network \textit{versus} compositional.}

We qualified our approach as being \emph{compositional} because a
network is not required to be fully assembled, nor its constituent
components to be all available, in order to start an analysis of those
already in place and connected. What's more, an already-connected
component $A$ can be removed and swapped with another one $B$, as long
as $A$ and $B$ have the same typing, \ie, as far as the rest of the
network is concerned, the invariants encoded by typings are oblivious
to the swapping of $A$ and $B$.  In the conventional categories
of algorithm design and analysis, our compositional
approach can be viewed as a form of \emph{divide-and-conquer} that allows
the re-design of parts without forcing a re-analysis of the same
parts.

These aspects of compositionality are important when modeling very
large networks which may contain broken or missing components, or
failure-prone and obsolete components that need to be replaced.  But
if these aspects do not matter, then there is an immediate drawback to
our compositional approach, as currently devised and used: It returns
the value $\abs{f}$ of a maximum flow $f$, but not $f$ itself.

By contrast, other approaches (any of those cited in
footnote~\ref{foot:survey}) 
construct a specific maximum flow $f$, whose value $\abs{f}$ can be
immediately read off from the total leaving the source(s) or,
equivalently, the total entering the sink(s). We may qualify the other
approaches as being \emph{whole-network}, because they presume all the
pieces (nodes, arcs, and their capacities) of a network are in place
before an analysis is started.

There is more than one way to bypass the forementioned drawback of our
compositional approach, none entirely satisfactory (as of now). A
natural but costly option is to augment the information that typings
encode: A typing $T$ is made to also encode information about paths
that carry a maximum flow in the component for which $T$ is a
principal typing, but the incurred cost is prohibitive, generally
exponential in the external dimension $p+q$ of the component 
(the number of its input/output ports).%
   \footnote{It is not a trivial matter to augment a typing $T$ for a
   network component $\N$ so that it also encodes information about
   max-flow paths in $\N$. It is out of the question to retain
   information about all max-flow paths. What needs to be
   done is to encode, for every ``extreme'' input/output assignment
   $f$ extendable to a max-flow, just one path or path-combination
   carrying a max-flow extending $f$. (An input/output assignment $f$
   is \emph{extreme} if, as a point in the space
   ${\reals}^{p+q}$, it is a vertex of $\poly{T}$.)  The cost of this
   extra encoding grows exponentially with $p+q$. From the perspective
   of compositionality, this exponential growth adds to another
   disadvantage: The \emph{more} information we make the typing $T$ to
   encode about $\N$'s internals beyond safety invariants --
   unless the choice of internal paths in $\N$ to carry max-flows is
   taken as another safety condition -- the \emph{fewer} the
   components of which $T$ is a typing that we can substitute for
   $\N$, thus narrowing the range of experimentation and possible
   substitutions between components during modeling and analysis.  }

A more promising option is a two-phase process, yet to be
investigated.  In the first phase, we use our compositional approach
to return the value of a max-flow.  In the second phase, we use this
max-flow value to compute an actual maximum flow in the network. It
remains to be seen whether this is doable efficiently, or within
the resource bounds of our algorithms below. We delay this question
to future research. 

\Hide{
A solution for this problem is an
algorithm which, given an arbitrary network $\N$ with one source node
$s$ and one sink node $t$, computes a maximum feasible flow $f$ from
$s$ to $t$ in $\N$. That $f$ is a \emph{feasible flow} means $f$ is an
assignment of non-negative values to the arcs of $\N$
satisfying \emph{capacity constraints} at every arc and \emph{flow
conservation} at every node other than $s$ and $t$. That $f$
is \emph{maximum} means the net outflow at node $s$ (equivalently, the
net inflow at node $t$) is maximized.  A standard assessment of a
max-flow algorithm measures its run-time complexity as a function of
the size of $\N$. 
}

\vspace{-.1in}
\paragraph{Highlights and wider connections.} 

Our main contribution in this report is a different framework for the
design and analysis of network algorithms, which we here illustrate by
presenting a new algorithm for the classical max-flow problem.  When
restricted to the class of planar networks with bounded
``outerplanarity'' and bounded ``external dimension'', 
our algorithm runs in linear time.

The \emph{external dimension} (or \emph{interface dimension})
of a network is the number of its input/output ports,
\ie, the number $p$ of its sources $+$ the number $q$ of its sinks. 
A network's planar embedding has \emph{outerplanarity} $k\geqslant 1$ if
it has $k$ layers of nodes, \ie, after iteratively removing the nodes (and
incident arcs) on the outer face at most $k$ times, we obtain the empty
network. A planar network is of outerplanarity $k$ if it has a planar
embedding (not necessarily unique) of outerplanarity $k$. A
more precise statement of our final result is this: 
\begin{quote}
    \emph{Given fixed parameters $k\geqslant 1$ and $\ell\geqslant 2$,
    for every planar $n$-node network $\N$ of outerplanarity
    $\leqslant k$ and external dimension $\leqslant \ell$, our
    algorithm simultaneously returns a max-flow value and a min-flow
    value in time $\bigOO{n}$, where the hidden multiplicative
    constant depends on $k$ and $\ell$ only.}%
         \footnote{The usual trick of directing new arcs from an
          artificial source node to all source nodes and again from
          all sink nodes to an artificial sink node, in order to
          reduce the case of $p > 1$ sources and $q > 1$ sinks to the
          single-source single-sink case, generally destroys the
          planarity of $\N$ and cannot be used to simplify our
          algorithm.}
\end{quote}
Our final algorithm combines several intermediate algorithms, each of
independent interest for computing network typings. We mention several
salient features that distinguish our approach:
\begin{enumerate}
\item 
      Nowhere do we invoke a linear-programming algorithm (\eg, the
      simplex network algorithm). Our many optimizations relative to
      linear constraints are entirely carried out by various
      transformations on networks and their underlying graphs, and by
      using no more than the operations of addition, subtraction, and
      comparison of numbers.  At the end, our complexity bounds are
      all functions of only the number of nodes and arcs, and are
      independent of costs and capacities, \ie, these are
      strongly-polynomial bounds.
\item 
      In all cases, our algorithms do not impose any restrictions on
      flow capacities and costs. These capacities and costs can be
      arbitrarily large or small, independent of each other, and not
      restricted to integral values.
\item
      Part of our results are a contribution to the vast body of work on
      \emph{fixed-parameter} low-degree-polynomial time
      algorithms, or linear-time algorithms, for problems that are
      intractable (\eg, NP-hard) or impractical on very large input
      data (\eg, non-linear polynomial time) when these parameters are
      unrestricted.%
       \footnote{A useful though somewhat dated survey of efficient
      fixed-parameter algorithms is~\cite{bodlaender:TCS1998}.
      A recent survey in a focused area
      (transportation engineering) is~\cite{gaiu2013}
      where \emph{parameter-tuning} refers to alternatives in
      selecting fixed-parameter algorithms.}
\item
      In the process of building a full system from smaller
      components, interface dimensions figure prominently in our
      analysis. The quality of our results, in minimizing algorithm
      complexities and simplifying their proofs, depends on keeping
      interface dimensions as small as possible.  This part of our
      work rejoins research on efficient algorithms for graph
      separators and decomposition.
      (One of our results below depends
      on the linear-time computation of an optimal partioning of a
      $3$-regular planar embedding.)
\Hide{
\item As formulated in this report and unlike other approaches, our
      final algorithm returns only the \emph{value} of a maximum flow,
      without specifying a set of actual \emph{paths} from source to
      sink that will carry such a flow. Other approaches handle the
      two problems simultaneously: Inherent in their operation is
      that, in order to compute a maximum-flow \emph{value}, they need
      to determine a set of maximum-flow \emph{paths}; ours does not
      need to. Though avoided here because of the extra cost
      and complications for a first presentation, our final algorithm
      can be adjusted to return a set of maximum-flow paths in
      addition to a maximum-flow value.
\item We view the uncoupling of the two problems just described as an advantage.
      It underlies our need to be able to replace components -- broken
      or defective -- by other components as long as their principal
      typings are equal, without regard to how they may direct flow
      internally from input ports to output ports.
\item As far as run-time complexity is concerned, our final algorithm
      performs very badly on some networks, \eg, networks whose graphs
      are dense. However, on other special classes of networks, ours
      outperforms the best currently available algorithms (\eg, on
      networks whose graphs are outer-planar or whose graphs are
      topologically equivalent to some ring graphs).
}
\end{enumerate}

\vspace{-.2in}
\paragraph{Organization of the report.}

This introduction and the following sections until the reference
pages, less than ten pages, are an \emph{extended abstract}, 
which includes several propositions and theorems without their proofs.  
Proofs and further technical material in support of claims in the
earlier part of the report are in the four appendices after the
reference pages.
\Hide{
This extended abstract runs to less than $10$ pages, excluding the
references. 
We here delay some of the proofs, and justification of some
concepts, to appendices (not part of the $10$-page extended abstract).
}
Sections~\ref{sect:flow-networks}, 
\ref{sect:valid-vs-principal}, and~\ref{sect:disassemble-and-reassemble},
present elements of our compositional approach. 
Section~\ref{sect:max+min-flows-in-planar-networks} is our
application to the max-flow problem in planar networks.
Section~\ref{sect:extensions-and-future} presents immediate
extensions of this report and proposes
directions for future research.

\vspace{-.1in}
\section{Flow Networks and Their Typings}
\label{sect:flow-networks}

We take flow networks in their simplest form, as
capacited finite directed graphs. We repeat standard
notions~\cite{1993orlin}, but now adapted to our context.%
   \footnote{For our purposes, we need a definition of flow networks
   that is more arc-centric and less node-centric than the standard
   one. Such alternative definitions have already been
   proposed (see, for example, Chapter 2 in~\cite{Klein2011}), 
   but are still not the most convenient for us.  }
A \emph{flow network} $\N$ is a pair $\N = (\nn,\aaa)$, where $\nn$ is
a finite set of nodes and $\aaa$ a finite set of directed arcs, with
each arc connecting two distinct nodes (no self-loops and no multiple
arcs in the same direction connecting the same two nodes).  
We write $\reals$ and $\nreals$ for the
sets of reals and non-negative reals. Such a flow
network $\N$ is supplied with \emph{capacity functions} on the arcs,
%
   $\lc : \aaa\to\nreals$ (lower-bound capacity) and
   $\uc : \aaa\to\nreals$ (upper-bound capacity),
%
such that $0\leqslant \lc(a) \leqslant \uc(a)$ and $\uc(a) \neq 0$ for
every $a\in\aaa$. 

We write $\tail{a}$ and $\head{a}$ for the two ends of arc $a\in\aaa$.
The set $\aaa$ of arcs is the disjoint union of three sets, \ie, 
$\aaa = \aaa_\text{\#} \uplus \aaa_\text{in} \uplus \aaa_\text{out}$ where:
\begin{alignat*}{5}
   &\aaa_\text{\#} &&:= 
   \ \ &&\Set{\,a\in\aaa\;|\;\head{a}\in\nn\text{ \& }\tail{a}\in\nn\,} 
   \qquad &&\text{(the internal arcs of $\N$)},
\\
   &\aaa_\text{in} &&:= 
   &&\Set{\,a\in\aaa\;|\;\head{a}\in\nn\text{ \& }\tail{a}\not\in\nn\,} 
   \qquad &&\text{(the input arcs of $\N$)},
\\
   &\aaa_\text{out}\ &&:= 
   &&\Set{\,a\in\aaa\;|\;\head{a}\not\in\nn\text{ \& }\tail{a}\in\nn\,} 
   \qquad &&\text{(the output arcs of $\N$)} .
\end{alignat*}
%
%
%
\Hide{
With no loss of generality, we make the following
\emph{connectedness assumption}:
\begin{itemize}
\item[(\$\$)]\quad  For every $a\in\aaa$ there is a
       directed path from an input arc to an output arc
       that visits $a$.
\end{itemize}
We do not assume $\N$ is connected as a directed graph -- an
assumption often made in studies of network flows, which is sensible
when there is only one input arc (or ``source node'') and only one
output arc (or ``sink node'').
}
A \emph{flow} is a function $f :
\aaa\to\nreals$ which, if \emph{feasible}, satisfies ``flow
conservation'' at every node and ``capacity constraints'' at every
arc, both defined as in the standard formulation~\cite{1993orlin}.

We call a bounded closed interval $[r,r']$ of real numbers (possibly
negative) a \emph{type}. A \emph{typing} is a partial map $T$
(possibly total) that assigns types to subsets of the input 
and output arcs. Formally, $T$ is of the following form,
where $\aaa_{\text{in,out}} = \aaa_{\text{in}}\cup\aaa_{\text{out}}$
and $\power{\ }$ is the power-set operator,
$\power{\aaa_{\text{in,out}}} = \Set{A\,|\,A\subseteq\aaa_{\text{in,out}}}$:%
       \footnote{The notation ``$\aaa_{\text{in,out}}$'' is ambiguous,
    because it does not distinguish between input arcs and 
    output arcs. We
    use it nonetheless for succintness. The context will always make
    clear which members of $\aaa_{\text{in,out}}$ are input arcs and
    which are output arcs. }
\begin{itemize}
\item[]
\( 
    T\;:\ \power{\aaa_{\text{in,out}}}\ \to\ \intervals{\reals} 
    \qquad\text{where\ \ }
       \intervals{\reals}\ :=
    \ \Bigl\{\,[r,r']\;\Bigl|
    \;r,r'\in\reals\text{ and } r\leqslant r'\,\Bigr\} ,
\) 
\end{itemize}
\ie, $\intervals{\reals}$ is the set of bounded closed intervals. 
As a function, $T$ is not totally arbitrary and satisfies 
conditions that make it a \emph{network typing}; in particular, it 
will always be that $T(\varnothing) = [0,0] = \Set{0} =
T(\aaa_{\text{in,out}})$, the latter condition expressing the fact
that the total amount entering a network must equal the total amount
exiting it.%
   \footnote{We assume there are no \emph{producer nodes}
   and no \emph{consumer nodes} in $\N$. In the presence of
   producers and consumers, our formulation here of
   flow networks and their typings has to be adjusted accordingly 
   (details in~\cite{kfouryCHAR:2012}).
   }
\Hide{
Henceforth, we use the term ``network'' to mean ``flow network'' in 
the sense just defined.
}

An \emph{input/output assignment} (or \emph{IO assignment}) is a
function $f : \aaa_{\text{in,out}}\to\nreals$. For a 
flow $f:\aaa\to\nreals$ or an IO assignment $f : \aaa_{\text{in,out}}\to\nreals$,
we say $f$ \emph{satisfies} the
typing $T$ iff, for every $A\in\power{\aaa_{\text{in,out}}}$ such
that $T(A)$ is defined and $T(A) = [r_1,r_2]$, we have:
\begin{itemize}
\item[]
\( 
   r_1\ \leqslant
   \ \sum\, f(A\cap\aaa_{\text{in}})\ -\ \sum\, f(A\cap\aaa_{\text{out}})
   \ \leqslant\ r_2
\) 
\end{itemize}
where $\sum f(X)$ means $\sum\Set{f(x)|x\in X}$. In words, this says
that the ``sum of the values assigned by $f$ to input arcs'' 
minus the ``sum of the values assigned by $f$ to output arcs''
is within the interval $[r_1,r_2]$.

\vspace{-.1in}
\section{Principal Typings}
\label{sect:valid-vs-principal}
\label{sect:sound-and-complete}

We say a typing $T$ is \emph{sound} for network $\N$ if:
\begin{itemize} 
\item Every IO assignment $f : \aaa_{\text{in,out}}\to\nreals$
      satisfying $T$ is extendable to
      a feasible flow $f' : \aaa\to\nreals$ in $\N$.
\end{itemize} 
A sound typing is one that is generally more conservative than
required to prevent system's malfunction: It filters out 
all unsafe IO assignments, \ie, not extendable to feasible flows, and
perhaps a few more that are safe.

For our application here (max-flow and min-flow values), not only do
we want to assemble networks for their safe operation, we want to
operate them to the \emph{limit} of their safety guarantees.  We 
therefore use the two limits of each interval/type to specify the
\emph{exact} minimum and the \emph{exact} maximum that an input/output arc (or a
subset of input/output arcs) can carry across interfaces. 
We thus say a typing $T$ is \emph{complete} for network $\N$ if:
\begin{itemize}
\item 
       Every feasible flow $f : \aaa\to\nreals$ in $\N$
       satisfies $T$.
\end{itemize}
Every min-flow in $\N$ and every
max-flow in $\N$ satisfy a sound and complete typing $T$ for $\N$.   

Let $\size{\aaa_{\text{in}}} = p \geqslant 1$ and
$\size{\aaa_{\text{out}}} = q \geqslant 1$, and assume a fixed
ordering of the arcs in $\aaa_{\text{in,out}}$. An IO
assignment $f : \aaa_{\text{in,out}}\to\nreals$ specifies a point,
namely $\Angles{f(a)\;|\;a\in\aaa_{\text{in,out}}}$, in the
vector space ${\reals}^{p+q}$, and the
collection of all IO assignments satisfying a typing $T$
form a \emph{compact convex polyhedral set} (or \emph{polytope}) in
the first orthant $(\nreals)^{p+q}$, which we denote $\poly{T}$.
Using standard notions of convexity in vector spaces ${\reals}^n$
and polyhedral analysis~\cite{schrijver12,boyd-vanderberghe09}, 
the following are straightforward:
\begin{proposition} 
[Sound and Complete Typings Are Equivalent]
  \label{fact:equivalence-of-principal-typings}
  If $T_1$ and $T_2$ are sound and complete typings for the same
      network $\N$, then  $\poly{T_1} = \poly{T_2}$.
\end{proposition} 
\begin{proposition} 
[Sound Typings Are Subtypings of Sound and Complete Typings]
  \label{fact:sound-typings-are-subtypes}
  If $T_1$ is a sound and complete typing for network $\N$ and $T_2$ is
     a sound typing for the same $\N$, then 
     $\poly{T_1} \supseteq \poly{T_2}$.  
\end{proposition} 
\noindent
The ``subtyping'' relation is contravariant w.r.t.
``$\subseteq$''. We say two networks $\N_1$ and $\N_2$ 
are \emph{similar} if they have the same number $p$ of input arcs
and same number $q$ of output arcs.
Proposition~\ref{fact:sound-typings-are-subtypes}
implies this: Given similar networks $\N_1$ and $\N_2$, with
respective sound and complete typings $T_1$ and $T_2$, if 
$\poly{T_1} \supseteq \poly{T_2}$, \ie, 
\emph{if $T_1$ is a subtyping of $T_2$,
then $\N_1$ can be safely substituted for $\N_2$, in any assembly of
networks containing $\N_2$.}%
  \footnote{\label{foot:angelic-vs-demonic}%
   An assignment of values to the input arcs (resp. output
   arcs) of a network does not uniquely determine the values at its
   output arcs (resp. input arcs) in the presence of multiple
   input/output arcs. In that sense, flow moves
   \emph{non-deterministically} between input arcs and output arcs.
   Non-determinism is usually classified in two ways: \emph{angelic}
   and \emph{demonic}, according to whether it proceeds to favor
   a desirable outcome or to obstruct it. Our notion of subtyping here,
   and with it the notion of \emph{safe substitution}, presumes that
   the non-determinism of flow networks is angelic. For the case when
   non-determinism is demonic, ``subtyping between network typings'' has
   to be defined in a more restrictive way. We elaborate
   on this question in 
   Appendix~\ref{sect:appendix:future}.}

One complication when dealing with typings as polytopes
are the alternatives in representing them (convex hulls \emph{vs.} 
intersections of halfspaces).
We choose to represent them by intersecting
halfspaces, with some (not all) redundancies in their defining linear
inequalities eliminated. 
We thus say the typing $T$
is \emph{tight} if, for every  $A\subseteq\aaa_{\text{in,out}}$ for
which $T(A)$ is defined and every $r\in T(A)$, there is an
IO assignment $f\in\poly{T}$ such that:
\begin{itemize}
\item[]
\(
 r = \sum f(A\cap\aaa_{\text{in}}) - \sum f(A\cap\aaa_{\text{out}}) .
\)
\end{itemize}
Informally, $T$ is tight if no defined $T(A)$ contains redundant
information. 

Another kind of redundancy occurs when an interval/type $T(A)$ is
defined for some $A = B\cup B'\subseteq\aaa_{\text{in,out}}$ 
with $B\neq\varnothing\neq B'$ even though
there is no communication between $B$ and $B'$. We eliminate
this kind of redundancy via what we call ``locally total'' typings.
We need a preliminary notion.
A network $\M = (\mm,\bbb)$ is a \emph{subnetwork} of
network $\N = (\nn,\aaa)$ if $\mm\subseteq\nn$ and $\bbb\subseteq\aaa$
such that:
\begin{alignat*}{3}
   &\bbb_\text{\#} &&= 
   &&\Set{\,a\in\aaa\;|\;\head{a}\in\mm\text{ \& }\tail{a}\in\mm\,} ,
\\
   &\bbb_\text{in} &&= 
   &&\Set{\,a\in\aaa\;|\;\head{a}\in\mm\text{ \& }\tail{a}\not\in\mm\,} ,
\\
   &\bbb_\text{out} &&= 
   &&\Set{\,a\in\aaa\;|\;\head{a}\not\in\mm\text{ \& }\tail{a}\in\mm\,} .
\end{alignat*}
We also say $\M$ is the subnetwork of $\N$ \emph{induced
by $\mm$}. The subnetwork $\M$ is a \emph{component} of $\N$ 
if $\M$ is connected and $\bbb_{\text{\#}}\subseteq\aaa_{\text{\#}}$,
$\bbb_{\text{in}}\subseteq\aaa_{\text{in}}$, and
$\bbb_{\text{out}}\subseteq\aaa_{\text{out}}$, \ie, $\M$ is a maximal
connected subnetwork of $\N$. If network $\N$ contains two distinct
component $\M$ and $\M'$, there is no communication between $\M$ and
$\M'$, and the typings of the latter two can be computed independently
of each other. We say a typing $T$ for $\N$ is \emph{locally total} if, 
for all components $\M = (\mm,\bbb)$ and $\M' = (\mm',\bbb')$ of $\N$, 
and all $B\subseteq\bbb_{\text{in,out}}$ and $B'\subseteq\bbb'_{\text{in,out}}$:
\begin{itemize}[itemsep=2pt,parsep=2pt,topsep=5pt,partopsep=0pt] 
\item The interval/type $T(B)$ is defined.
\item 
   If $\M'\neq \M$ and $B\neq\varnothing\neq B'$, 
      the interval/type $T(B\cup B')$ is not defined.
\end{itemize}
Whereas ``tight'' and ``locally total'' can be viewed (and are in
fact) properties of a typing $T$, independent of any network $\N$ for
which $T$ is a typing, ``sound'' and ``complete'' are properties of
$T$ relative to a particular $\N$.  If $\N$ has only one component
(itself), a locally-total typing for $\N$ is a total function on
$\power{\aaa_{\text{in,out}}}$. We can prove:
\begin{theorem}
      [Uniqueness of Locally Total, Tight, Sound and Complete Typings]
\label{fact:existence-and-uniqueness}
      For all networks $\N$, there is a unique
      typing $T$ which is locally total, tight, sound and complete --
      henceforth called the \emph{principal typing} of $\N$.
\end{theorem}
\noindent
The principal typing of $\N$ is a characterization of \emph{all} IO
assignments extendable to feasible flows in $\N$.%
    \footnote{A related result is 
    established in~\cite{Hagerup1998} by a different method that
    invokes a linear-programming procedure. The motivation for that 
    latter work, different from ours, is whether the external
    (\ie, input/output) flow pattern of a multiterminal network $\N$
    can be completely described independently of the size of $\N$. }
In particular, if $\N$ is connected, its min-flow and max-flow values
are the two limits of the type $T(\aaa_{\text{in}})$, or
equivalently, the negated two limits of $T(\aaa_{\text{out}})$.  
Theorem~\ref{fact:existence-and-uniqueness} implies that two
similar networks $\N_1$ and $\N_2$ are equivalent iff their 
principal typings are equal, \emph{regardless of their respective sizes and
internal details}.

We can compute a principal typing typing $T$ via linear-programming
(but we do not): For every component $\M = (\mm,\bbb)$ of $\N$ and
every $B\subseteq\bbb_{\text{in,out}}$, we specify an objective
$\theta_B$ to be minimized and maximized, corresponding to the two
limits of the type $T(B)$, relative to the collection $\CC$ of
flow-preservation equations (one for each node) and
capacity-constraint inequalities (two for each arc). Following this
approach in the proof of Theorem~\ref{fact:existence-and-uniqueness},
it is relatively easy to show that the resulting $T$ is locally total,
tight, and complete -- but it takes non-trivial work in polyhedral
analysis to prove that $T$ is also sound. 
Besides being relatively expensive (the result of invoking a
linear-programming procedure), the drawback of this approach
is that it is \emph{whole-network}, as opposed to \emph{compositional},
requiring prior knowledge of all constraints in $\CC$ before
$T$ can be computed.

\vspace{-.1in}
\section{Disassembling and Reassembling Networks}
\label{sect:disassemble-and-reassemble}

In earlier reports~\cite{BestavrosKfouryLapetsOcean:crts09,
  BestavrosKfouryLapetsOcean:hscc10,
  BestKfoury:dsl11,Kfoury:sblp11,kfouryDSL:2011},
  we used our compositional
  approach to model/design/analyze systems incrementally, from components
  that are supplied separately at different times.  Here, we assume
  we are given all of a network $\N$ at once, which we then disassemble
  into its smallest units (\ie, one-node components), compute their 
  principal typings, and then combine the latter to produce a principal
  typing for $\N$. Because we are given all of $\N$ at once,
  we can control the order in which we reassemble it. A schematic example is 
  Figure~\ref{fig:disassemble-and-reassemble}.

\begin{figure}[h] 
%
\begin{center}
\includegraphics[scale=.56,trim=0cm 15.72cm 0cm 1cm,clip]%
    {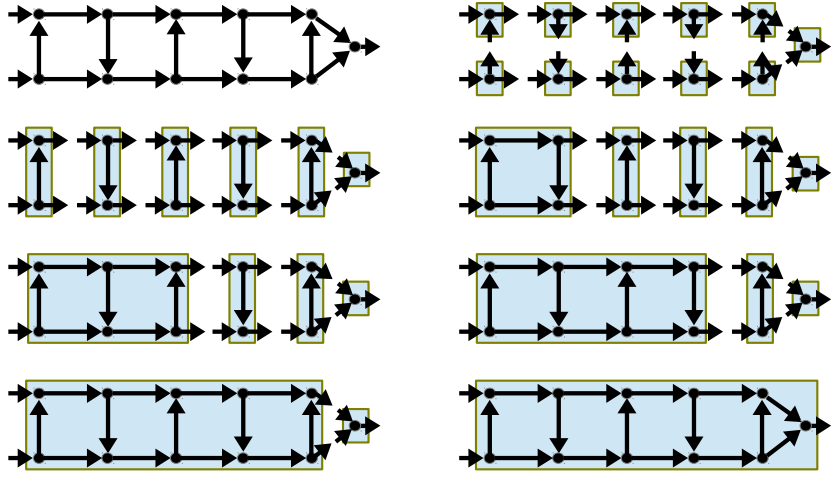}
\end{center}
\caption{Disassembling a network with external dimension $3$
         (\emph{left} top row: initially given $\N$,
         \emph{right} top row: broken-up into one-node components as $\N_0$),
         then reassembling it in stages so as to minimize
         the external dimension (here $3$ or $4$) of %
         each component (in the shaded areas) in each of the intermediate
         networks.}
\label{fig:disassemble-and-reassemble}
\end{figure}

The process of disassembling $\N = (\nn,\aaa)$ involves ``cutting in
halves'' some of its internal arcs: If internal arc
$a\in\aaa_{\text{\#}}$ is cut in two halves, then we remove $a$ and
introduce a new input arc $a^{+}$ with $\head{a^{+}} = \head{a}$ and a
new output arc $a^{-}$ with $\tail{a^{-}} = \tail{a}$.  Formally,
given a two-part partition of $X\uplus Y = \aaa_{\text{\#}}$, we
define:
\Hide{
   \[  \aaa_{\text{\#}}^{(+)} :=
    \, \Set{\, a^{+}\;|\; a \in X\,}\ \text{(the new input arcs)},
    \ \aaa_{\text{\#}}^{(-)} :=
    \, \Set{\, a^{-}\;|\; a \in X\,}\ \text{(the new output arcs)},
    \ \aaa_{\text{\#}}^{(\pm)} :=\, Y  .
   \]
}
\begin{itemize}
\item[]
\(  \aaa_{\text{\#}}^{(+)} :=
    \, \Set{\, a^{+}\;|\; a \in X\,}\ \text{(the new input arcs)},
    \ \aaa_{\text{\#}}^{(-)} :=
    \, \Set{\, a^{-}\;|\; a \in X\,}\ \text{(the new output arcs)},
    \ \aaa_{\text{\#}}^{(\pm)} :=\, Y  .
\)
\end{itemize}
$Y$ is the set of internal arcs that are not cut or that have had 
their two halves reconnected.

Given the initial network $\N = (\nn,\aaa)$ where every internal arc
is connected, we define another network $\Break{\N}$ where every
internal arc $a\in\aaa_{\text{\#}}$ is cut into two halves $a^{+}$ and
$a^{-}$. The input arcs and output arcs of $\Break{\N}$ are therefore:
$\aaa_{\text{in}}\cup\aaa_{\text{\#}}^{(+)}$ and
$\aaa_{\text{out}}\cup\aaa_{\text{\#}}^{(-)}$, respectively, with
$\aaa_{\text{\#}}^{(\pm)} = \varnothing$. $\Break{\N}$ is
 a network of $n = \size{\nn} \geqslant 1$ one-node components. 
If we call $\N_0$ the fully disassembled network:
\begin{itemize}
\item[] 
\( \N_0 = \Break{\N} = \Set{\M_1, \M_2, \ldots, \M_n} \)
\end{itemize}
we reassemble the original $\N$ by defining the sequence
of networks: $\N_0, \N_1,\ldots,\N_m$ where $\N_m = \N$ and $\N_{i+1}
= \bind{a}{\N_i}$ for every $0\leqslant i < m
= \size{\aaa_{\text{\#}}}$, where $\Bind$ is the operation that
splices $a^{+}$ and $a^{-}$.  What we call the
\emph{binding schedule} $\sigma$ is the order in which
the internal arcs are spliced, \ie, if:
\begin{itemize}
\item[] 
\(
   \N_{1} = \bind{b_1}{\N_0},\quad
   \N_{2} = \bind{b_2}{\N_1},\quad
   \ldots\quad, \quad
   \N_{m} = \bind{b_m}{\N_{m-1}}
\)
\end{itemize}
then $\sigma = b_1 b_2\cdots b_m$, where $\Set{b_1,\ldots,b_m}
= \aaa_{\text{\#}}$. If $\M$ is a connected network, we write
$\exDim{\M}$ for the \emph{external dimension} of $\M$.  For each of
the intermediate networks $\N_i$ with $0\leqslant i \leqslant m$,
we define:
\begin{itemize}
\item[] 
\(
   \Index{\N_{i}} := \max\,\Set{\exDim{\M}\;|\;\M\text{ is a 
                 component of $\N_i$}\,} 
\)
\end{itemize}
and also $\Index{\sigma} := 
\max\,\Set{\Index{\N_{0}},\Index{\N_{1}},\ldots,\Index{\N_{n}}}$.
Without invoking a linear-programming procedure and only using the
arithmetical operators $\Set{\max,\min, +, -}$ on arc capacities,
we can prove:
\begin{theorem}
  [Principal Typing of $\N$ from the Principal Typing of 
   $\Break{\N}$ in Stages]
\label{fact:principal-typing-build-up}\ \ 
\begin{enumerate}[itemsep=2pt,parsep=2pt,topsep=5pt,partopsep=0pt] 
\item
   We can compute the principal typing of a one-node network $\M$ 
   in time $\bigOO{2^d}$ where \emph{$d = \exDim{\M}$}. 
\item
   We can compute $\N_{i+1}$'s principal typing from $\N_{i}$'s
   principal typing in time $2^{\bigOO{d}}$ where \emph{$d = \Index{\N_i}$}. 
\item
   Let $\N$ be reassembled from $\Break{\N}$ using the binding schedule
   $\sigma$. We can compute the principal typing of $\N$ in time 
   $m\cdot 2^{\bigOO{\Index{\sigma}}}$ where $m = \size{\aaa_{\text{\#}}}$.
\end{enumerate}
\end{theorem}
\noindent
Part 3 in Theorem~\ref{fact:principal-typing-build-up} follows
from parts 1 and 2. It implies that,
to minimize the time to compute a principal typing for network
$\N$, we need to 
minimize $\Index{\sigma}$. There are natural network topologies for which
there is a binding schedule $\sigma$ with constant or slow-growing 
$\Index{\sigma}$ as a function of $m$ and $n$. 
Section~\ref{sect:max+min-flows-in-planar-networks} is an example.

\vspace{-.1in}
\section{An Application: Max-Flows and Min-Flows in Planar Networks}
   \label{sect:max+min-flows-in-planar-networks}

\Hide{
In recent years, there has been a flurry of new algorithms for
computing max flows in graphs. They consider
a wide range of restrictions and generalizations
on graphs (directed or undirected~\cite{Italiano:2011:IAM:1993636.1993679}, 
general or planar~\cite{DBLP:journals/corr/abs-1210-4811},
single-commodity or multi-commodity~\cite{DBLP:journals/corr/abs-1304-2338},
single source-sink~\cite{borradaileKlein2009,10.1109/FOCS.2012.66} or
multiple sources-sinks~\cite{borradaileKlein2011}) 
and rely on a large assortment
of algorithmic techniques~\cite{Orlin2013,borradaileKlein2011,LeeRS13}.
Some produce \emph{exact} max flows, others 
\emph{approximate} max flows~\cite{DBLP:journals/corr/abs-1304-2077,%
DBLP:journals/corr/abs-1304-2338}. 
They all achieve nearly
linear time in the size of graphs -- 
but not quite linear, unless arc capacities obey 
restrictions~\cite{EisenstatK13,LeeRS13}. }
There is a wide range of algorithms for the max-flow problem.
Some produce \emph{exact} max-flows, others 
\emph{approximate} max-flows~\cite{DBLP:journals/corr/abs-1304-2077,%
DBLP:journals/corr/abs-1304-2338}. 
They all achieve nearly
linear time in the graph size -- 
but not quite linear, unless arc capacities obey 
restrictions~\cite{EisenstatK13,LeeRS13}. 
A recent result is the following: \emph{There exists an
algorithm that solves the max-flow problem with multiple sources
and sinks in an $n$-node directed planar graph in $\bigOO{n \log^3 n}$
time}~\cite{borradaileKlein2011}.

Our type-based compositional approach offers an alternative,
with other benefits unrelated to algorithm run-time.  
Assume $\N$ is given with a planar embedding already, \ie, we do
not have to compute the embedding (which can be computed in linear
time in any case~\cite{patrignani2013}). 
Every planar embedding of an undirected
graph has an \emph{outerplanarity} $\geqslant 1$,
and so does therefore the network $\N$ by considering its underlying
graph where all arcs directions are ignored.%
    \footnote{Also ignoring two parallel arcs resulting from 
              omitting arc directions in two-node cycles.}
The planar embedding of $\N$ has \emph{outerplanarity} $k\geqslant 1$
if deleting all the nodes on the unbounded face leaves an
embedding of outerplanarity $(k-1)$.%
    \footnote{The ``outerplanarity index'' and the ``outerplanarity''
       of an undirected graph $G$ are not the same thing.
       The \emph{outerplanarity index} of $G$ is the smallest integer
       $k\geqslant 1$ such that $G$ has a planar embedding of
       outerplanarity $= k$.  To compute the outerplanarity index of
       $G$, and produce a planar embedding of $G$ of
       outerplanarity $=$ its outerplanarity index, is not a
       trivial problem, for which the best known algorithm
       requires quadratic time $\bigOO{n^2}$ in
       general~\cite{Kammer2007} -- which, if we used it to
       pre-process the input to our algorithm in
       Theorem~\ref{fact:principal-typing-of-planar}, would upend its
       final linear run-time.  On the other hand, a planar embedding of $G$
       of outerplanarity = $4$-approximation of its outerplanarity index can
       be found in linear time, also shown in~\cite{Kammer2007}.
       Although an interesting problem, we do not bother with finding
       a planar embedding of $\N$ of outerplanarity matching its
       outerplanarity index. It is worth noting that for a tri-connected planar
       network, ``outerplanarity'' and ``outerplanarity index''
       are the same measure, since the planar embedding of a tri-connected
       graph is unique (\cite{dattaPrakriya2011} or
       Section 4.3 in~\cite{diestel2012}).  
       However, there are very simple examples of bi-connected, 
       but not tri-connected, planar networks with planar
       embeddings of arbitrarily large outerplanarity but whose
       outerplanarity index is as small as $1$. }
An example of a $1$-outerplanar embedding is the network
shown in Figure~\ref{fig:disassemble-and-reassemble}.
Based on Theorem~\ref{fact:principal-typing-build-up}, we can prove:

\begin{theorem}[Principal Typing of Planar Network]
  \label{fact:principal-typing-of-planar} 
     If an $n$-node network $\N$ is given in a planar embedding
     of outerplanarity $k\geqslant 1$, with $p\geqslant 1$ input ports
     and $q\geqslant 1$ output ports, we can compute a principal typing for
     $\N$ in time $\bigOO{n}$ where the hidden multiplicative
     constant depends only on $k$, $p$ and $q$.
\end{theorem}
\noindent
If the algorithm in Theorem~\ref{fact:principal-typing-of-planar} is
made to work on the planar embedding of a $3$-regular network $\N$
(\eg, the network in Figure~\ref{fig:disassemble-and-reassemble}), it
proceeds by disassembling and reassembling $\N$ in a manner to
minimize the interface dimension of all components in intermediate
stages (as in Figure~\ref{fig:disassemble-and-reassemble}). The proof
of Theorem~\ref{fact:principal-typing-of-planar} is a simple
generalization of this operation to the $3$-regular planar embedding
of an arbitrary planar $\N$, based on: 
\begin{lemma}
[From Arbitrary Networks to 3-Regular Networks]
\label{lem:from-arbitrary-to-3-regular}
Let $\N = (\nn,\aaa)$ be a flow network, not necessarily planar.
In time $\bigOO{n}$, we can transform $\N$ into a similar network
$\N' = (\nn',\aaa')$ such that:
\begin{enumerate}[itemsep=0pt,parsep=2pt,topsep=5pt,partopsep=0pt] 
\item There are no two-node cycles in $\N'$.
\item The degree of every node in $\N'$ is $3$.
\item Every typing 
      $T:\power{\aaa_{\text{\rm in,out}}}\to\intervals{\reals}$
      is principal for $\N$ iff $T$ is principal for $\N'$.
\item $\size{\nn'} \leqslant 2m$\ and
    \ $\size{\aaa'} \leqslant 3m$,
    \ where $m = \size{\aaa}$. 
\item If $\N$ is given in a $k$-outerplanar embedding,
      $\N'$ is returned in a $k'$-outerplanar embedding with  
      $k'\leqslant 2k$.
\end{enumerate}
\end{lemma}
\noindent
The proofs for parts 1-4 in Lemma~\ref{lem:from-arbitrary-to-3-regular}
are relatively straightforward, only the proof for part 5 is
complicated.  An immediate consequence of
Theorem~\ref{fact:principal-typing-of-planar} is: \emph{For every $k,
p, q\geqslant 1$, we have an algorithm which, given an arbitrary
$n$-node network $\N$ in a planar embedding of outerplanarity $\leqslant k$ and
external dimension $\leqslant p+q$, simultaneously computes a max-flow
value and a min-flow value in time $\bigOO{n}$.}

\Hide{
This result, and Fact~\ref{fact:principal-typing-of-planar}
which implies it, appear in our
recent technical report~\cite{kfouryMirzaei:2013} and are yet to be
submitted for publication.
}

\vspace{-.1in}
\section{Extensions and Future Work}
   \label{sect:extensions-and-future}   

There are natural generalizations that will provide
material for a more substantial comparison with other
approaches. Among such generalizations:
\begin{enumerate}[itemsep=0pt,parsep=2pt,topsep=0pt,partopsep=0pt] 
\item Adjust the formal framework 
      to handle the commonly-considered cases of: 
      \begin{itemize}[itemsep=0pt,parsep=0pt,topsep=0pt,partopsep=0pt]  
      \item \emph{multicommodity flows}
      (formal definitions in~\cite{1993orlin}, Chapt. 17)
      \item \emph{minimum-cost flows}, \emph{minimum-cost max flows}, and
      variations (formal definitions in~\cite{1993orlin}, Chapt. 9-11)
      \end{itemize}
      These cases introduce new kinds of linear constraints, often more general
      than \emph{flow-conservation} equations and 
      \emph{capacity-constraint} inequalities (as in this report),
      where all coefficients are $+1$ or $-1$.
\item Identify natural \emph{network topologies} that are amenable to the
      kind of examination we already applied to \emph{planar networks}, 
      following the presentation in 
      Section~\ref{sect:max+min-flows-in-planar-networks}
      and leading to similar results.
\end{enumerate}
Beyond commonly-studied generalizations and topologies, there are a 
number of questions more directly related to the fine-tuning of our
approach and/or to the modeling and analysis of networking systems.

Among such questions  
are practical situations where flows are regulated
by \emph{non-linear} constraints. For example, a common
case is that of a non-linear \emph{convex cost function}
which may or may not be transformed into a \emph{piecewise linear
cost function} (Chapt. 14 in~\cite{1993orlin}). 
Another example is provided by \emph{mass conservation},
more general than \emph{flow conservation}: If $\delta(a)$
denotes the ``density'' of the flow carried by arc/channel
$a\in\aaa$ and $v(a)$ the ``velocity'' at which it travels along
$a$, then \emph{mass conservation} at node $\nu\in\nn$ is expressed
as the non-linear constraint:
  \(
   \sum\Set{\delta(a)\cdot v(a)\,|\,\text{arc $a$ enters node $\nu$}}
     = \sum\Set{\,\delta(b)\cdot v(b)\;|\;\text{arc $b$ exits node $\nu$}\,} ,
  \)
which is equivalent to flow conservation at node $\nu$
only when velocity $v$ is uniformly the same on all arcs.%
    \footnote{Velocity is measured in \emph{unit distance}/\emph{unit time},
    \eg, \emph{mile}/\emph{hour}. Density is measured in
    \emph{unit mass}/\emph{unit distance}, \eg,
    \emph{ton}/\emph{mile}. Hence, the value of $d(a)\cdot v(a)$ on
    arc $a$ is measured in \emph{unit mass}/\emph{unit time}, \eg,
    \emph{ton}/\emph{hour}, commonly called the \emph{mass flow} on $a$.
    }

We leave all of the preceding questions for future research.
Below we mention three specific directions under current investigation.

\paragraph{Leaner Representations of Principal Typings.}

We have not eliminated all redundancies in our representation of
principal typings. The presence of these redundancies increases the
run-time of our algorithms, as well as complicates the process of
inferring typings and reasoning about them.  For example, by
Lemma~\ref{lem:necessity} in Appendix~\ref{sect:appendix:disassemble},
if $A\uplus B = \aaa_{\text{\rm in,out}}$ is a two-part partition of
$\aaa_{\text{\rm in,out}}$ and 
$T:\power{\aaa_{\text{\rm in,out}}}\to\intervals{\reals}$ is the
principal typing (as defined in this report) for a network $\N$ with
input/output arcs $\aaa_{\text{\rm in,out}}$, then $T(A) = -T(B)$,
which means that at most one of the two intervals/types, $T(A)$ and
$T(B)$, is necessary for defining $\poly{T}$.

For two concrete examples, consider typings $T_1$ and $T_2$ in
Examples~\ref{ex:simplest-1} and~\ref{ex:simplest-2} in 
Appendix~\ref{sect:appendix:basic}. These are principal typings.
In addition to the interval/type assigned to $\varnothing$ and
$\Set{a_1,a_2,a_3,a_4}$, which is always $[0,0]$, the remaining
$14$ intervals/types are ``symmetric'', in the sense that
$T(A) = -T(B)$ whenever $A\uplus B = \Set{a_1,a_2,a_3,a_4}$.
Hence, $7$ type assignments will suffice to uniquely define
$\poly{T_1}$ and $\poly{T_2}$. 

This is a general fact: For a network $\N$ of external dimension
$\exDim{\N} = 4$, at most $7$ non-zero interval/type assignments are required
to uniquely define a tight, sound, and complete, typing for $\N$. 
Hence, we need a weaker requirement than ``locally total''
to guarantee uniqueness of a new notion of ``principal typing''.

\vspace{-.1in}
\paragraph{Augmenting Typings for Resource Management.}

In recent years, programming-language theorists have augmented
types and typings to enforce more than safety invariants. In particular,
there are strongly-typed functional languages (mostly experimental now)
where types encode information related to \emph{resource management}
and \emph{security guarantees}. The same can be done with our network typings
encoded as polytopes.

For example, researchers in traffic engineering consider
\emph{objective functions} (to be optimized relative to such
constraints as ``flow conservation'', ``capacity constraints'', and
others), which also keep track of uses of resources (\eg,
see~\cite{BalonLeduc:icon2006}).  Possible measurements of resources
are -- let $\lc(a) = 0$ for all arcs $a$ for simplicity:
\begin{enumerate}[itemsep=0pt,parsep=2pt,topsep=2pt,partopsep=0pt] 
\item \emph{Hop Routing} (HR). The hop-routing value of $f$
     is the number of channels $a\in\aaa$ such that $f(a)\neq 0$.
\item \emph{Channel Utilization} (CU). The utilization
    of a channel $a$ is defined as $u(a) = f(a)/\uc(a)$.
\item \emph{Mean Delay} (MD). The mean delay
    of a channel $a$ can be measured by $d(a) = 1/(\uc(a)-f(a))$.
\end{enumerate}
  HR, CU, and MD, can be taken as objectives to be minimized, along
  with flow to be maximized, resulting in a more complex optimization
  problem. But HR, CU, and MD, can also be taken 
  as measures to be composed across interfaces, requiring a different
  (and simpler) adjustment of our network typings.

\vspace{-.1in}
\paragraph{Sensitivity (aka Robustness) Analysis.}

A key concept in many research areas is \emph{function robustness} or,
by another name, \emph{function sensitivity}.%
     \footnote{For example, in \emph{differential privacy}, 
               one of the most common mechanisms for turning a
               (possibly privacy-leaking) query into a differentially 
                private one involves establishing a bound on its 
                sensitivity.}
The concept appears in programming-language studies,
which more directly informed our own
work~\cite{Chaudhuri:2012:CRP:2240236.2240262,
        Chaudhuri:2011:PPR:2025113.2025131}: 
A program is said to be \emph{$K$-robust} if an $\varepsilon$-variation of
the input can cause the output to vary by at most $\pm K\varepsilon$.
More recently, a type-based approach to robustness analysis of
functional programs was introduced and shown to offer additional
benefits~\cite{D'Antoni:2013:SAU:2505351.2505353}.

The counterpart of function robustness in network-flow problems
is \emph{sensitivity analysis}, which has a longer
history~(Chapt. 9-11 in~\cite{1993orlin},
Chapt. 3 in \cite{1977bradley}). The purpose
is to determine changes in the optimal solution of a 
flow problem resulting from small changes in the data (supply/demand vector
or the capacity or cost of any arc). There are
two basic different ways of performing sensitivity analysis in relation
to flow problems (Chapt. 9 in~\cite{1993orlin}): 
(1) using combinatorial methods and (2) using simplex-based methods 
from linear programming -- and both are essentially \emph{whole-system}
approaches. 

A natural outgrowth from these earlier studies will be to adapt them to our
particular type-based \emph{compositional} approach.  In particular,
in our case, robustness analysis should account for the effects of,
not only small changes in network parameters, but also complete
break-down of an arc/channel or a subnetwork -- and, preferably, in
such a way as to not obstruct efficient inference of network typings.

\Hide{

 The fundamental question is how to measure
robustness, parametrize typings with this measure, and compose it
across network interfaces.

}

\ifTR
  \clearpage
\else
\fi

{\footnotesize 
   \bibliographystyle{plain}
\bibliography{./generic,./extra}
}

\ifTR
  \clearpage
\else
\fi

\appendix
\section{Appendix: Further Comments and Examples  
         for \hyperref[sect:flow-networks]{Section~\ref*{sect:flow-networks}}
         and \hyperref[sect:valid-vs-principal]%
             {Section~\ref*{sect:valid-vs-principal}}}
  \label{sect:appendix:basic}

Our notion of a network ``typing'' as an assignment of intervals/types
to members of a powerset resembles in some ways, but is still
different from, the notion of a ``typing'' (different from a ``type'')
in the study of strongly-typed programming languages. This is quite
apart from the differences in syntactic conventions -- the first from
vector spaces and polyhedral analysis, the second from first-order
logic. In strongly-typed programming languages, a ``typing'' refers to
the result of what is called a ``derivable typing judgment'' and
consists of: a program expression $M$; a type (not a typing) assigned
to $M$ and, at least implicitly, to every well-formed subexpression of
$M$; \emph{and} a type environment that includes a type for every
variable occurring free in $M$.

Report~\cite{Jim:POPL-1996} and the longer~\cite{Jim:MIT-LCS-1995-532}
are at the origin of this notion of ``typing'' in programming
languages.  These two reports also discuss the distinction between
``modular'' and ``compositional'', in the same sense we explained in
Section~\ref{sect:intro}, but now in the context of type inference for
strongly-typed functional programs. The notion of a ``typing'' for a
program, as opposed to a ``type'' for it, came about as a result of
the need to develop a ``compositional'', as opposed to a just
``modular'', static analysis of programs.

\begin{Proof}{for Theorem~\ref{fact:existence-and-uniqueness}}
There are different approaches to proving
Theorem~\ref{fact:existence-and-uniqueness}. One approach relies
heavily on polyhedral analysis and invokes a linear-programming
procedure (such as the simplex) repeatedly. Moreover, it requires a
preliminary collection of all flow-preservation equations (one for
each node) and all capacity-constraint inequalities (two for each
arc), for the given network $\N$, before we can start an analysis to
compute the principal typing of $\N$. This is the approach taken
in~\cite{kfouryCHAR:2012}.

But we can do better here. The alternative is to simply invoke
Theorem~\ref{fact:principal-typing-build-up} later in this report.
This alternative approach does not use any pre-defined
linear-programming procedure, and is also in harmony with our emphasis
on ``compositionality'' -- allowing for partial analyses to be 
incrementally composed. \hfill\QED
\end{Proof}

\begin{terminology}
In all previous articles that informed the work reported in this 
report, we made a distinction between \emph{valid typings} and
\emph{principal typings} (for the same network $\N$). What we called
``valid'' before is what we call ``sound'' here, and what we called
``principal'' before is what we call ``sound and complete'' here.

What we call ``principal'' here is not only ``sound and complete'', but
also satifies uniqueness conditions -- in this report, expressed by the
notions of ``tight'' and ``locally total''.
\end{terminology}

\medskip

The three examples below illustrate several of the concepts in
Sections~\ref{sect:flow-networks} and~\ref{sect:valid-vs-principal}.
These are three similar networks, \ie, they each have $p=2$ input arcs
and $q=2$ output arcs. Their principal typings are generated using
the material in Section~\ref{sect:disassemble-and-reassemble}.

\begin{example}
\label{ex:simplest-1}
Network $\N_1$ is shown on the left in
Figure~\ref{fig:six-and-eight-node},  
%
where all omitted lower-bound
capacities are $0$  and all omitted upper-bound capacities are $K$.
$K$ is an unspecified ``very large number''.
A min-flow in $\N_1$ pushes $0$ units through,
and a max-flow in $\N_1$ pushes $30$ units. The value of every
feasible flow in $\N_3$ will therefore be in the interval $[0,30]$.
A principal typing $T_1$ for $\N_1$ is such that 
$T_1(\varnothing) = [0,0] =T_1(\Set{a_1,a_2,a_3,a_4})$ 
and makes the following type assignments:
\begin{alignat*}{5}
  & \framebox{$a_1:[0,15]$}\quad 
  && \fbox{$a_2:[0,25]$}\quad && \fbox{$-a_3:[-15,0]$}\quad 
  && \fbox{$-a_4:[-25,0]$}
\\
  & \fbox{$a_1+a_2:[0,30]$}\quad && \underline{a_1-a_3:[-10,10]}\quad 
  && \underline{a_1-a_4:[-25,15]}
\\
  & \underline{a_2-a_3:[-15,25]}\quad && \underline{a_2-a_4:[-10,10]}
  && \fbox{$-a_3-a_4:[-30,0]$}\quad 
\\
  & a_1+a_2-a_3: [0,25]\qquad&& a_1+a_2-a_4:[0,15] \qquad
   && a_1-a_3-a_4: [-25,0]\qquad && a_2-a_3-a_4: [-15,0]\quad
\end{alignat*}
To explain our notational convention, consider the type assignment
``$a_1+a_2-a_3: [0,25]$'', one of the $14$ non-trivial assignments made
by $T_1$. We write ``$a_1+a_2-a_3: [0,25]$'' to mean that 
$T_1(\Set{a_1,a_2,a_3}) = [0,25]$.
The minus preceding $a_3$ in the expression ``$a_1+a_2-a_3$'' 
indicates that $a_3$ is an output arc, whereas $a_1$ and $a_2$ are input arcs.
The boxed type assignments and the underlined type assignments
are for purposes of comparison with the typing $T_2$ in 
Example~\ref{ex:simplest-2} and typing $T_3$ in 
Example~\ref{ex:simplest-3}.
\end{example}

\begin{example}
\label{ex:simplest-2}
Network $\N_2$ is shown in the middle
in Figure~\ref{fig:six-and-eight-node}.
%
A min-flow in $\N_2$ pushes $0$ units through,
and a max-flow in $\N_2$ pushes $30$ units. The value of all
feasible flows in $\N_2$ will therefore be in the interval $[0,30]$,
the same as for $\N_1$ in Example~\ref{ex:simplest-1}. 
A principal typing $T_2$ for $\N_2$ is such that
$T_2(\varnothing) = [0,0] = T_2(\Set{a_1,a_2,a_3,a_4})$,
and makes the following type assignments:
\begin{alignat*}{5}
  & \fbox{$a_1:[0,15]$}\quad && \fbox{$a_2:[0,25]$}\quad 
  && \fbox{$-a_3:[-15,0]$}\quad && \fbox{$-a_4:[-25,0]$}
\\
  & \fbox{$a_1+a_2:[0,30]$}\quad && \underline{a_1-a_3:[-10,12]}\quad
  && \underline{a_1-a_4:[-23,15]}
\\
  & \underline{a_2-a_3:[-15,23]}\quad && \underline{a_2-a_4:[-12,10]}
  && \fbox{$-a_3-a_4:[-30,0]$}\quad 
\\
  & a_1+a_2-a_3: [0,25]\qquad&& a_1+a_2-a_4:[0,15] \qquad
   && a_1-a_3-a_4: [-25,0]\qquad && a_2-a_3-a_4: [-15,0]\quad
\end{alignat*}
In this example and the preceding one,
the type assignments in rectangular boxes are for subsets
of the input arcs $\Set{a_1,a_2}$ and for subsets of the output arcs
$\Set{a_3,a_4}$, but not for subsets mixing input arcs and output arcs.
Note that the boxed assignments are the same for $T_1$
and $T_2$.

The underlined type assignments are among those that mix input and
output arcs. We underline those in
Example~\ref{ex:simplest-1} that are different from the
corresponding ones in this
Example~\ref{ex:simplest-2}. This difference implies there
are IO assignments $f : \Set{a_1,a_2,a_3,a_4}\to\nreals$
extendable to feasible flows in $\N_1$ (resp. in $\N_2$) but not in 
$\N_2$ (resp. in $\N_1$). This is perhaps counter-intuitive, since
$T_1$ and $T_2$ make exactly the same type assignments to input
arcs and, separately, output arcs (the boxed assignments). For
example, the IO assignment $f$ defined by:
\[
    f(a_1) = 15\qquad f(a_2) = 0\qquad f(a_3) = 3\qquad f(a_4) = 12
\]
is extendable to a feasible flow in $\N_2$ but
not in $\N_1$. The reason is that $f(a_1) - f(a_3) = 12$ violates
(\ie, is outside) the type $T_1(\Set{a_1,a_3}) = [-10,10]$.
Similarly, the IO assignment $f$ defined by:
\[
    f(a_0) = 0\qquad f(a_2) = 25\qquad f(a_3) = 0\qquad f(a_4) = 25
\]
is extendable to a feasible flow in $\N_1$ but not in $\N_2$, the
reason being that $f(a_2) - f(a_3) = 25$ violates
the type $T_2(\Set{a_2,a_3}) = [-15,23]$.

From the preceding, neither $T_1$ nor $T_2$ is a
subtyping of the other, in the sense explained in
Section~\ref{sect:valid-vs-principal}. \emph{Neither $\N_1$ nor $\N_2$
can be safely substituted for the other in a larger assembly of networks}.
\end{example}

\begin{example}
\label{ex:simplest-3}
Network $\N_3$ is shown on the right in
Figure~\ref{fig:six-and-eight-node}.  
A principal typing $T_3$ for $\N_3$ is such that 
$T_3(\varnothing) = [0,0] =T_3(\Set{a_1,a_2,a_3,a_4})$ 
and makes the following type assignments:
\begin{alignat*}{5}
  & \framebox{$a_1:[0,15]$}\quad 
  && \fbox{$a_2:[0,25]$}\quad && \fbox{$-a_3:[-15,0]$}\quad 
  && \fbox{$-a_4:[-25,0]$}
\\
  & \fbox{$a_1+a_2:[0,30]$}\quad && \underline{a_1-a_3:[-10,10]}\quad 
  && \underline{a_1-a_4:[-23,15]}
\\
  & \underline{a_2-a_3:[-15,23]}\quad && \underline{a_2-a_4:[-10,10]}
  && \fbox{$-a_3-a_4:[-30,0]$}\quad 
\\
  & a_1+a_2-a_3: [0,25]\qquad&& a_1+a_2-a_4:[0,15] \qquad
   && a_1-a_3-a_4: [-25,0]\qquad && a_2-a_3-a_4: [-15,0]\quad
\end{alignat*}
In this example and the two preceding,
the type assignments in rectangular boxes are for subsets
of the input arcs $\Set{a_1,a_2}$ and for subsets of the output arcs
$\Set{a_3,a_4}$, but not for subsets mixing input arcs and output arcs.
Again, the boxed assignments are the same for $T_1$, $T_2$, and
$T_3$. The differences between $T_1$, $T_2$, and
$T_3$ are in the underlined type assignments.

We write $\subT{T_1}{T_3}$ and $\subT{T_2}{T_3}$ to indicate that 
$T_1$ and $T_2$ are \emph{subtypings} of $T_3$, contravariantly
corresponding to the fact that $\Poly{T_1} \supseteq \Poly{T_3}$
and $\Poly{T_2} \supseteq \Poly{T_3}$. In fact, $T_3$ is the least
typing (in the partial order ``$\subTsym$'') such that both $T_1$ and $T_2$
are subtypings of $T_3$, because $T_3(A) = T_1(A)\cap T_2(A)$ for 
every $A\subseteq \Set{a_1,a_2,a_3,a_4}$. To be specific,
the subsets $A$ on which $T_1$ and $T_2$ disagree are:
\[
     \Set{a_1,a_3},\ \ \Set{a_1,a_4},\ \ \Set{a_2,a_3},\ \ \Set{a_2,a_4},
\]
and intersecting the intervals/types assigned by $T_1$ and $T_2$
to these sets, we obtain the following equalities:
\begin{alignat*}{8}
  & T_1(\Set{a_1,a_3})\cap T_2(\Set{a_1,a_3})\ &&=\ \ &&T_3(\Set{a_1,a_3}),
  \qquad
  &&T_1(\Set{a_1,a_4})\cap T_2(\Set{a_1,a_4})\ &&=\ \ &&T_3(\Set{a_1,a_4}),
\\[1ex]
  & T_1(\Set{a_2,a_3})\cap T_2(\Set{a_2,a_3})\ &&=\ \ &&T_3(\Set{a_2,a_3}),
  \qquad
  &&T_1(\Set{a_2,a_4})\cap T_2(\Set{a_2,a_4})\ &&=\ \ &&T_3(\Set{a_2,a_4}),
\end{alignat*}
\emph{Both $\N_1$ and $\N_2$
can be safely substituted for $\N_3$ in a larger assembly of networks} --
provided the non-determinism of flow movement through $\N_1$ and $\N_2$ 
is \emph{angelic}.%
    \footnote{If the non-determinism of flow movement is \emph{demonic},
    as opposed to \emph{angelic}, 
    we need a more restrictive notion of subtyping.
    Further comments in footnote~\ref{foot:angelic-vs-demonic} and
    Appendix~\ref{sect:appendix:future}.}
\end{example}

\begin{figure}[h] 
\begin{center}
%
\begin{custommargins}{-2.5cm}{-2.5cm}
\begin{minipage}[b]{0.32\linewidth}
         \includegraphics[scale=.25,trim=0cm 12.50cm 1cm 1.5cm,clip]%
         {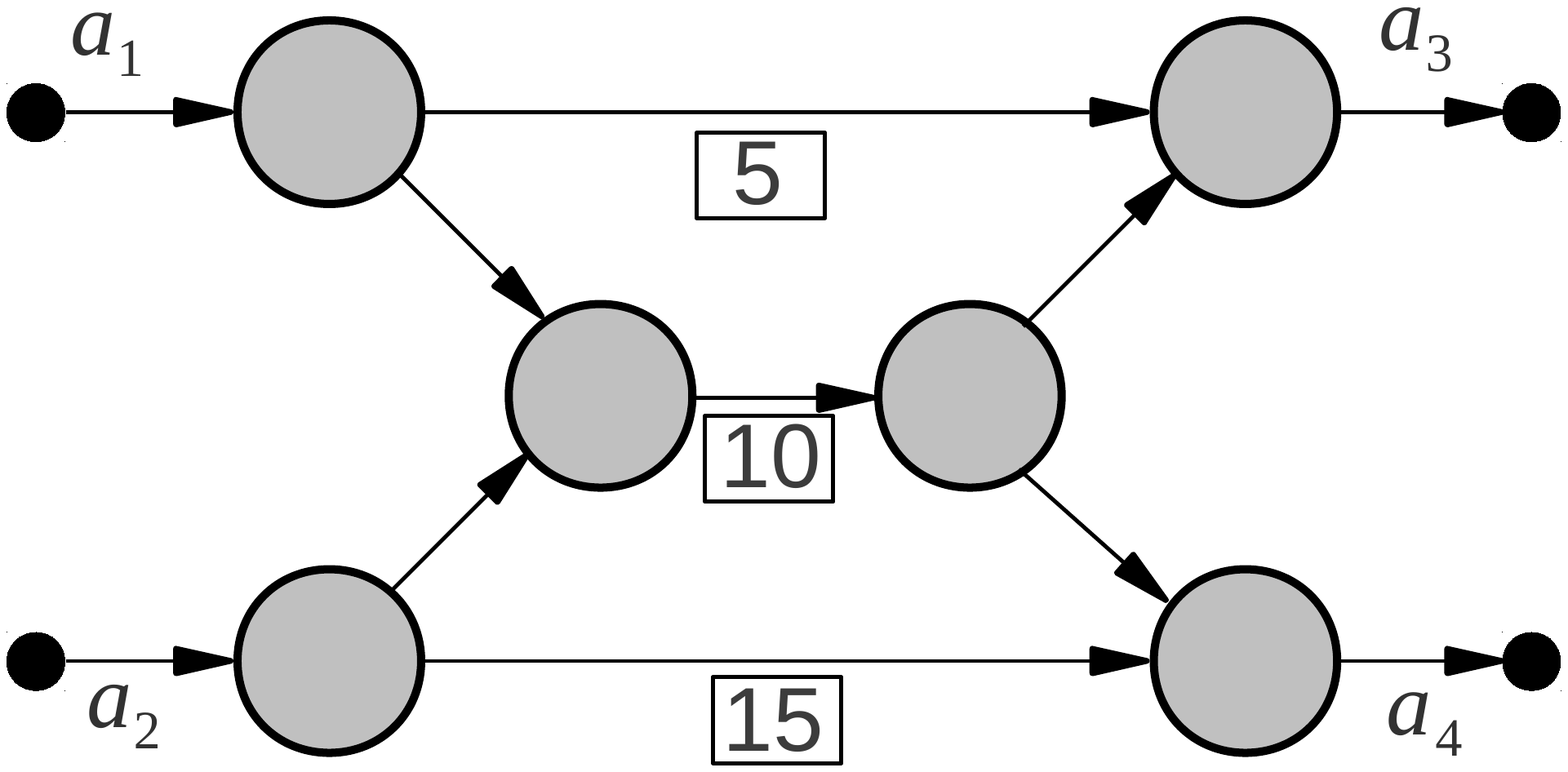}
\end{minipage}
\begin{minipage}[b]{0.32\linewidth}
         \includegraphics[scale=.25,trim=0cm 10.50cm 0cm 1.5cm,clip]%
         {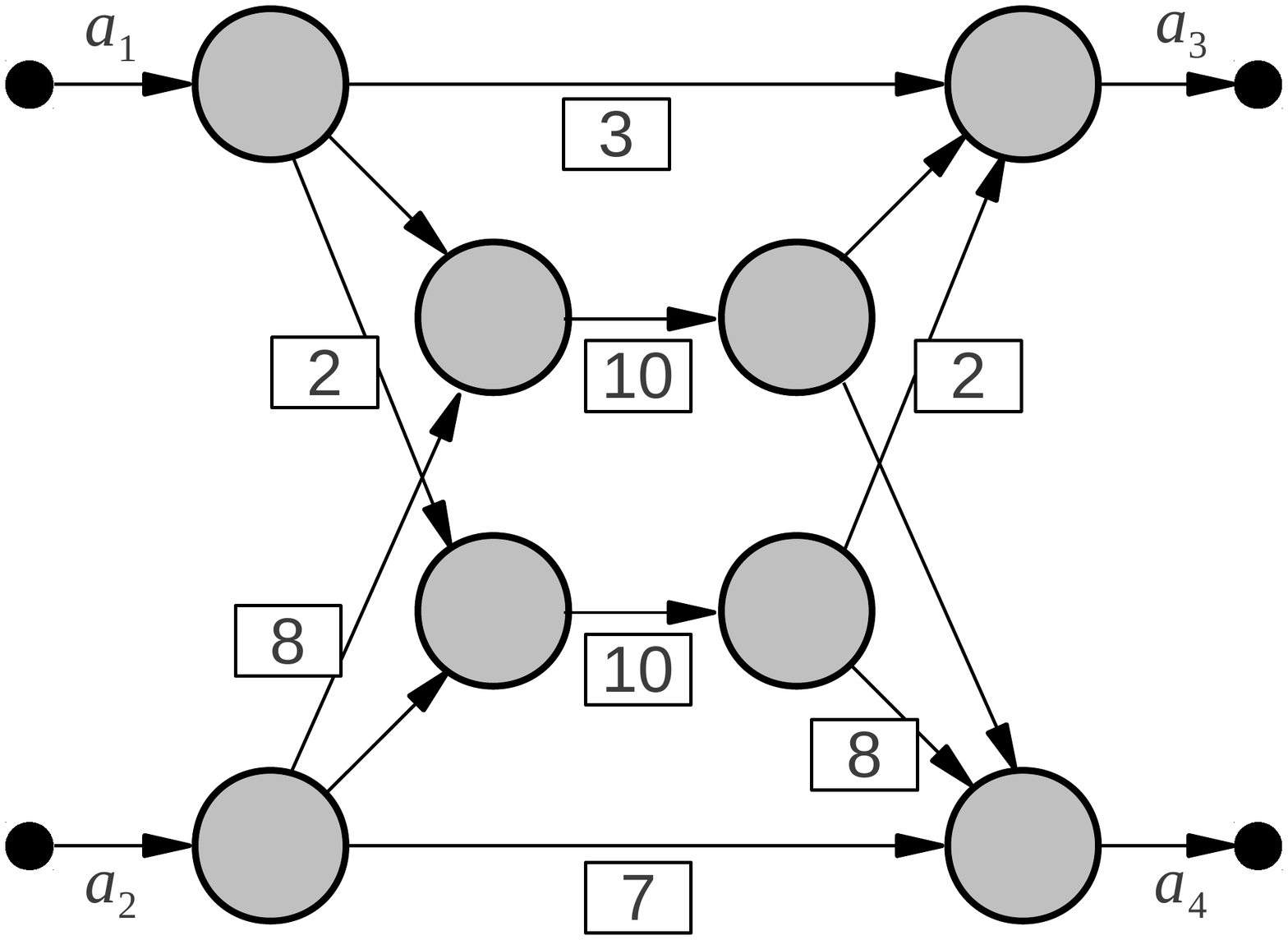}
\end{minipage}
\begin{minipage}[b]{0.32\linewidth}
         \includegraphics[scale=.25,trim=0cm 11.20cm 0cm 1.5cm,clip]%
         {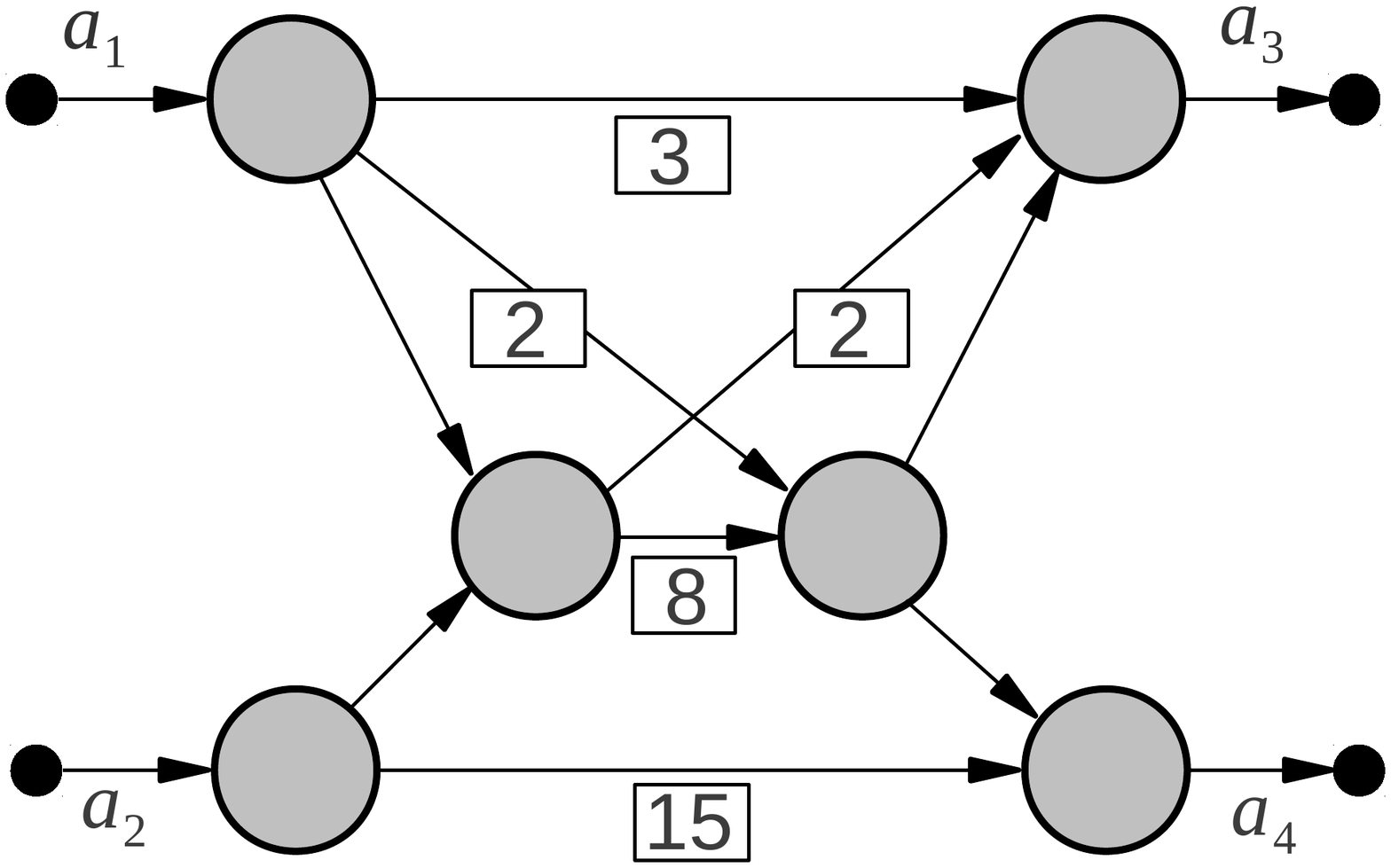}
\end{minipage}
\end{custommargins}
\caption{$\N_1$ (on the left) in Example~\ref{ex:simplest-1},
         $\N_2$ (in the middle) in Example~\ref{ex:simplest-2},
         and $\N_3$ (on the right) in Example~\ref{ex:simplest-3},
         are similar networks, \ie, they have the same number of
         input arcs and the same number of output arcs.
         All missing capacities are the trivial lower bound $0$ and
         the trivial upper bound $K$ (a ``very large number''). }
\label{fig:six-and-eight-node}
\end{center}
\end{figure}

\vspace{-.2in}
\section{Appendix: Proofs and Supporting Lemmas  
         for
         \hyperref[sect:disassemble-and-reassemble]%
         {Section~\ref*{sect:disassemble-and-reassemble}}}
  \label{sect:appendix:disassemble}

Let $\aaa_{\text{in}} = \Set{a_1,\ldots,a_p}$ and $\aaa_{\text{out}}
= \Set{a_{p+1},\ldots,a_{p+q}}$ be fixed, where $p,q\geqslant 1$.  
Let $T\;:\;\power{\aaa_{\text{\rm in,out}}}\,\to\,\intervals{\reals}$.
If $[r,s]$ is an interval of real numbers for some $r\leqslant s$, 
we write $-[r,s]$ to denote the interval $[-s,-r]$.
If $T(A) = [r,s]$ for some $A\subseteq\aaa_{\text{in,out}}$, 
we define $\tmin{}{A} = r$ and $\tmax{}{A} = s$.

The next two results, Lemma~\ref{lem:necessity} and
Lemma~\ref{lem:necessity-11}, are about \emph{tight} typings $T$,
independently of whether $T$ is sound and/or complete for a network
$\N$.

\begin{lemma}
\label{lem:necessity}
Let $T\;:\;\power{\aaa_{\text{\rm in,out}}}\,\to\,\intervals{\reals}$
be a tight typing such that:
\( 
 T(\varnothing) = T(\aaa_{\text{\rm in,out}}) = [0,0]
\).

\medskip
\noindent
\textbf{Conclusion}: For every two-part partition of 
$\aaa_{\text{\rm in,out}}$, say $A\uplus B = \aaa_{\text{\rm in,out}}$, 
if both $T(A)$ and $T(B)$ are defined, then $T(A) = -T(B)$.
\end{lemma}

\begin{proof}
One particular case in the conclusion is when $A=\varnothing$
and $B=\aaa_{\text{\rm in,out}}$, so that trivially
$A\uplus B = \aaa_{\text{\rm in,out}}$, which also implies 
$T(A) = -T(B)$. 

Consider the general case when 
$\varnothing\neq A,B\subsetneq\aaa_{\text{\rm in,out}}$.
From Section~\ref{sect:sound-and-complete}, $\poly{T}$ is the polytope
defined by $T$ and $\constOf{T}$ is the set of linear
inequalities induced by $T$.
For every $(p+q)$-dimensional 
point $f\in\poly{T}$, we have:
\[
  \sum\, f(\aaa_{\text{in}})\ -\ \sum\, f(\aaa_{\text{out}})\ =\ 0
\]
because $T(\aaa_{\text{\rm in,out}}) = \Set{0}$ and therefore:
\[
  0\ \leqslant\ \sum\, \Set{\,a\;|\;a\in\aaa_{\text{in}}\,}
  \ -\ \sum\, \Set{\,a\;|\;a\in\aaa_{\text{out}}\,}\ \leqslant 0
\]
are among the inequalities in $\constOf{T}$. Consider
arbitrary $\varnothing\neq A,B\subsetneq\aaa_{\text{\rm in,out}}$ such that
$A\uplus B = \aaa_{\text{\rm in,out}}$. We can therefore write the equation:
\[
  \quad\sum\, f(A\cap\aaa_{\text{in}})
  \ +\ \sum\, f(B\cap\aaa_{\text{in}})
  \ -\ \sum\, f(A\cap\aaa_{\text{out}})\ 
  \ -\ \sum\, f(B\cap\aaa_{\text{out}})=\ 0
\]
Or, equivalently:
\begin{itemize}
\item[$(\ddag)$]
  \quad $\displaystyle
  \sum\, f(A\cap\aaa_{\text{in}})\ -\ \sum\, f(A\cap\aaa_{\text{out}}) =
  -\sum\, f(B\cap\aaa_{\text{in}}) \ +\ \sum\, f(B\cap\aaa_{\text{out}})
  $
\end{itemize}
for every $f\in\poly{T}$. Hence, relative to $\constOf{T}$,
$f$ maximizes (resp. minimizes)
the left-hand side of equation $(\ddag)$ iff $f$ maximizes
(resp. minimizes) the right-hand side of $(\ddag)$. 
Negating the right-hand side of $(\ddag)$, we also have:
\begin{quote}
   $f$ maximizes (resp. minimizes)
   $\sum f(A\cap\aaa_{\text{in}}) - \sum f(A\cap\aaa_{\text{out}})$
   \ \ if and only if \\[.8ex]
   $f$ minimizes (resp. maximizes)
   $\sum\, f(B\cap\aaa_{\text{in}}) - \sum f(B\cap\aaa_{\text{out}})$
   \ \ and the two quantities are equal.
\end{quote}
Because $T$ is tight, 
every point $f\in\poly{T}$ which maximizes (resp.
minimizes) the objective function:
\[
   \sum\,\Set{\,a\;|\;a\in A\cap\aaa_{\text{in}}\,}
   \ -\ \sum\,\Set{\,a\;|\;a\in A\cap\aaa_{\text{out}}\,}
\]
must be such that:
\begin{alignat*}{1}
   & \tmax{}{A}\ =\ 
   \sum f(A\cap\aaa_{\text{in}})\ -\ \sum f(A\cap\aaa_{\text{out}})
\\[1ex]
  & \Bigl(\text{resp.\ \ }
   \tmin{}{A}\ =
   \ \sum f(A\cap\aaa_{\text{in}})\ -\ \sum f(A\cap\aaa_{\text{out}})
   \Bigr)
\end{alignat*}
We can repeat the same reasoning for $B$. Hence,
if $f\in\poly{T}$ maximizes both sides of $(\ddag)$, then:
\begin{alignat*}{3}
   & \tmax{}{A}\ &&= 
   \ && +\sum f(A\cap\aaa_{\text{in}})\ -\ \sum f(A\cap\aaa_{\text{out}})
\\[1ex]
   & &&=\ && -\sum f(B\cap\aaa_{\text{in}})\ +\ \sum f(B\cap\aaa_{\text{out}})
\\[1ex]
   & &&=\ && -\tmin{}{B}
\end{alignat*}
and, respectively, if $f\in\poly{T}$ minimizes
both sides of $(\ddag)$, then:
\begin{alignat*}{3}
   & \tmin{}{A}\ &&=
   \ && +\sum f(A\cap\aaa_{\text{in}})\ -\ \sum f(A\cap\aaa_{\text{out}})
\\[1ex]
   & &&=\ && -\sum f(B\cap\aaa_{\text{in}})\ +\ \sum f(B\cap\aaa_{\text{out}})
\\[1ex]
   & &&=\ && -\tmax{}{B}
\end{alignat*}
The preceding implies $T(A) = -T(B)$ and concludes the proof.
\end{proof}

\begin{lemma}
\label{prop:necessity-11}
\label{lem:necessity-11}
Let $T\;:\;\power{\aaa_{\text{\rm in,out}}}\,\to\,\intervals{\reals}$
be a tight typing such that
$T(\varnothing) = T(\aaa_{\text{\rm in,out}}) = [0,0]$.

\smallskip
\noindent
\textbf{Conclusion}: For every two-part partition
$A\uplus B = \aaa_{\text{\rm in,out}}$, if $T(A)$ is defined 
and $T(B)$ is undefined, then:
\[
  \min\;\theta(B) = -\tmax{}{A}\quad
  \text{and}\quad
  \max\;\theta(B) = -\tmin{}{A},
\]
where the objective function $\theta(B) := 
\sum (B\cap\aaa_{\text{\rm in}}) - \sum (B\cap\aaa_{\text{\rm out}})$
is minimized and maximized, respectively, w.r.t.
$\constOf{T}$.
\end{lemma}

Hence, if we extend the typing $T$ to a typing $T'$
that includes the type assignment $T'(B) := -T(A)$, then $T'$ is a
tight typing such that $\poly{T} = \poly{T'}$.

\begin{proof}
 If $T(A)=[r,s]$, then $r = \min\theta(A)$ and $s = \max\theta(A)$ 
 where $\theta(A) := \sum (A\cap\aaa_{\text{\rm in}}) - \sum
 (A\cap\aaa_{\text{\rm out}})$ is minimized/maximized
 w.r.t. $\constOf{T}$. Consider the objective function $\Theta := \theta(A) +
 \theta(B)$. Because $T(\aaa_{\text{\rm in,out}}) = [0,0]$, we have
 $\min\Theta = 0 = \max\Theta$ where $\Theta$ is minimized/maximized
 w.r.t. $\constOf{T}$. Think of $\Theta$ as defining a line through the
 origin of the $\bigl(\theta(A),\theta(B)\bigr)$-plane with slope 
 $-45^{o}$ with, say, $\theta(A)$ the horizontal coordinate and 
 $\theta(B)$ the vertical coordinate.  Hence, 
 $\min\theta(B) = -\max\theta(A)$ and 
 $\max\theta(B) = -\min\theta(A)$, which implies the desired conclusion.
\end{proof}

\newcommand{\IOpair}[1]{\Set{#1}}

\begin{Proof}{for part 1 in 
Theorem~\ref{fact:principal-typing-build-up}} 
Let $\N = (\nn,\aaa)$ be a one-node flow network, \ie, 
$\nn = \Set{\nu}$ with all input arcs in $\aaa_{\text{in}}$
and all output arcs in $\aaa_{\text{out}}$ incident to the single
node $\nu$. Algorithm~\ref{alg:oneNodePT} computes a 
principal typing for $\N$ and 
Lemma~\ref{prop:one-node-principal-typings} proves its
correctness.

{\spacing{\normalsize}{1.2}
\begin{algorithm} 
\caption {\quad Calculate 
          Principal Typing for One-Node Network $\N$} 
          \label{alg:oneNodePT}
\begin{algorithmic}[1]
    \Statex \textbf{algorithm name}: $\OnePT$ 
    \Statex \textbf{input}: 
      one-node network $\N$ with input/output
      arcs $\aaa_{\text{in,out}} = \aaa_{\text{in}}\uplus\aaa_{\text{out}}$
    \Statex $\phantom{\textbf{input}:}$
      and lower-bound and upper-bound capacities
      $\lc,\uc : \aaa_{\text{in,out}}\to\nreals$ 
    \Statex \textbf{output}: 
      principal typing 
      $T:\power{\aaa_{\text{\rm in,out}}}\to\intervals{\reals}$
      for $\N$
    \Statex \rule[2pt]{15cm}{.9pt} 
    \State $T(\varnothing)\ :=\ [0,0]$ 
    \State $T(\aaa_{\text{\rm in,out}})\ :=\ [0,0]$  
    \For {every two-part partition $A\uplus B = \aaa_{\text{in,out}}$
          with $A\neq\varnothing\neq B$}  
    \State $A_{\text{in}} := A\cap \aaa_{\text{in}}$;
           $A_{\text{out}} := A\cap \aaa_{\text{out}}$
    \State $B_{\text{in}} := B\cap \aaa_{\text{in}}$;
           $B_{\text{out}} := B\cap \aaa_{\text{out}}$
    \State\label{step:oneNodePT1}
      $r_1 := 
     -\min \SET{\sum \uc(B_{\text{in}}) - \sum \lc(B_{\text{out}}),
      \ \sum \uc(A_{\text{out}}) - \sum \lc(A_{\text{in}})}$
    \State\label{step:oneNodePT2}
      $r_2 := 
     +\min \SET{\sum \uc(A_{\text{in}}) - \sum \lc(A_{\text{out}}),
      \ \sum \uc(B_{\text{out}}) - \sum \lc(B_{\text{in}})}$
    \State
        $T(A)\ :=\ [r_1,r_2]$
    \EndFor
    \State \Return $T$
\end{algorithmic}
\end{algorithm}
}

\begin{lemma}[Principal Typings
for One-Node Networks]
\label{prop:one-node-principal-typings}
Let $\N$ be a network with one node, input arcs $\aaa_{\text{\rm in}}$,
output arcs $\aaa_{\text{\rm out}}$, and lower-bound and upper-bound
capacities 
$\lc,\uc : \aaa_{\text{\rm in}}\uplus\aaa_{\text{\rm out}}\to\nreals$.

\smallskip
\noindent
\textbf{Conclusion}: \emph{$\OnePT(\N)$} is a 
principal typing for $\N$.
\end{lemma}

\begin{proof}
Let $A\uplus B = \aaa_{\text{in,out}}$ be an arbitrary two-part
partition of $\aaa_{\text{in,out}}$, with $A\neq\varnothing\neq B$.
Let $A_{\text{in}} := A\cap \aaa_{\text{in}}$,
$A_{\text{out}} := A\cap \aaa_{\text{out}}$,
$B_{\text{in}} := B\cap \aaa_{\text{in}}$, and
$B_{\text{out}} := B\cap \aaa_{\text{out}}$. Define the non-negative
numbers:
\begin{alignat*}{8}
   & s_{\text{in}}\ &&:=\ && \sum \lc(A_{\text{in}})
     \qquad && s'_{\text{in}}\ &&:=\ && \sum \uc(A_{\text{in}})
\\[1ex]
   &s_{\text{out}}\ && :=\ &&\sum \lc(A_{\text{out}})
     \quad && s'_{\text{out}}\ &&:=\ && \sum \uc(A_{\text{out}})
\\[1ex]
   & t_{\text{in}} &&:= &&\sum \lc(B_{\text{in}})
     \qquad && t'_{\text{in}} &&:= &&\sum \uc(B_{\text{in}})
\\[1ex]
   & t_{\text{out}} &&:= &&\sum \lc(B_{\text{out}})
     \quad && t'_{\text{out}} &&:= &&\sum \uc(B_{\text{out}})
\end{alignat*}
Although tedious and long, one approach to complete the proof
is to exhaustively consider all possible orderings
of the 8 values just defined, using the standard ordering on real
numbers. Cases that do not allow any feasible flow can be 
eliminated from consideration; for feasible flows to be possible,
we can assume that:
\[
  s_{\text{in}}\leqslant s'_{\text{in}}, \quad
  s_{\text{out}}\leqslant s'_{\text{out}}, \quad
  t_{\text{in}}\leqslant t'_{\text{in}}, \quad
  t_{\text{out}}\leqslant t'_{\text{out}}, \quad
\]
and also assume that:
\[
  s_{\text{in}} + t_{\text{in}} \leqslant 
  s'_{\text{out}} + t'_{\text{out}}, \quad
  s_{\text{out}} + t_{\text{out}} \leqslant 
  s'_{\text{in}} + t'_{\text{in}}, \quad
\]
thus reducing the total number of cases to consider. We 
consider the intervals $[s_{\text{in}},s'_{\text{in}}]$,
$[s_{\text{out}},s'_{\text{out}}]$, $[t_{\text{in}},t'_{\text{in}}]$, 
and $[t_{\text{out}},t'_{\text{out}}]$, and their relative positions,
under the preceding assumptions. Define the objective $\theta(A)$:
\[
    \theta(A)\ :=\ \sum A_{\text{in}} - \sum A_{\text{out}} .
\]
If $T$ is a principal typing for $\N$, and therefore tight,
with $T(A) = [r_1,r_2]$, then $r_1$ is the
minimum possible feasible value of $\theta(A)$ and $r_2$ is the maximum
possible feasible value of $\theta(A)$ relative to $\constOf{T}$.

We omit the details of the just-outlined exhaustive proof by cases.
Instead, we argue for the correctness of $\OnePT$ more informally.
It is helpful to consider the particular case when all lower-bound
capacities are zero, \ie, the case when 
$s_{\text{in}} = s_{\text{out}} = t_{\text{in}} = t_{\text{out}} = 0$. In this
case, it is easy to see that:
\begin{alignat*}{4}
   &r_1\ &&=\ 
   &&- \min \SET{ \sum \uc(B_{\text{in}}), \sum \uc(A_{\text{out}})}
   \quad &&\text{maximum amount entering at 
                 $B_{\text{in}}$ and exiting at $A_{\text{out}}$},
\\ & && && &&
          \text{while minimizing amount entering at $A_{\text{in}}$},
\\[1ex]
   &r_2\ &&=\ 
   &&+ \min \SET{ \sum \uc(A_{\text{in}}), \sum \uc(B_{\text{out}})} 
   \quad &&\text{maximum amount entering at 
                 $A_{\text{in}}$ and exiting at $B_{\text{out}}$},
\\ & && && &&
          \text{while minimizing amount exiting at $A_{\text{out}}$},
\end{alignat*}
which are exactly the endpoints of the type $T(A)$ returned
by $\OnePT(\N)$ in the particular case when all lower-bounds
are zero.

Consider now the case when some of the lower-bounds are not zero.
To determine the maximum throughput $r_2$ using the arcs of $A$, 
we consider two quantities:
\[
   r_2' := \sum \uc(A_{\text{in}}) - \sum \lc(A_{\text{out}})
   \quad\text{and}\quad
   r_2'' := \sum \uc(B_{\text{out}}) - \sum \lc(B_{\text{in}}).
\]
It is easy to see that $r_2'$ is the flow that is simultaneously
maximized at $A_{\text{in}}$ and minimized at $A_{\text{out}}$,
provided that $r_2' \leqslant r_2''$, \ie, the whole amount $r_2'$ can
be made to enter at $A_{\text{in}}$ and to exit at
$B_{\text{out}}$. However, if $r_2' > r_2''$, then only the amount
$r_2''$ can be made to enter at $A_{\text{in}}$ and to exit at
$B_{\text{out}}$. Hence, the desired value of $r_2$ is
$\min\Set{r_2',r_2''}$, which is exactly the higher endpoint
of the type $T(A)$ returned by $\OnePT(\N)$. 
A similar argument, here omitted, is used again to determine
the minimum throughput $r_1$ using the arcs of $A$.
\end{proof}

\paragraph{Complexity of $\OnePT$.} We estimate
the run-time of $\OnePT$ as a function of: $d =
\size{\aaa_{\text{in,out}}} \geqslant 2$, the number of input/output arcs,
also assuming that there is at least one input arc and one output arc
in $\N$.  $\OnePT$ assigns the type/interval $[0,0]$ to $\varnothing$
and $\aaa_{\text{in,out}}$. For every $\varnothing\neq
A\subsetneq\aaa_{\text{in,out}}$, it then computes a type $[r_1,r_2]$,
simultaneously for $A$ and its complement $B =
\aaa_{\text{in,out}}-A$.  (That $B$ is assigned $[-r_2,-r_1]$ is not
explicitly shown in Algorithm~\ref{alg:oneNodePT}.) Hence, $\OnePT$
computes $(2^d-2)/2 = (2^{d-1}-1)$ such types/intervals, each
involving $8$ summations and $4$ subtractions, in
lines~\ref{step:oneNodePT1} and~\ref{step:oneNodePT2}, on the
lower-bound capacities ($d$ of them) and upper-bound capacities ($d$
of them) of the input/output arcs. The run-time complexity of
$\OnePT$ is therefore $\bigO{2^d}$.

\end{Proof}

\Hide{

\begin{remark}
A different version of Algorithm~\ref{alg:oneNodePT} uses
linear programming to compute the typing of a one-node
network, but this is an unnecessary overkill. The resulting
run-time complexity is also worse than that of our version here.
The linear-programming version works as follows. 
Let $\EE$ be the set of flow-conservation equations and $\CC$
the set of capacity-constraint inequalities of the one-node
network. For every $A\in\power{\aaa_{\text{\rm in,out}}}$, 
we define the objective
$\theta(A) := \sum (A\cap\aaa_{\text{\rm in}})
- \sum (A\cap\aaa_{\text{\rm out}})$. 
The desired type $T(A) = [r_1,r_2]$ is obtained by setting: 
\( r_1 := \min \theta(A)\ \text{and}\ r_2 := \max \theta(A), \) 
\ie, the objective $\theta(A)$ is minimized/maximized
relative to $\EE\cup\CC$ using linear programming.
\end{remark}

}

\begin{Proof}{for part 2 in 
Theorem~\ref{fact:principal-typing-build-up}} 
Given a principal typing $T$ for network $\N$ with arcs 
$\aaa = \aaa_{\text{in}}\uplus\aaa_{\text{out}}\uplus\aaa_{\text{\#}}$,
together with $a^{+}\in\aaa_{\text{in}}$ and $a^{-}\in\aaa_{\text{out}}$,  
we need to compute a principal typing $T'$ for the network
$\bind{a}{\N}$. This is carried out by Algorithm~\ref{alg:effbind}
below and its correctness established by 
Lemma~\ref{lem:efficient-binding}.

{\spacing{\normalsize}{1.2}
\begin{algorithm}
\caption {\quad Bind One Input-Output Pair Efficiently} 
          \label{alg:effbind}
\begin{algorithmic}[1]
    \Statex \textbf{algorithm name}: $\BindTone$
    \Statex \textbf{input}: principal typing 
      $T:\power{\aaa_{\text{\rm in,out}}}\to\intervals{\reals}$,
      \ $a^{+}\in\aaa_{\text{\rm in}}$, \ $a^{-}\in\aaa_{\text{\rm out}}$
    \Statex $\phantom{\textbf{input}:}$ 
      where for every two-part partition $A\uplus B = 
 \aaa'_{\text{\rm in,out}} := {\aaa_{\text{\rm in,out}}} - \Set{a^{+},a^{-}}$,
    \Statex $\phantom{\textbf{input}:}$ 
      either both $T(A)$ and $T(B)$ are defined or both 
      $T(A)$ and $T(B)$ are undefined
    \Statex \textbf{output}: principal typing 
      $T':\power{\aaa'_{\text{\rm in,out}}}\to\intervals{\reals}$
    \Statex $\phantom{\textbf{output}:}$ 
    where $\poly{T'} = 
 \rest{\poly{\constOf{T}\cup\Set{a^{+}=a^{-}}}}{\aaa'_{\text{\rm in,out}}}$
    \Statex \rule[2pt]{15cm}{.9pt} 
    \Statex Definition of intermediate typing 
          $T_1 :\power{\aaa'_{\text{\rm in,out}}}\to\intervals{\reals}$
    \State\label{step:effBing1}
       $T_1\ :=\ \rest{T}{\power{\aaa'_{\text{\rm in,out}}}}$,
    \Statex
      \ie, every type assigned by $T$ to a set $A$ such that
      $A\cap\Set{a^{+},a^{-}}\neq\varnothing$ is omitted in $T_1$
    \Statex \rule[2pt]{10cm}{.1pt} 
    \Statex Definition of intermediate typing 
          $T_2 :\power{\aaa'_{\text{\rm in,out}}}\to\intervals{\reals}$
    \State\label{step:effBing2}
       $T_2\ :=\ T_1\bigl[\aaa'_{\text{\rm in,out}}\mapsto [0,0]\bigr]$
    \Statex
      \ie, the type assigned by $T_1$ to $\aaa'_{\text{\rm in,out}}$,
      if any, is changed to the type $[0,0]$ in $T_2$
    \Statex \rule[2pt]{10cm}{.1pt} 
    \Statex Definition of final typing 
          $T' :\power{\aaa'_{\text{\rm in,out}}}\to\intervals{\reals}$
    \For {every two-part partition $A\uplus B = \aaa'_{\text{in,out}}$}
      \label{step:effBing5}
      \If {$T_2(A)$ is defined with $T_2(A) = [r_1,s_1]$ 
        \textbf{and} $T_2(B)$ is defined with $-T_2(B) = [r_2,s_2]$} 
      \label{step:effBing55}
      \State\label{step:effBing555}
             $T'(A) := [\max\Set{r_1,r_2},\min\Set{s_1,s_2}]$ ;
             $T'(B) := -T'(A)$ 
      \ElsIf {both $T_2(A)$ \textbf{and} $T_2(B)$ are undefined} 
      \State $T'(A)$ is undefined ; $T'(B)$ is undefined
      \EndIf
    \EndFor\label{step:effBing6}
    \State \Return $T'$
\end{algorithmic}
\end{algorithm}
}

\begin{lemma}[Typing for Binding One Input/Output Pair]
\label{prop:efficient-binding}
\label{lem:efficient-binding}
Let $T :\power{\aaa_{\text{\rm in,out}}}\to\intervals{\reals}$ be a 
principal typing for network $\N$, with input/output arcs
$\aaa_{\text{\rm in,out}} = \aaa_{\text{\rm in}}\uplus\aaa_{\text{\rm out}}$,
and let $a^{+}\in\aaa_{\text{\rm in}}$ and $a^{-}\in\aaa_{\text{\rm out}}$.

\smallskip
\noindent
\textbf{Conclusion}: \emph{$\BindTone(a,T)$} is a 
principal typing for \emph{$\bind{a}{\N}$}.
\end{lemma}

\begin{proof}
 Consider the intermediate typings 
 $T_1$ and $T_2$ as defined in algorithm $\BindTone$:
 \[
  T_1, T_2 :\power{\aaa'_{\text{\rm in,out}}}\to\intervals{\reals} 
  \quad
  \text{where $\aaa'_{\text{\rm in,out}} = 
      \aaa_{\text{\rm in,out}} - \Set{a^{+},a^{-}}$.}
 \]
 The definitions of $T_1$ and $T_2$ can be repeated
 differently as follows.
 For every $A\subseteq\aaa_{\text{\rm in,out}}$:
\begin{alignat*}{4}
  & T_1(A)\ &&:=\ && 
    \begin{cases}
    T(A) &\text{if $A\subseteq\aaa'_{\text{\rm in,out}}$
                and $T(A)$ is defined},
    \\
    \text{undefined}
       \quad &\text{if $A\not\subseteq\aaa'_{\text{\rm in,out}}$
               or $T(A)$ is undefined},
    \end{cases}
\\
  & T_2(A)\ &&:=
     \ &&\begin{cases}
         T_1(A) \qquad &\text{if $A\subsetneq\aaa'_{\text{\rm in,out}}$
           and $T_1(A)$ is defined},
          \\
         [0,0]  \qquad &\text{if $A=\aaa'_{\text{\rm in,out}}$},
          \\
         \text{undefined}
         \quad &\text{if $A\not\subseteq\aaa'_{\text{\rm in,out}}$
             or $T_1(A)$ is undefined}.
             \end{cases}
\end{alignat*}
If $T$ is tight, then so is $T_1$.  
The only difference
between $T_1$ and $T_2$ is that the latter includes the new
type assignment $T_2(\aaa'_{\text{\rm in,out}}) = [0,0]$, which
is equivalent to the constraint:
\[
  \sum \aaa'_{\text{\rm in}} - \sum \aaa'_{\text{\rm out}} = 0,
  \qquad\text{where 
  $\aaa'_{\text{\rm in}} = \aaa_{\text{\rm in}}-\Set{a^{+}}$ and
  $\aaa'_{\text{\rm out}} = \aaa_{\text{\rm out}}-\Set{a^{-}}$},
\]
which, given 
the fact that $\sum \aaa_{\text{\rm in}} - \sum \aaa_{\text{\rm out}} = 0$,
is in turn equivalent to the constraint $a^{+}=a^{-}$. 
This implies the following equalities:
\begin{alignat*}{4}
 &\rest{\poly{\constOf{T}\cup\Set{a^{+}=a^{-}}}}{\aaa'_{\text{\rm in,out}}} 
  \ &&=\ &&
  \poly{\constOf{T_1}\cup\Set{\aaa'_{\text{\rm in}}=\aaa'_{\text{\rm out}}}}
  \\
  & &&=\ &&\poly{\constOf{T_2}}
  \\
  & &&=\ &&\poly{T_2}
\end{alignat*}
Hence, if $T$ is a principal typing for $\N$, then 
$T_2$ is a principal typing for $\bind{a}{\N}$.
It remains to show that $T'$ as defined in algorithm 
$\BindTone$ is the tight version of $T_2$.

We define an additional typing 
$T_3 :\power{\aaa'_{\text{\rm in,out}}}\to\intervals{\reals}$ as
follows -- for the purposes of this proof only, $T_3$ is not
computed by algorithm $\BindTone$.
For every $A \subseteq\aaa'_{\text{\rm in,out}}$ for which
$T_2(A)$ is defined, let the objective $\theta_A$ be
$\sum(A\cap\aaa'_{\text{\rm in}}) - \sum(A\cap\aaa'_{\text{\rm out}})$ 
and let:
\begin{alignat*}{4}
   & T_3(A)\ :=\ [r,s]
 \qquad \text{where $r = \min\theta_A$ and $s = \max\theta_A$
   relative to $\constOf{T_2}$}.
\end{alignat*}
$T_3$ is obtained from $T_2$ in an ``expensive'' process, because it
uses a linear-programming algorithm to minimize/maximize the
objectives $\theta_A$. Clearly $\poly{T_2} = \poly{T_3}$.  Moreover,
$T_3$ is guaranteed to be tight by the definitions and results in
Section~\ref{sect:flow-networks} -- we leave to the reader the
straightforward details showing that $T_3$ is tight.  In
particular, for every $A \subseteq\aaa'_{\text{\rm in,out}}$ for which
$T_2(A)$ is defined, it holds that $T_3(A) \subseteq T_2(A)$.
Hence, for every $A\uplus B =\aaa'_{\text{\rm in,out}}$ for which
$T_2(A)$ and $T_2(B)$ are both defined:
\begin{equation}
\label{eq:effbind1} 
   T_3(A)\ \subseteq\ T_2(A)\,\cap\,-T_2(B) ,
\end{equation}
since $T_3(A) = -T_3(B)$ by Lemma~\ref{lem:necessity}.
Keep in mind that:
\begin{equation}
\label{eq:effbind2}
  \text{$\poly{T_3}$ is the largest polytope satisfying $\constOf{T_2}$},
\end{equation}
and every other polytope satisfying $\constOf{T_2}$ is 
a subset of $\poly{T_3}$.
We define one more typing 
$T_4 :\power{\aaa'_{\text{\rm in,out}}}\to\intervals{\reals}$ by
appropriately restricting $T_2$; namely, for every two-part partition
$A\uplus B =\aaa'_{\text{\rm in,out}}$:
\begin{alignat*}{4}
  & T_4(A)\ &&:=\ && 
    \begin{cases}
    \phantom{-}T_2(A)\,\cap\,-T_2(B)\quad
      &\text{if both $T_2(A)$ and $T_2(B)$ are defined},
    \\
    \phantom{-}\text{undefined}  
    &\text{if both $T_2(A)$ and $T_2(B)$ are undefined}.
    \end{cases}
\end{alignat*}
Hence, $\poly{T_4}$ satisfies $\constOf{T_2}$, so that also
for every $A\subseteq \aaa'_{\text{\rm in,out}}$ for which
$T_4(A)$ is defined, we have $T_3(A) \supseteq T_4(A)$, by
(\ref{eq:effbind2}) above.
Hence, for every $A\uplus B =\aaa'_{\text{\rm in,out}}$ for which
$T_4(A)$ and $T_4(B)$ are both defined, we have:
\begin{equation}
\label{eq:effbind3} 
   T_3(A)\ \supseteq\ T_4(A)\ =
   \ -T_4(B)\ =\ T_2(A)\,\cap\,-T_2(B) .
\end{equation}
Putting (\ref{eq:effbind1}) and (\ref{eq:effbind3}) together:
\[ 
   T_2(A)\,\cap\,-T_2(B)\ \subseteq
   \ T_3(A)\ \subseteq\ \ T_2(A)\,\cap\,-T_2(B),
\]
which implies $T_3(A) = T_2(A)\,\cap\,-T_2(B) = T_4(A)$.
Hence, also, for every $A\subseteq \aaa'_{\text{\rm in,out}}$ 
for which $T_3(A)$ is defined, $T_3(A) = T_4(A)$. This implies 
$\poly{T_3} = \poly{T_4}$ and that $T_4$ is tight. $T_4$ is none
other than $T'$ in algorithm $\BindTone$, thus concluding the proof
of the first part in the conclusion of the proposition. For
the second part, it is readily checked that if $T$ is a total
typing, then so is $T'$ (details omitted).
\end{proof}
\end{Proof}

\paragraph{Complexity of $\BindTone$.} We measure the run-time
of $\BindTone$ by the number of bookkeeping steps (whether a
variable/arc name is in a set or not) and the number of number-comparisons
(there are no additions and subtractions in $\BindTone$) as a function
of:
\begin{itemize}
\item $\size{\aaa_{\text{in,out}}}$, the number of input/output arcs,
\item $\size{T}$, the
      number of assigned types in the initial typing $T$.
\end{itemize} 
We consider each of the three parts separately:
\begin{enumerate}
\item The first part, line \ref{step:effBing1},  
      runs in $\bigO{\size{\aaa_{\text{in,out}}}\cdot\size{T}}$ time,
      according to the following reasoning. Suppose the types of $T$
      are organized as a list with $\size{T}$ entries, which we can
      scan from left to right. The algorithm removes every type
      assigned to a subset $A\subseteq\aaa_{\text{in,out}}$
      intersecting $\Set{a^{+},a^{-}}$. There are $\size{T}$ types to
      be inspected, and the subset $A$ to which $T(A)$ is assigned has
      to be checked that it does not contain $a^{+}$ or $a^{-}$.  The
      resulting intermediate typing $T_1$ is such that
      $\size{T_1}\leqslant\size{T}$.
\item The second part of $\BindTone$, line 
      \ref{step:effBing2}, runs in
      $\bigO{\size{\aaa'_{\text{in,out}}}\cdot\size{T_1}}$ time. 
      It inspects each of the $\size{T_1}$ types, looking
      for one assigned to $\aaa'_{\text{in,out}}$, each
      such inspection requiring $\size{\aaa'_{\text{in,out}}}$
      comparison steps. If it finds such a type, it replaces it by
      $[0,0]$. If it does not find such a type, it adds
      the type assignment $\Set{\aaa'_{\text{in,out}}\mapsto [0,0]}$.
      The resulting intermediate typing $T_2$ is such that
      $\size{T_2} = \size{T_1}$ or $\size{T_2} = 1+\size{T_1}$.
\item The third part, from line \ref{step:effBing5} to
      line \ref{step:effBing6}, runs in 
      $\bigO{\size{\aaa'_{\text{in,out}}}\cdot\size{T_2}^2}$ time.
      For every type $T_2(A)$, it looks for a type $T_2(B)$ in 
      at most $\size{T_2}$ scanning steps, such that 
      $A\uplus B = \aaa'_{\text{in,out}}$ in at most 
      $\size{\aaa'_{\text{in,out}}}$ comparison steps; if and
      when it finds a type $T_2(B)$, which is guaranteed to be
      defined, it carries out the operation
      in line~\ref{step:effBing555}.
\end{enumerate}
Adding the estimated run-times in the three parts,
the overall run-time of $\BindTone$ is  
$\bigO{\size{\aaa_{\text{in,out}}}\cdot\size{T}^2}$. 
Let $\delta = \size{\aaa_{\text{in,out}}} \geqslant 2$.
In the particular case when $T$ is a total typing which
therefore assigns a type to each of the $2^{\delta}$ subsets of 
$\aaa_{\text{in,out}}$, the overall run-time of $\BindTone$ is 
$\bigO{\delta\cdot 2^{2\delta}} = \bigO{2^{\log\delta + 2\delta}} = 2^{\bigOO{\delta}}$.

Note there are no arithmetical steps (addition, multiplication, etc.)
in the execution of $\BindTone$; besides the bookkeeping involved in
partitioning $\aaa'_{\text{in,out}}$ in two disjoint parts, $\BindTone$
uses only comparison of numbers in line~\ref{step:effBing555}.

\begin{Proof}{for part 3 in 
Theorem~\ref{fact:principal-typing-build-up}} 
Let $\N'$ and $\N''$ be two separate networks, with arcs 
$\aaa' = \aaa'_{\text{in}}\uplus\aaa'_{\text{out}}\uplus\aaa'_{\text{\#}}$
and 
$\aaa''=\aaa''_{\text{in}}\uplus\aaa''_{\text{out}}\uplus\aaa''_{\text{\#}}$,
respectively. 
The \emph{parallel addition} of $\N'$ and $\N''$,
denoted $\ConnP{\N'}{\N''}$, simply places $\N'$ and $\N''$
next to each other without connecting any of their external
arcs. If $\N = \ConnP{\N'}{\N''}$, then
the input and output arcs of $\N$ are
$\aaa'_{\text{in}}\uplus\aaa''_{\text{in}}$ and
$\aaa'_{\text{out}}\uplus\aaa''_{\text{out}}$, respectively, 
and its internal arcs are $\aaa'_{\text{\#}}\uplus\aaa''_{\text{\#}}$.
If $\N'$ and $\N''$ are each connected, we view $\N$ as
a network with two separate connected components.

Let $T'$ and $T''$ be 
principal typings for the
two separate networks $\N'$ and $\N''$. We define the
the \emph{parallel addition} of $T'$ and $T''$ as follows:
\[
    (\ParAdd{T'}{T''})(A)
    \ :=\ \begin{cases}
         [0,0]\qquad   & \text{if $A = \varnothing$ or 
            $A = \aaa'_{\text{\rm in,out}}\cup\aaa''_{\text{\rm in,out}}$},
         \\[1ex]
         T'(A)    &\text{if $A\subseteq\aaa'_{\text{\rm in,out}}$
                        and $T'(A)$ is defined},
         \\[1ex]
         T''(A)    &\text{if $A\subseteq\aaa''_{\text{\rm in,out}}$
                        and $T''(A)$ is defined},
         \\[1ex]
         \text{undefined}\qquad &\text{otherwise}.
         \end{cases}
\]

\begin{lemma}[Typing for Parallel Addition]
\label{prop:parallel-addition}
\label{lem:parallel-addition}
\label{prop:partial-addition}
Let $\N'$ and $\N''$ be two separate networks with external arcs
$\aaa'_{\text{\rm in,out}} = \aaa'_{\text{\rm in}}\uplus\aaa'_{\text{\rm out}}$
and $\aaa''_{\text{\rm in,out}} = 
\aaa''_{\text{\rm in}}\uplus\aaa''_{\text{\rm out}}$,
respectively. Let $T'$ and $T''$ be 
principal typings for $\N'$ and $\N''$, respectively.

\smallskip
\noindent
\textbf{Conclusion}: $(\ParAdd{T'}{T''})$ is a 
principal typing for the network $\ConnP{\N'}{\N''}$.
\end{lemma}

\begin{proof}
There is no communication between $\N'$ and $\N''$.
The conclusion of the proposition is a straightforward consequence
of the definitions. All details omitted.
\end{proof}

The typing $(\ParAdd{T'}{T''})$ is partial even when
$T'$ and $T''$ are total typings. We need to define the 
\emph{total parallel addition} 
of total typings which is another total typing.
If $[r_1,s_1]$ and $[r_2,s_1]$
are intervals of real numbers for some $r_1\leqslant s_1$
and $r_2\leqslant s_2$, we write $[r_1,s_1]+[r_2,s_2]$ 
to denote the interval $[r_1+r_2,s_1+s_2]$:
\[
  [r_1,s_1]+[r_2,s_2]\ :=
  \ \Set{\, t\in\reals\;|\; r_1+r_2\leqslant t\leqslant s_1+s_2\,}
\]
Let $T'$ and $T''$ be tight and total typings over 
disjoint sets of external arcs, 
$\aaa'_{\text{\rm in,out}} = \aaa'_{\text{\rm in}}\uplus\aaa'_{\text{\rm out}}$ 
and 
$\aaa''_{\text{\rm in,out}}=\aaa''_{\text{\rm in}}\uplus\aaa''_{\text{\rm out}}$, 
respectively.
We define the \emph{total parallel addition} $(\TotAdd{T'}{T''})$ of
the typings $T'$ and $T''$ as follows. For every
$A\subseteq\aaa'_{\text{\rm in,out}}\cup\aaa''_{\text{\rm in,out}}$:
\[
    (\TotAdd{T'}{T''})(A)
    \ :=\ \begin{cases}
         [0,0]\qquad   & \text{if $A = \varnothing$ or 
            $A = \aaa'_{\text{\rm in,out}}\cup\aaa''_{\text{\rm in,out}}$},
         \\[1ex]
         T'(A')+T''(A'')
              &\text{if $A = A'\uplus A''$ with
               $A' = A\cap\aaa'_{\text{\rm in,out}}$
               and $A'' = A\cap\aaa''_{\text{\rm in,out}}$.}
         \end{cases}
\]

\begin{lemma}[Total Typing for Parallel Addition]
\label{prop:total-addition}
\label{lem:total-addition}
Let $\N'$ and $\N''$ be two separate networks with external arcs
$\aaa'_{\text{\rm in,out}} = \aaa'_{\text{\rm in}}\uplus\aaa'_{\text{\rm out}}$
and 
$\aaa''_{\text{\rm in,out}}=\aaa''_{\text{\rm in}}\uplus\aaa''_{\text{\rm out}}$,
respectively. Let $T'$ and $T''$ be 
principal typings for $\N'$ and $\N''$, respectively.

\smallskip
\noindent
\textbf{Conclusion}: $(\TotAdd{T'}{T''})$ is a 
tight, sound, and complete typing for the network $\ConnP{\N'}{\N''}$.
Moreover, if $T'$ and $T''$ are total (because $\N'$ and $\N''$ are
each a connected network), then $(\TotAdd{T'}{T''})$ is also total --
but not locally total, because $\ConnP{\N'}{\N''}$ consists of two
separate components.
\end{lemma}

\begin{proof}
Similar to the proof of Lemma~\ref{prop:partial-addition}.
Straightforward consequence from the fact there is no communication
between $\N'$ and $\N''$. All details omitted.
\end{proof}

\paragraph{Complexity of $\ParAdd{}{}$ and $\TotAdd{}{}$.}
The cost of $(\ParAdd{T'}{T''})(A)$ is the cost of determining
whether $A \subseteq\aaa'_{\text{\rm in,out}}$ or 
$A \subseteq\aaa''_{\text{\rm in,out}}$, which is therefore
a number of bookkeeping steps linear in
$\size{\aaa'_{\text{\rm in,out}}}+\size{\aaa''_{\text{\rm in,out}}}$.
There are no arithmetical operations in the computation of
$(\ParAdd{T'}{T''})(A)$.

The cost of $(\TotAdd{T'}{T''})(A)$ is a little more involved.  In
addition to the bookkeeping steps, it includes two number-additions
after the two subsets $A' = A\cap\aaa'_{\text{\rm in,out}}$ and $A'' =
A\cap\aaa''_{\text{\rm in,out}}$ are determined.

\bigskip

We need one more operation on typings, $\BindT$, to define
Algorithm~\ref{alg:basic-compositional} precisely.  $\BindT$ uses
operations $\TotAdd{}{}$ defined above and $\BindTone$ defined in
the proof of part 2 of Theorem~\ref{fact:principal-typing-build-up}.
Let $\N$ be a network with $k\geqslant 2$ components, say, $\N
= \ConnPP{\M_{1}}{\ConnPP{\M_{2}}{\ConnPP{\cdots}{\M_{k}}}}$, with
input arcs $\aaa_{\text{in}}$ and output arcs $\aaa_{\text{out}}$.
Let $T$ be a 
principal typing for $\N$, say, $T
= \ParAdd{U_{1}}{\ParAdd{U_{2}}{\ParAdd{\cdots}{U_{k}}}}$, where $U_i$
is a 
principal typing for component $\M_i$. Let
$a^{+}\in\aaa_{\text{in}}$ and $a^{-}\in\aaa_{\text{out}}$. 

There are two cases: (1) The two halves $a^{+}$ and $a^{-}$ occur in
the same component $\M_i$, and (2) the two halves occur in two
separate components $\M_i$ and $\M_j$ with $i\neq j$.  We define
$\BindT(a,T)$ by:
\[
   \BindT(a,T)\ :=\ 
   \begin{cases}
   \bigoplus \bigl(\Set{U_{1},\ldots,U_k} - \Set{U_i}\bigr)
    \ \oplus\ \BindTone(a,U_i)
     \quad &\text{if $a^{+}$ and $a^{-}$ are both in}
   \\
     \quad &\text{component $\M_i$,}
   \\[1.2ex]
   \bigoplus \bigl(\Set{U_{1},\ldots,U_k} - \Set{U_i,U_j}\bigr)
    \ \oplus\ \BindTone(a,\TotAdd{U_i}{U_j})
     \quad &\text{if $a^{+}$ and $a^{-}$ are in two}
   \\
     \quad &\text{separate components,}
   \\
     \quad &\text{$\M_i$ and $\M_j$, with $i\neq j$.}
   \end{cases}
\]
\end{Proof}

\paragraph{Complexity of $\BindT$.}
Part of the information given to $\BindT$ is whether $a^{+}$ and
$a^{-}$ are in the same component $\M_i$ or in two different
components $\M_i$ and $\M_j$ of $\N$. This is a bookkeeping task that
can be included in the execution of
Algorithm~\ref{alg:basic-compositional}. Suppose $\delta\geqslant 2$
is an upper bound on the number of external arcs of $\M_i$ (in the first
case) or $\ConnP{\M_i}{\M_j}$ (in the second case). 

In the first case, the cost of executing $\BindT(a,T)$ is the cost of
executing $\BindTone(a,T)$, which is 
$\bigO{\delta\cdot 2^{2\delta}}  = 2^{\bigOO{\delta}}$.
In the second case, we need to add the initial cost of computing
$\TotAdd{U_i}{U_j}$, which consists of performing two additions for
each of $2^{\delta}- 2$ subsets (of the external arcs of
$\ConnP{\M_i}{\M_j}$, \ie, the set of arcs/variables over which
$\TotAdd{U_i}{U_j}$ is defined).  This initial cost is
$\bigO{2^{\delta}}$, so that the total of executing $\BindT(a,T)$ in
the second case is again $2^{\bigOO{\delta}}$.
 
\bigskip

{\spacing{\normalsize}{1.2} 
\begin{algorithm} 
\caption {\quad Calculate a Principal Typing for 
          Network $\N$} 
          \label{alg:basic-compositional}
\begin{algorithmic}[1]
    \Statex \textbf{algorithm name}: $\CompPT$
    \Statex \textbf{input}: 
      $\N = (\nn,\aaa)$ is a connected flow network,
    \Statex $\phantom{\textbf{input}:}$
      $\sigma = b_1 b_2 \cdots b_m$ is an ordering 
      of the internal arcs of $\N$ (a ``binding schedule''), 
    \Statex $\phantom{\textbf{input}:}$
         where $\nn = \Set{\nu_1,\nu_2,\ldots,\nu_n}$,
         \ \ $\aaa = \aaa_{\text{in,out}}\uplus\aaa_{\text{\#}}$,\ \ and
         \ $\aaa_{\text{\#}} = \Set{b_1, b_2, \ldots , b_m}$
    \Statex \textbf{output}: 
       principal typing $T$ for $\N$
    \Statex \rule[2pt]{15cm}{.9pt} 
    \State \label{step:basiccomp1}
      $\N_0\ :=\ \ConnPP{\M_{1}}{\ConnPP{\M_{2}}
                 {\ConnPP{\cdots}{\M_{n}}}}$
    \Statex \qquad where
             $\Set{\M_{1}, \M_{2}, \ldots, \M_{n}} = \Break{\N}$
    \State \label{step:basiccomp2}
      $T_0\ := 
    \ \ParAdd{\OnePT(\M_{1})}{\ParAdd{\OnePT(\M_{2})}
      {\ParAdd{\cdots}{\OnePT(\M_{n})}}}$
    \For {every $k = 0, 1, \ldots, m-1$} \label{step:basiccomp3}
     \State \label{step:basiccomp4}
        {$\N_{k+1}\ :=\ \bind{b_{k+1}}{\N_{k}}$}
     \State 
        {$T_{k+1}\ :=\ \BindT(b_{k+1},T_{k})$}
    \EndFor \label{step:basiccomp33}
    \State {$T := T_m$}
    \State \Return $T$
\end{algorithmic}
\end{algorithm}
}

\begin{lemma}[Inferring Principal Typings]
\label{thm:principal-typing-in-basic-compositional}
  Let $\N$ be a connected flow network and let
  \ $\sigma = b_1  b_2\cdots b_m$ \ be an ordering of all the 
  internal arcs of $\N$ (a
  ``binding schedule'').  Then the typing 
  \emph{$T = \CompPT(\N,\sigma)$} is principal
  for network $\N$.
\end{lemma}

\begin{proof}
It suffices to show that, for every $k = 0, 1, 2,\ldots,m$, the typing
$T_k$ is principal for $\N_k$, where $T_k$ consists
of the parallel addition of as many 
principal typings as there are components in $\N_k$. This is true for $k=0$ by
Lemmas~\ref{prop:one-node-principal-typings}
and~\ref{prop:partial-addition}, and is true again
for each $k \geqslant 1$ by Lemmas~\ref{prop:efficient-binding}
and~\ref{prop:total-addition}. The final $\N_m$ consists of a single
component (itself) and the resulting typing $T_m$ is therefore total.
\end{proof}

\Hide{
Let $\N$ be a connected network and $\aaa_{\text{in,out}}$ its set of
external arcs. We define the \emph{dimension} of $\N$,
$\dime{\N} := \size{\aaa_{\text{in,out}}}$, which is the number of
external arcs of $\N$. 

Let $\N$ be a network with $k\geqslant 2$ components, say, $\N
= \ConnPP{\M_{1}}{\ConnPP{\M_{2}}{\ConnPP{\cdots}{\M_{k}}}}$, 
where each $\M_i$ is connected. We define $\dime{\N}$ by setting:
\[
   \dime{\N}\ :=
   \ \max\,\Set{\dime{\M_{1}},
   \;\dime{\M_{2}},\;\ldots\;,\;\dime{\M_{k}}}.
\]
}
Let $\N$ be a connected network and $\sigma = b_1 b_2 \cdots b_m$ 
an ordering of the internal arcs of $\N$. As in 
line~\ref{step:basiccomp1} of Algorithm~\ref{alg:basic-compositional},
let: 
\[
  \N_0\ =\ \ConnPP{\M_{1}}{\ConnPP{\M_{2}}
                 {\ConnPP{\cdots}{\M_{n}}}}
\]
where $\Set{\M_{1}, \M_{2}, \ldots, \M_{n}} = \Break{\N}$, and as
in line~\ref{step:basiccomp4}, let:
\[
  \N_{1}\ =\ \bind{b_{1}}{\N_{0}},
  \ \N_{2}\ =\ \bind{b_{2}}{\N_{1}},
  \ \ldots\ ,
  \ \N_{m}\ =\ \bind{b_{m}}{\N_{m-1}} .
\]
These $\N_0, \N_1, \ldots, \N_m$ are the same as in 
Section~\ref{sect:disassemble-and-reassemble}.
\Hide{
We define $\Index{\sigma}$, the \emph{index of the binding
schedule} $\sigma$, by setting:
\[
   \Index{\sigma}\ :=\ \max\,\Set{\exDim{\N_0},
   \exDim{\N_1},\ldots, \exDim{\N_m}} .
\]
In words, $\Index{\sigma}$ 
is an upper-bound on the number of external arcs
of any component in any intermediate network in 
$\Set{\N_{0}, \N_{1}, \ldots, \N_{m}}$.}

\paragraph{Complexity of $\CompPT$.} 
We ignore the effort to define the initial network $\N_0$ and then to
update it to $\N_1, \N_2, \ldots, \N_m$.  In fact, beyond the initial
$\N_0$, which we use to define the initial typing $T_0$, the
intermediate networks $\N_1,\N_2,\ldots,\N_m$ play no role in the
computation of the final typing $T = T_m$ returned by $\CompPT$. We
included the intermediate networks in the algorithm for clarity and to
make explicit the correspondence between $T_k$ and $\N_k$ for 
every $1\leqslant k\leqslant m$ (which is
used in the proof of
Lemma~\ref{thm:principal-typing-in-basic-compositional}).

The run-time complexity of $\CompPT$ is dominated by the computation
of type assignments (involving arithmetical additions, subtractions,
and comparisons), not by the bookkeeping steps.  Let $\delta
= \Index{\sigma}$.

There are at most $n\cdot 2^{\delta}$ type assignments in $T_0$ in
line~\ref{step:basiccomp2}.  In the \textbf{for}-loop from
line~\ref{step:basiccomp3} to line~\ref{step:basiccomp33}, $\CompPT$
calls $\BindT$ once in each of $m$ iterations, for a total of $m$
calls. Each such call to $\BindT$ runs in time 
$\bigO{\delta\cdot 2^{2\delta}}  = 2^{\bigOO{\delta}}$. 
Hence, the run-time complexity of
$\CompPT$ is:
\[
    \bigO{n\cdot 2^{\delta} + m\cdot \delta\cdot 2^{2\delta}} = 
    \bigO{(m+n)\cdot \delta\cdot 2^{2\delta}} = (m+n)\cdot 2^{\bigOO{\delta}}.
\]

\vspace{-.2in}
\section{Appendix: Proofs and Supporting Lemmas  
         for \hyperref[sect:max+min-flows-in-planar-networks]%
         {Section~\ref*{sect:max+min-flows-in-planar-networks}}}
  \label{sect:appendix:planar}

We need a few definitions and several technical result 
before we can prove Theorem~\ref{fact:principal-typing-of-planar} and
Lemma~\ref{lem:from-arbitrary-to-3-regular}
in Section~\ref{sect:max+min-flows-in-planar-networks}.

\begin{definition}{Onion-Peel Arcs and Cross Arcs}
\label{def:peeling-and-cross}
Let $\N = \N_1 = (\nn,\aaa)$ be a planar network, given with a
specific planar embedding. We define $L_1$ as the set of nodes
incident to $\OutF{\N_1}$, and define $L_i$ for $i>1$ recursively
as the set of nodes incident to $\OutF{\N_i}$, where $\N_i$ is
the planar embedding obtained after deleting all the nodes in
$L_1\cup\cdots\cup L_{i-1}$ and all the arcs incident to them.

We call $L_i$, for $i\geqslant 1$, the $i$-th \emph{onion peel}, or
the $i$-th \emph{peeling}, of the given planar embedding of $\N$.  If
the outerplanarity of the planar embedding is $k$, then there are $k$
non-empty peelings. We pose $\N_{k+1} = \varnothing$, the empty
network obtained after deleting the $k$-th and last non-empty peeling
$L_k$.

We call an arc $a$ which is bounding $\OutF{\N_i}$ an 
\emph{arc of level-$i$ peeling}, or also a \emph{level-$i$ peeling arc}. 
If we ignore the level of $a$, we simply say $a$ is a \emph{peeling arc}.
The two endpoints of $a$ are necessarily two distinct nodes in $L_i$.

All the arcs of $\N$ which are not peeling arcs, and which are not
input/output arcs, are called
\emph{cross arcs}. If $a$ is a cross arc with endpoints 
$\Set{\nu,\nu'} = \Set{\head{a},\tail{a}}$, then there are
one of two cases:
\begin{itemize}[itemsep=1pt,parsep=2pt,topsep=2pt,partopsep=0pt] 
\item \emph{either} there are two consecutive peelings $L_i$ and
      $L_{i+1}$, with $1\leqslant i < k$, such that $\nu\in L_i$
      and $\nu'\in L_{i+1}$,
\item \emph{or} there is a peeling $L_i$, with $1\leqslant i \leqslant k$,
      such that $\nu,\nu'\in L_i$ and $a$ is not bounding $\OutF{\N_i}$.
\end{itemize}
In either case, a cross arc $a$ bounds two adjacent
inner faces of $\N_i$ and is one
of the arcs to be deleted when we define $\N_{i+i}$ from $\N_i$.

We thus have a classification of all arcs in a
$k$-outerplanar embedding of a planar network: (1) the peeling arcs,
which are further partioned into $k$ disjoint levels, (2) the cross
arcs, and (3) the input/output arcs.
\end{definition}

Definition~\ref{def:peeling-and-cross} is written for a planar
network $\N$, but it applies just as well to an undirected planar
graph. 
We can view networks as undirected finite graphs, ignoring directions 
of arcs and ignoring the presence of input arcs and output arcs.  
To make the distinction between networks and graphs
explicit, we switch from ``nodes'' and
``directed arcs'' (for networks) to ``vertices'' and ``undirected edges''
(for graphs).

We write $G = (V(G),E(G))$ to denote an undirected simple (no
self-loops, no multiple edges) graph $G$, whose set of vertices is
$V(G)$ and set of edges is $E(G) \subseteq V(G)\times V(G)$.  We write
$\Angles{v,w}$ for an edge whose endpoints are the vertices $v$ and
$w$, which is the same as $\Angles{w,v}$, ignoring the direction in
the two ordered pairs. 

\begin{lemma}
\label{lem:outerplanarity-increased-by-factor-of-atmost-2}
Let $G$ be a simple planar graph, given with a fixed planar
embedding. Let $G'$ be obtained from $G$ by the following operation, 
for every vertex $v$ of degree $\geqslant 4$:
\begin{enumerate}[itemsep=1pt,parsep=2pt,topsep=2pt,partopsep=0pt] 
\item If {$\degree{v} = d \geqslant 4$} and $\Set{e_1,\ldots,e_d}$
      are the $d$ edges incident to $v$, we replace vertex $v$ by a
      simple cycle with $d$ fresh vertices $\Set{v_1,\ldots,v_d}$.
\item For every $1\leqslant i\leqslant d$,
      we replace vertex $v$ by $v_i$ as one of the two
      endpoints of edge $e_i$. 
\end{enumerate}
\textbf{Conclusion}: If the planar embedding of $G$ has outerplanarity
$k$, the resulting $G'$ is a planar graph with a planar
embedding of outerplanarity $k'\leqslant 2k$ and
where every vertex has degree $\leqslant 3$.
\end{lemma}

\begin{proof}
The proof is by induction on the outerplanarity $k\geqslant 1$.
We omit the straightforward proof for the case $k = 1$: The construction
of $G'$ produces a planar embedding with outerplanarity $\leqslant 2$
and where every vertex has degree $\leqslant 3$. To be more specific, if
$G$ has an inner face $F$, a vertex $v$ on the boundary of $F$ with 
$\degree{v}\geqslant 4$, and an edge $\Angles{v,w}$ not contained in
$\OutF{G}$, then the new $G'$ has outerplanarity $2$. Otherwise,
if this condition is not satisfied, $G'$ has outerplanarity $1$.

Proceeding inductively, the \emph{induction hypothesis} assumes that,
given an arbitrary planar $G$ with a planar embedding of outerplanarity
$k\geqslant 1$, the transformation described in the lemma statement
produces a planar $G'$ with a planar embedding of outerplanarity
$k'\leqslant 2k$ and where every vertex has degree $\leqslant 3$.

We prove the result again for an arbitrary planar graph $G$ with a
planar embedding of outerplanarity $k+1$. For the rest of the proof, we
uniquely label every edge with a positive
integer. Thus, we denote an edge $e$ by a triple $\Angles{v,w,\ell}$
where $v$ and $w$ are distinct vertices and $\ell\in\posnats$.
We also introduce an additional set of vertices, which we
call \emph{hooks}. For every $v\in V(G)$, the set of hooks
associated with $v$ is:
\[
   \hook{*}{v}\ :=\ \Set{\,\hook{\ell}{v}\;|
   \; \text{there is $w\in V(G)$ and $\ell\in\posnats$
            such that $\Angles{v,w,\ell}\in E(G)$}\,} ,
\]
\ie, there are as many hooks associated with $v$ as
there are edges incident to $v$. The set of all the hooks in $G$ is
$\hook{*}{G} := \bigcup \Set{\hook{*}{v}\,|\,v\in V(G)}$.

Let $P$ and $Q$ be the first and second onion-peels of $G$,
respectively. $P$ and $Q$ are disjoint subsets of vertices.  Although
``$\Angles{v,w,\ell}$'' and ``$\Angles{w,v,\ell}$'' denote the same
edge, it will be clearer to write $\Angles{v,w,\ell}$ instead of
$\Angles{w,v,\ell}$ whenever $v\in P$ and $w\in Q$. We define two
graphs $G_1$ and $G_2$ from the $(k+1)$-outerplanar $G$:
\begin{alignat*}{8}
 & V(G_1)\ &&:=\ && P\ \cup\ \Set{\,\hook{\ell}{w}\;|
   \;\text{$w\in Q$, $\ell\in\posnats$, and there is $v\in P$  
           such that $\Angles{v,w,\ell}\in E(G)$}\,}
\\[1.5ex]
 & E(G_1)\ &&:=\ &&
   \Set{\,\Angles{v,w,\ell}\;|
      \;v,w\in P\text{ and }\Angles{v,w,\ell}\in E(G)\,}\ \cup
\\
 & && &&
   \Set{\,\Angles{v,\hook{\ell}{w},\ell}\;|
          \;\text{$v\in P$, $w\in Q$, and $\Angles{v,w,\ell}\in E(G)$}\,}
\\[1.5ex]
 & V(G_2) &&:= && (V(G) - P)\ \cup
    \ \Set{\,\hook{\ell}{v}\;|
    \;\text{$v\in P$, $\ell\in\posnats$, and there is $w\in Q$  
           such that $\Angles{v,w,\ell}\in E(G)$}\,}
\\[1.5ex]
 & E(G_2)\ &&:=\ &&
   \Set{\,\Angles{v,w,\ell}\;|
      \;v,w\in (V(G)-P)\text{ and }\Angles{v,w,\ell}\in E(G)\,}\ \cup
\\
 & && &&
   \Set{\,\Angles{\hook{\ell}{v},w,\ell}\;|
          \;\text{$v\in P$, $w\in Q$, and $\Angles{v,w,\ell}\in E(G)$}\,}
\end{alignat*}
Observe that every hook, \ie, a vertex of the form $\hook{\ell}{v}$
has degree $= 1$, and that every edge of the form
$\Angles{v,\hook{\ell}{w},\ell}$ or $\Angles{\hook{\ell}{v},w,\ell}$
is entirely contained in both $\OutF{G_1}$ and $\OutF{G_2}$. We need
to distinguish between \emph{open edges} and \emph{closed edges} of
$G_1$ and $G_2$. For $G_1$ first:
\begin{alignat*}{8}
  & E_{o}(G_1)\ &&:=\ && \Set{\,\Angles{v,\hook{\ell}{w},\ell}\;|
         \; v\in P,\ w\in Q,\ \ell\in\posnats\,} 
         \qquad &&\text{the \emph{open edges} of $G_1$}
\\[.8ex]
  & E_{c}(G_1)\ &&:=\ && E(G_1) - E_{o}(G_1)
         \qquad &&\text{the \emph{closed edges} of $G_1$}
\end{alignat*}
And similarly for $G_2$:
\begin{alignat*}{8}
  & E_{o}(G_2)\ &&:=\ && \Set{\,\Angles{\hook{\ell}{v},w,\ell}\;|
         \; v\in P,\ w\in Q,\ \ell\in\posnats\,} 
         \qquad &&\text{the \emph{open edges} of $G_2$}
\\[.8ex]
  & E_{c}(G_2)\ &&:=\ && E(G_2) - E_{o}(G_2)
         \qquad &&\text{the \emph{closed edges} of $G_2$}
\end{alignat*}
An \emph{open edge} is therefore an edge with a hook as one of its two
endpoints, which is always of degree $= 1$.
The graphs $G_1$ and $G_2$ are planar, and their definitions are such
that they produce a planar embedding of $G_1$ with outerplanarity $= 2$
and a planar embedding of $G_2$ with outerplanarity $= k$. These
assertions follow from the two facts below, together with
the fact that the presence of open edges drawn inward (as in $G_1$)
increases outerplanarity by $1$, and drawn outward (as in $G_2$)
does not increase outerplanarity:
\begin{itemize}[itemsep=1pt,parsep=2pt,topsep=2pt,partopsep=0pt] 
\item If we delete every open edge in $G_1$, we obtain the $1$-outerplanar
      subgraph of $G$ induced by $P$.
\item If we delete every open edge in $G_2$, we obtain the $k$-outerplanar
      subgraph of $G$ induced by $(V(G)-P)$.
\end{itemize}
An example of how $G$ is broken up into two graphs $G_1$ and $G_2$
is shown in Figure~\ref{fig:breaking-up}.
%

\begin{minipage}[c][13cm][c]{0.9\linewidth} 
\begin{figure}[H] 
\begin{center}
\includegraphics[scale=.55,trim=0cm 12.25cm 0cm 0.80cm,clip]%
    {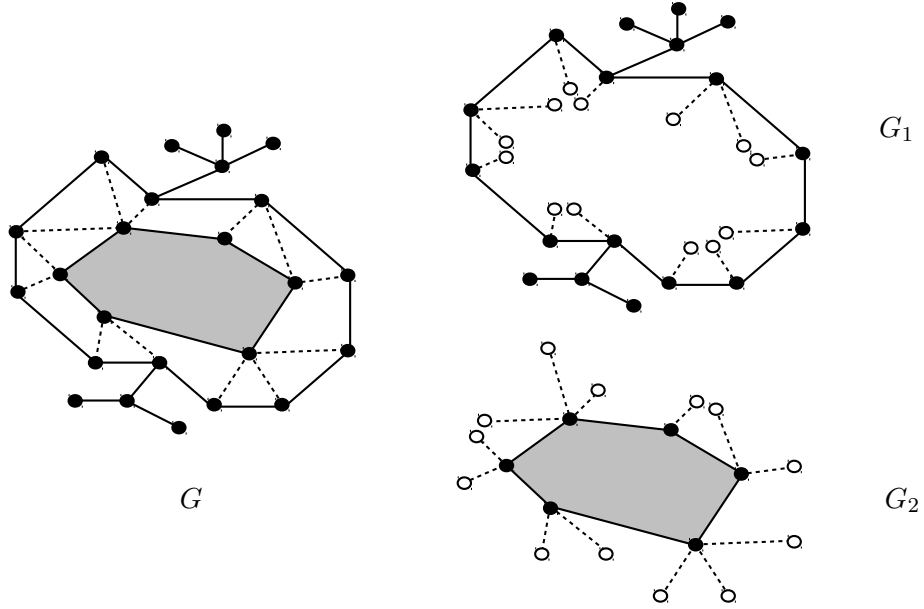}
\end{center}
\caption{Example for the proof of 
         Lemma~\ref{lem:outerplanarity-increased-by-factor-of-atmost-2}.
         The outermost edges in $G$ form its first peeling; the
         edges enclosing the shaded area form its second peeling;
         the shaded area, including the nodes of the second peeling, is
         a subgraph of $G$ of outerplanarity $k$.
         $G_1$ and $G_2$ are shown in planar embeddings of outerplanarity $2$
         and $k$, respectively. An \emph{open edge} (dashed line)
         in $G_1$ and $G_2$ has a \emph{hook} (white node) as one of
         its two endpoints. }
\label{fig:breaking-up}
\end{figure}

\ \vspace{-14cm}

\begin{minipage}[c][13cm][c]{0.9\linewidth} 
%
%
%
%
%
%
%
\hspace{13.8cm} $G_1$ \\
\vspace{4cm} 

\hspace{4.5cm} $G$\hspace{9cm} $G_2$
\end{minipage}
\end{minipage}

\bigskip
\noindent
We can ``re-build'' $G$ from $G_1$ and $G_2$ as follows:
\begin{alignat*}{8}
 & V(G)\ &&=\ \ && \bigl(V(G_1) \cup V(G_2)\bigr)\ -
    \ \bigl(\hook{*}{G_1} \cup \hook{*}{G_2}\bigr)
\\[.8ex]
  & E(G)\ &&= && E_c(G_1)\ \cup\ E_c(G_2)\ \cup
\\
  & && &&
    \Set{\,\Angles{v,w,\ell}\;|\;\Angles{v,\hook{\ell}{w},\ell}\in E_o(G_1)
    \text{ and } \Angles{\hook{\ell}{v},w,\ell}\in E_o(G_2)\,}
\end{alignat*}
The preceding are not definitions of $V(G)$ and $E(G)$, but rather 
equalities that are easily checked against the earlier
definitions. They make explicit the way in which we use hooks and open
edges to connect two graphs.

\medskip
Let $G_1'$ be the graph obtained from $G_1$ according to the
transformation defined in the lemma statement, which produces a planar
embedding of $G_1'$ with outerplanarity $\leqslant 2$.  And let $G_2'$
be the graph obtained from $G_2$ according to the transformation
defined in the lemma statement, which, by the \emph{induction
hypothesis}, produces a planar embedding of $G_2'$ with outerplanarity
$\leqslant 2k$. 

We note carefully how open edges in $G_1$ may get transformed into open
edges in $G_1'$. Consider a vertex $v\in P$ with $\degree{v} =
d \geqslant 4$ and let the open edges of $G_1$ that have
$v$ as one of their two endpoints be:
\[
   \Angles{v,\hook{\ell_1}{w_1},\ell_1}\ ,\ \ \ldots\ \ ,
   \ \Angles{v,\hook{\ell_t}{w_t},\ell_t}
\]
where $1\leqslant t\leqslant d$. The number $t$ of open
edges with endpoint $v$ is not necessarily $d$. The corresponding open
edges in $G_1'$ are:
\[
   \Angles{v_{i_1},\hook{\ell_1}{w_1},\ell_1} ,\ \ldots\ ,
   \Angles{v_{i_t},\hook{\ell_t}{w_t},\ell_t}
\]
where $\Set{v_{i_1},\ldots,v_{i_t}}\subseteq
\Set{v_1,\ldots,v_d}$, and $\Set{v_1,\ldots,v_d}$ is the set
of fresh vertices in the simple cycle that replaces $v$ in $G'_1$.
Note that the transformation from $G_1$ to $G'_1$ does not affect
the second endpoints (the hooks) of these open edges, because the
degree of a hook is always $= 1$. Similar observations apply to the
way in which open edges in $G_2$ get transformed into open
edges in $G_2'$.

We are ready to define the desired graph $G'$ 
by connecting $G_1'$ and $G_2'$ via their hooks and open edges,
in the same way in which we can re-connect $G$ from $G_1$ and $G_2$:
\begin{alignat*}{8}
 & V(G')\ &&=\ \ && \bigl(V(G'_1)\ \ \cup\ &&V(G'_2)\bigr)\ -
    \ \bigl(\hook{*}{G'_1} \cup \hook{*}{G'_2}\bigr)
\\[1ex]
  & E(G')\ &&= && E_c(G'_1)\ \ \cup\ &&E_c(G'_2)\ \ \cup
\\
  & && &&
    \{\,\Angles{v_i,w_j,\ell}\; &&|
         \;\text{there are $v\in P$ and $w\in Q$ such that}
\\
  & && && &&\;\;\; \Angles{v_i,\hook{\ell}{w},\ell}\in E_o(G'_1)\text{ and }
    \Angles{\hook{\ell}{v},w_j,\ell}\in E_o(G'_2),
\\
  & && && &&\;\;\; \text{where $1\leqslant i\leqslant\degree{v}$
              and $1\leqslant j\leqslant\degree{w}$}\,\}
\end{alignat*}
This produces a planar graph $G'$ together with a planar embedding. 
To conclude the induction and the proof, it suffices to note that
the outerplanarity of $G'$ is ``the outerplanarity of $G_1$'' +
``the outerplanarity of $G_2$'' which is therefore $\leqslant 2k$.
\end{proof}

\begin{Proof}{for Lemma~\ref{lem:from-arbitrary-to-3-regular}} 
The computation is a little easier if we introduce a new node
$\nu_{\text{in,out}}$ and connect the tail of every input arc
$a\in\aaa_{\text{in}}$ to $\nu_{\text{in,out}}$,
\ie, $\tail{a} = \nu_{\text{in,out}}$, and the head of every output arc
$b\in\aaa_{\text{out}}$ to $\nu_{\text{in,out}}$,
\ie, $\head{b} = \nu_{\text{in,out}}$. In the resulting network,
there are no input arcs and no output arcs, with $r = n+1$ nodes.
The number $m$ of arcs remains unchanged, and they are now all
internal arcs. 
With no loss of generality, we assume for every $\nu\in\nn$:
\begin{itemize}[itemsep=1pt,parsep=2pt,topsep=2pt,partopsep=0pt] 
\item[(1)]  $\degree{\nu} \geqslant 3$, 
\item[(2)]  $\indegree{\nu} \neq 0 \neq \outdegree{\nu}$.
\end{itemize}
It is easy to see that every node $\nu$ violating one of the preceding
assumptions can be eliminated from $\N$. Also, with no loss of
generality, assume that $\size{\aaa_{\text{in}}}
+ \size{\aaa_{\text{out}}} = p+q \geqslant 3$, so that 
$\degree{\nu_{\text{in,out}}} \geqslant 3$. We also assume:
\begin{itemize}[itemsep=1pt,parsep=2pt,topsep=2pt,partopsep=0pt] 
\item[(3)]  the outerplanarity $k$ of $\N$ is $\geqslant 2$.
\end{itemize}
The case $k = 1$ is handled similarly, with few minor
adjustments (which, in fact, makes it easier).

\medskip
The construction of $\N'$ from $\N$ proceeds in two parts,
with each taking time $\bigOO{n}$.
The first part consists in eliminating all two-node cycles.
Let $\gamma$ be a two-node cycle, \ie, there are
two arcs $a$ and $a'$ such that:
\[
   \tail{a} = \head{a'} = \nu \quad
   \text{\ and\ }             \quad
   \head{a} = \tail{a'} = \nu'.
\]
To eliminate $\gamma$ as a two-node cycle, we insert a new
node $\mu$ in the middle of $a$, another new node $\mu'$ in the middle
of $a'$, and add a new arc $b = \Angles{\mu,\mu'}$. We make the new
arc $b$ \emph{dummy}, by setting $\lc(b) = \uc(b) = 0$, which prevents it 
from carrying any flow. Clearly:
\begin{itemize}[itemsep=1pt,parsep=2pt,topsep=2pt,partopsep=0pt] 
\item $\degree{\mu} = \degree{\mu'} = 3$.
\item If the two arcs $a$ and $a'$ do not enclose an inner face of $\N$,
      one of the two can be redrawn so that after inserting the new
      arc $b$, planarity is preserved.
\item If $k\geqslant 2$, it is easy to see that
      the outerplanarity remains the same. (This is the reason
      for assumption (3).)
\end{itemize}
This operation can be extended to all two-node
cycles in $\N$ in time $\bigOO{n}$, resulting in an equivalent
network, with fewer than $m$ new nodes and fewer than $\lceil m/2\rceil$
new arcs, and where the degree of every node $\geqslant 3$.

\medskip
The second part of the construction consists in replacing every node
$\nu\in\nn$ with $\degree{\nu} = s \geqslant 4$ by an appropriate
cycle with $s$ new nodes, say $\Set{\nu_1,\ldots,\nu_s}$, to obtain
a network satisfying conclusions 2, 3, 4, and 5.
More specifically, consider the arcs incident to $\nu$, say:
\[
   \Set{a_1,\ldots,a_s}\ := 
   \ \Set{\,a\in\aaa\;|\; \head{a} = \nu\text{ or } \tail{a} = \nu\,} .
\]
We introduce $s$ new arcs, say $\Set{b_1,\ldots,b_s}$, to form a
directed cycle connecting the new nodes $\Set{\nu_1,\ldots,\nu_s}$.
We make each new node $\nu_i$ the endpoint of an arc in
$\Set{a_1,\ldots,a_s}$, \ie, for every $1\leqslant i\leqslant s$:
\begin{itemize}
\item if $\head{a_i} = \nu$, we set $\head{a_i} := \nu_i$,
\item if $\tail{a_i} = \nu$, we set $\tail{a_i} := \nu_i$.
\end{itemize}
An example of the transformation from the node $\nu$ to the
directed cycle replacing it is shown in Figure~\ref{fig:one-node-expanded}.

\begin{figure}[t] 
%
\begin{minipage}[b][3.70cm][t]{0.4\linewidth}
\begin{center}
\includegraphics[scale=.25,trim=0cm 13.00cm 6cm 0.0cm,clip]%
    {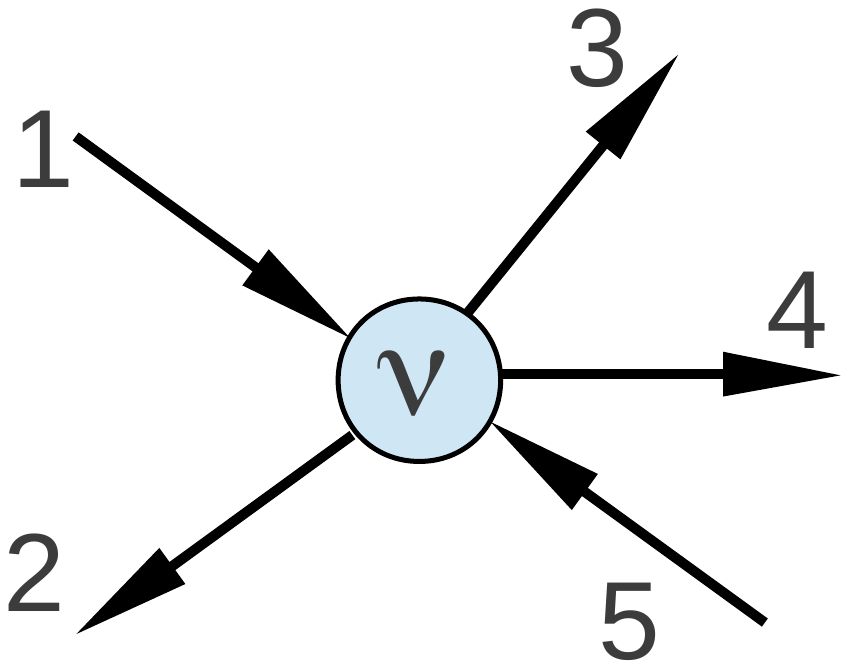}
\end{center}
\end{minipage}
\begin{minipage}[b][3.70cm][t]{0.4\linewidth}
\begin{center}
\includegraphics[scale=.25,trim=2cm 12.0cm 0cm 1.0cm,clip]%
    {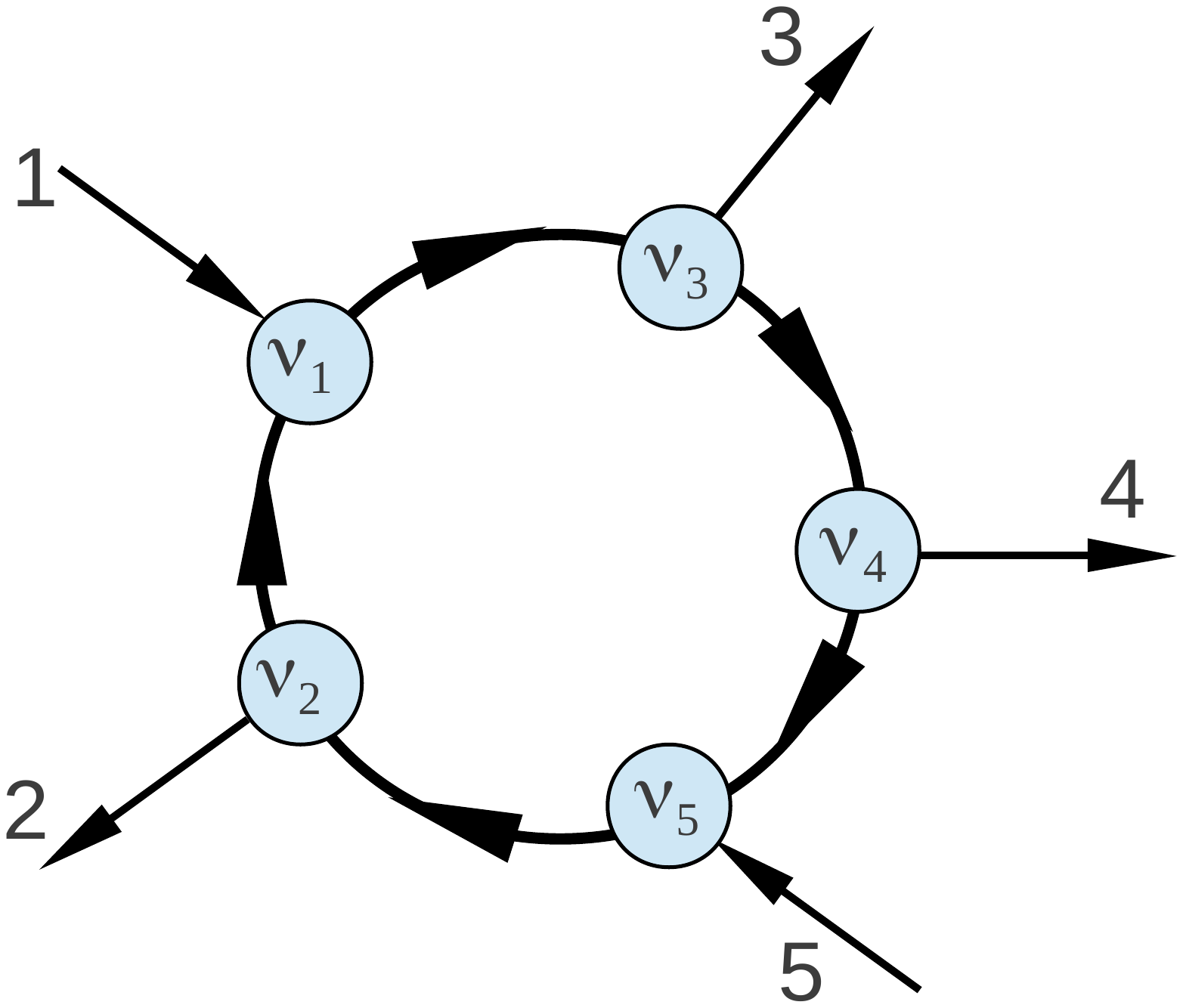}
\end{center}
\end{minipage}
\caption{A node $\nu$ of degree $= 5$ (on the left) is transformed into
     a cycle with $5$ nodes $\nu_1, \nu_2, \nu_3, \nu_4,$
     and $\nu_5$, 
     (on the right), each of degree $= 3$. The transformation preserves
     planarity.}
\label{fig:one-node-expanded}
\end{figure}

We want the cycle connecting the new nodes $\Set{\nu_1,\ldots,\nu_s}$
to put no restriction on flows, so we set all the lower
bounds to $0$ and all the upper bounds to the ``very large number'' $K$:
\[
  \lc(b_1) := \cdots := \lc(b_s) := 0\quad \text{ and }\quad
  \uc(b_1) := \cdots := \uc(b_s) := K .
\]
Repeating the preceding operation for every node $\nu\in\nn$,
it is straightforward to check that conclusions 1, 2, and 3,
in the statement of Lemma~\ref{lem:from-arbitrary-to-3-regular}
are satisfied, and the construction
can be carried out in time $\bigOO{n}$.

\medskip
For conclusion 4, consider an arc $a\in\aaa$.  If $a$ is incident to
one node $\nu$ of degree $\geqslant 4$ (resp. two nodes $\nu$ and
$\mu$ of degree $\geqslant 4$), then the preceding construction
introduces one new arc corresponding to $a$ in the cycle simulating
$\nu$ in $\N'$ (resp. two new arcs corresponding to $a$, one in the
cycle simulating $\nu$ and one in the cycle simulating $\mu$, in
$\N'$). Hence, the number $m'$ of arcs in $\N'$ is such that
$m'\leqslant 3m$.

Moreover, because every arc is incident to two distinct
nodes \emph{and} for every node $\nu$ in $\N'$ there are exactly three
arcs incident to $\nu$, the number $n'$ of nodes in $\N'$ is such that
$n' = 2m'/3 \leqslant 2 m$.

\medskip
It remains to prove conlusion 5. The preceding construction preserves
planarity: If $\N$ is given with a planar embedding, the new $\N'$
is produced with a planar embedding. 

The standard notion of an undirected graph $G$ does not include
the presence of one-ended edges corresponding to input/output arcs in
a network $\N$. In order to turn the input/output arcs of network $\N$ into
two-ended edges in graph $G$ we simply add a new node of degree $= 1$
at the end of every input/output arc missing a node.

With the preceding qualification, we can
take $G$ to be the undirected simple graph corresponding to network
$\N$ \emph{after elimination
of all two-node directed cycles}. This is the first part in the
two-part construction of $\N'$ from $\N$. Absence of two-node cycles
allows us to take $G$ as a simple graph (no multiple edges). 
Every vertex in $G$ has thus degree $1$ or degree $\geqslant 3$, as we assume
that all nodes in $\N$ have degrees $\geqslant 3$.

We take $G'$ to be the undirected simple graph corresponding to
network $\N'$ after the second part in the construction. 
The construction of $G'$ from $G$ in
Lemma~\ref{lem:outerplanarity-increased-by-factor-of-atmost-2}
corresponds to the construction of $\N'$ from $\N$ after elimination
of all two-node cycles. The conclusion of 
Lemma~\ref{lem:outerplanarity-increased-by-factor-of-atmost-2}
implies conclusion 5 in
Lemma~\ref{lem:from-arbitrary-to-3-regular}. \hfill \QED
\end{Proof}

\begin{example}
\label{ex:running-example1}
Consider the planar network $\N$ on the left in 
Figure~\ref{fig:thirty-node-network}. We transform $\N$ into
a $3$-regular network $\N'$ according to the construction
in the proof of Lemma~\ref{lem:from-arbitrary-to-3-regular}.
We omit all arc directions
in $\N$ which play no role in the transformation. 

The outerplanarity of $\N$ is $3$: There are three enclosing peelings
(drawn with foldface arcs on the left in Figure~\ref{fig:thirty-node-network}). 
The outerplanarity of $\N'$ is guaranteed not to exceed $6$ by
Lemma~\ref{lem:outerplanarity-increased-by-factor-of-atmost-2}, but
is also $3$ in this example, as one can easily check. 

On the right in
Figure~\ref{fig:thirty-node-network-on-rectangular-grid}, $\N'$ is
re-drawn on a rectangular grid -- except for two arcs because of
a missing north-east corner and a missing south-west corner --
which makes its outerplanarity ($= 3$)
explicit and reasoning in the proof of 
Theorem~\ref{fact:principal-typing-of-planar} easier to follow. This
re-drawing can always be done in linear time~\cite{nishizeki2013}.
\end{example}

\begin{assumption}
\label{ass:3-regular-connected}
From now on, there is no loss of generality if we assume that:
\begin{enumerate}[itemsep=1pt,parsep=2pt,topsep=2pt,partopsep=0pt] 
\item Networks are connected.
\item Networks are $3$-regular.
\item There are no two-node cycles in networks.
\item No two distinct input/output arcs are incident to the
      same node.
\end{enumerate}
The second and third conditions follow from the construction
in the proof of Lemma~\ref{lem:from-arbitrary-to-3-regular}.
The fourth condition is equivalent to saying that a node cannot
be both a source and a sink.
\end{assumption}

The next definition, and lemma based on it, are not essential. But,
together with the preceding assumption, they simplify considerably
Algorithm~\ref{alg:bindSchedule} and proving its correctness.

\begin{definition}{Good Planar Embeddings}
\label{def:full-plane-embeddings}
\label{def:balanced-planar-embeddings}
\label{def:good-planar-embeddings}
Let $\N = (\nn,\aaa)$ be a network satisfying
Assumption~\ref{ass:3-regular-connected} and given
in a fixed $k$-outerplanar embedding, for some $k\geqslant 1$. 
From the peelings $L_1,\ldots,L_k$ specified in 
Definition~\ref{def:peeling-and-cross}, we define the
sets of nodes $L'_1,\ldots,L'_k$, respectively, as follows.
For every $1\leqslant i\leqslant k$:
\[
    L_i' := L_i - \Set{\,\nu\in L_i\;|
    \;\text{$\nu$ is incident to at most one peeling arc}\,} .
\]
In words, $L_i'$ is a subset of $L_i$ which is proper whenever
$L_i$ contains a node $\nu$ such that:
\begin{enumerate}[itemsep=1pt,parsep=2pt,topsep=2pt,partopsep=0pt] 
\item $\nu$ is incident to three cross arcs. 
\item $\nu$ is incident to two cross arcs and one input/output arc.
\item $\nu$ is incident to one cross arc and one input/output arc.
\end{enumerate}
Thus, $L_i'$ is defined to exclude all the nodes of 
$L_i$ that are of degree $\leqslant 1$ in
the network $\N_i$ (see Definition~\ref{def:peeling-and-cross}).

We say the planar embedding of $\N$ is \emph{good} if for every 
$1\leqslant i\leqslant k$, the nodes in $L_i'$  
form a single (undirected) simple cycle, namely, the outermost one, in the 
network $\N_i$. 
\end{definition}

\begin{example}
\label{ex:running-example2}
For an example of how $L_i'$ may be different from $L_i$, consider the
$3$-outerplanar embedding in
Figure~\ref{fig:thirty-node-network-on-rectangular-grid}: $L_1 = L_1'$
and $L_3 = L_3'$, but $L_2 \neq L_2'$. The latter inequality is caused
by one of the nodes on the periphery of the south-east face, which is
incident to one cross arc and one input/output arc. 

In a good planar embedding, the sets $L'_1,\ldots,L'_k$
can be viewed as forming $k$ concentric simple cycles.
All the nodes in $(L_i - L_i')$ occur between the level-$i$
concentric cycle and the one immediately enclosing it (the
level-$(i-1)$ concentric cycle). This implies that, if
$\N$ is $3$-regular and no two distinct input/output arcs
are incident to the same node (as required by 
Assumption~\ref{ass:3-regular-connected}), then $L_1 = L_1'$
but we may have $L_i \neq L_i'$ for $i\geqslant 2$.
The $3$-outerplanar embedding in 
Figure~\ref{fig:thirty-node-network-on-rectangular-grid} is good.
\end{example}

\begin{lemma}
[From Planar Embeddings to Good Planar Embeddings]
\label{prop:good-planar-embeddings}
\label{lem:good-embeddings}
Let $\N = (\nn,\aaa)$ be a network satisfying 
Assumption~\ref{ass:3-regular-connected} and  
given in a specific planar embedding. 
In time $\bigOO{n}$, where $n = \size{\nn}$, we can
transform the given planar embedding of $\N$ into a planar
embedding of an equivalent $\N' = (\nn',\aaa')$ such
that:
\begin{enumerate}[itemsep=1pt,parsep=2pt,topsep=2pt,partopsep=0pt] 
\item The planar embedding of $\N'$ is good
      (and, in particular, $\N'$ satisfies 
       Assumption~\ref{ass:3-regular-connected}).
\item $\size{\nn'} \leqslant 
      2\cdot \size{\nn}$ and $\size{\aaa'} \leqslant 
      2\cdot \size{\aaa}$.
\item $\N$ and $\N'$ have the same outerplanarity.
\end{enumerate}
\end{lemma}

\begin{proof}
Straightforward, by appropriately inserting \emph{dummy arcs}, also
making sure not to violate $3$-regularity and not to increase
outerplanarity.  An arc $a$ is dummy arc if $\lc(a) = \uc(a) =
0$, \ie, $a$ cannot carry any flow and therefore cannot affect the
overall flow properties of the network.
\end{proof}

\begin{lemma}
[Paths of Nodes in $(L_i - L_i')$]
\label{lem:paths-in-good-embeddings}
Let network $\N$ be given in a good $k$-outerplanar embedding,
for some $k\geqslant 1$,
which satisfies in particular Assumption~\ref{ass:3-regular-connected}.
Suppose $(L_i - L_i')\neq\varnothing$ for some $1\leqslant i\leqslant k$
and consider a maximal-length undirected path 
$\Angles{\nu_0,\nu_1,\ldots,\nu_{p-1}}$ formed by node $\nu_0\in L_i'$ and
nodes $\Set{\nu_1,\ldots,\nu_{p-1}}\subseteq (L_i - L_i')$ for some
$p\geqslant 2$.

\medskip
\noindent
\textbf{Conclusion}: There is a node $\nu_{p}\in L'_{i-1}$,
together with  $(p-1)$ level-$i$ peeling arcs
$\Set{a_1,\ldots,a_{p-1}}$, one cross arc $a_p$, and
$(p-1)$ input/output arcs $\Set{b_1,\ldots,b_{p-1}}$ of $\N$, such that:
\begin{enumerate}[itemsep=2pt,parsep=2pt,topsep=5pt,partopsep=0pt] 
\item $\Set{\head{a_1},\tail{a_1}} = \Set{\nu_{0},\nu_{1}},\ \ldots\ ,
      \Set{\head{a_{p-1}},\tail{a_{p-1}}} = \Set{\nu_{p-2},\nu_{p-1}}$ \\[.8ex]
      and $\Set{\head{a_p},\tail{a_p}} = \Set{\nu_{p-1},\nu_{p}}$.
\item For every $1\leqslant j\leqslant p-1$,
      either $\head{b_j}$ is defined and  $\head{b_j}=\nu_j$
      or $\tail{b_j}$ is defined and  $\tail{b_j}=\nu_j$.
\end{enumerate}
In words, the undirected path $\Angles{\nu_0,\nu_1,\ldots,\nu_{p-1},\nu_p}$,
which is the same as $\Angles{a_1,\ldots,a_p}$ as a sequence of
internal arcs, connects node $\nu_0\in L_i'$ and node $\nu_{p}\in L'_{i-1}$,
with $(p-1)$ input/output arcs incident to the $(p-1)$ intermediate
nodes along this path.
\end{lemma}

\begin{proof}
Straightforward from the definitions. All details omitted.
\end{proof}


\begin{figure}[t] 
%
\begin{minipage}[b][7.40cm][t]{0.95\linewidth}
   \begin{center}
      \begin{custommargins}{-1.5cm}{-2.5cm}
      \includegraphics[scale=.4,trim=0cm 9.00cm 0cm 0.60cm,clip]%
         {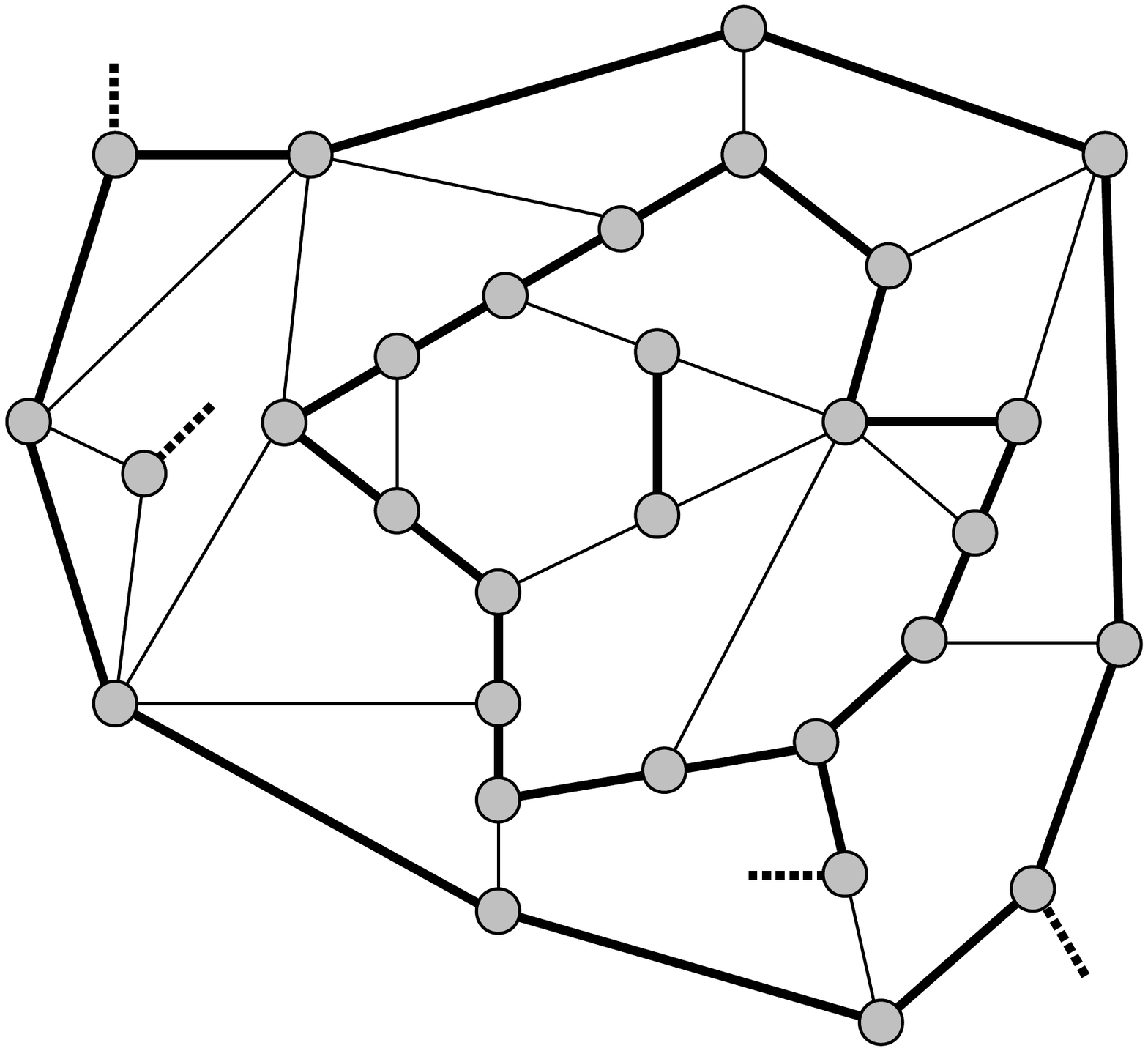}
      \includegraphics[scale=.4,trim=0cm 9.0cm 0cm 0.60cm,clip]%
         {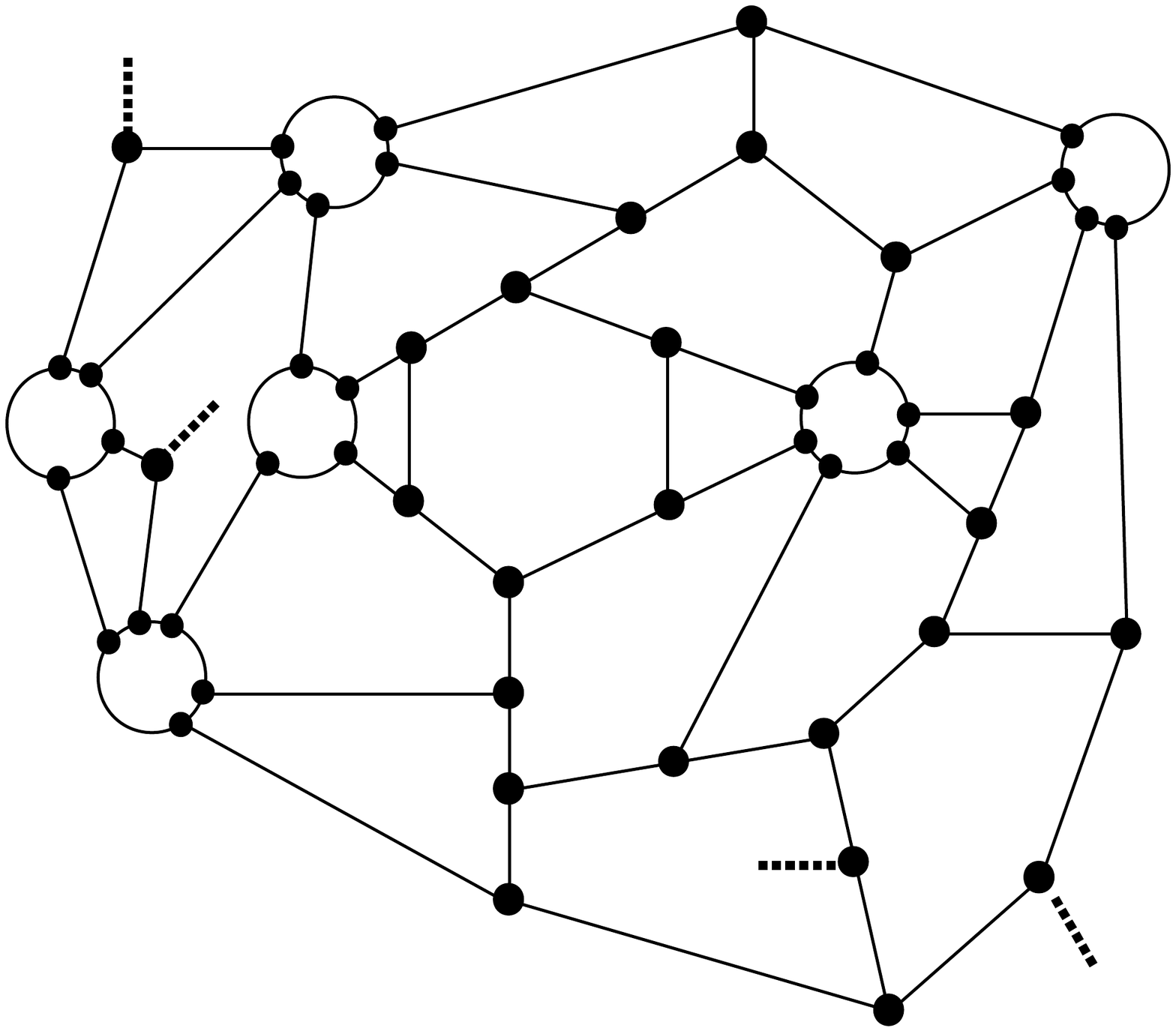}
      \end{custommargins}
    \end{center}
\end{minipage}
\caption{Example of a planar network $\N$
         (with all arc directions ignored) on the left,
         its transformation into a 3-regular network $\N'$
         according to Lemma~\ref{lem:from-arbitrary-to-3-regular}
         on the right. The dashed arcs are input/output arcs,
         $4$ of them. }
\label{fig:thirty-node-network}
\end{figure}

\begin{figure}[t] 
%
\begin{minipage}[b][7.40cm][t]{0.95\linewidth}
   \begin{center}
      \begin{custommargins}{-1.5cm}{-2.5cm}
      \includegraphics[scale=.4,trim=0cm 9.0cm 0cm 0.60cm,clip]%
         {./Graphics/thirty-node-networkA}
      \includegraphics[scale=.4,trim=0cm 7.00cm 0cm 3.60cm,clip]%
         {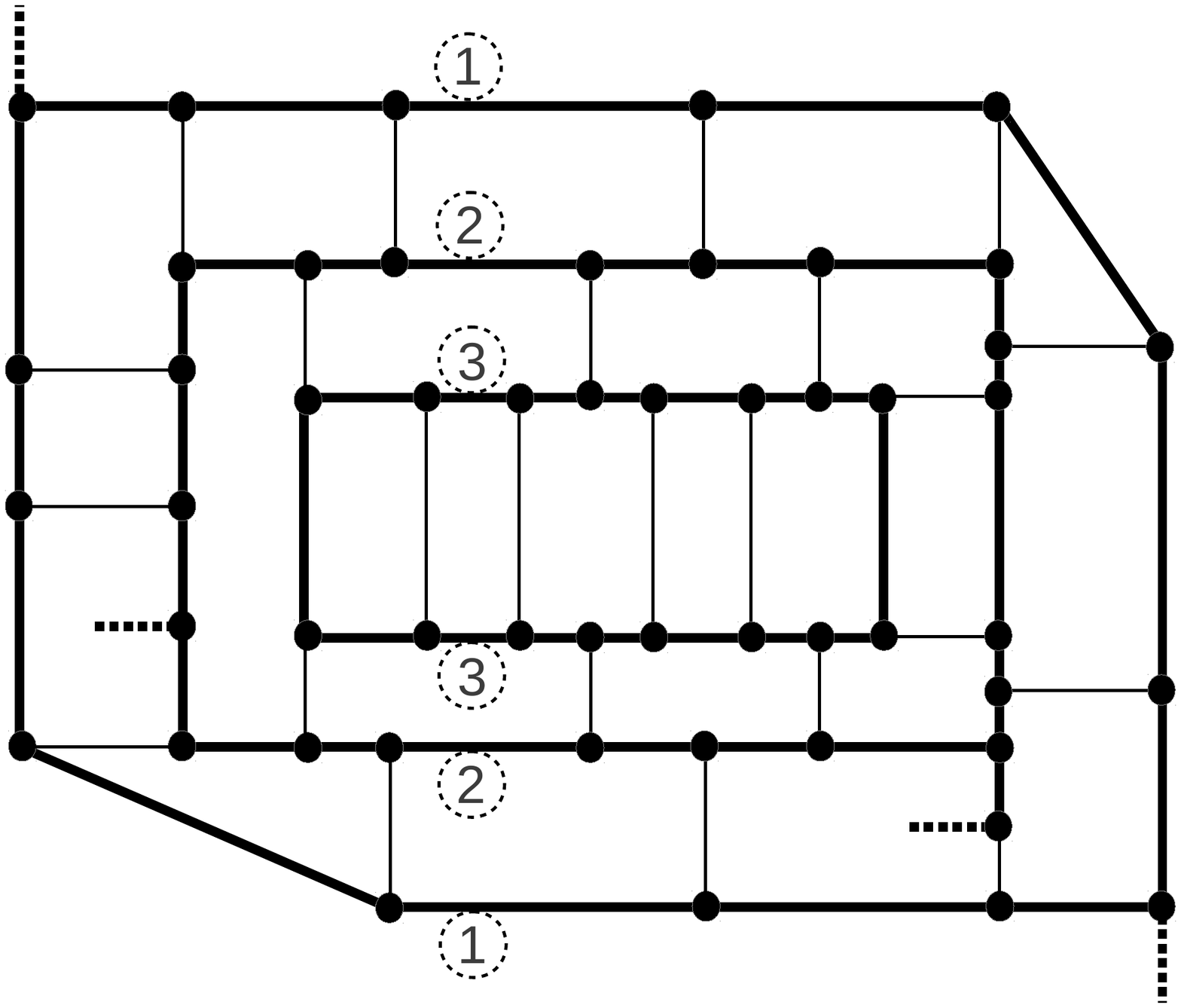}
      \end{custommargins}
    \end{center}
\end{minipage}
\caption{The planar network $\N'$ on the right in 
         Figure~\ref{fig:thirty-node-network}
         is reproduced without change on the left in this figure, 
         and re-drawn on a rectangular grid  
         on the right in this figure -- except for two arcs because
         of the missing north-east corner and
         south-west corner.
         The dashed arcs are input/output arcs,
         $4$ of them.}
\label{fig:thirty-node-network-on-rectangular-grid}
\end{figure}


\begin{remark}
In Section~\ref{sect:disassemble-and-reassemble}
and Appendix~\ref{sect:appendix:disassemble},
where we had to worry about finding typings of subnetworks of the 
given network $\N = (\nn,\aaa)$, the same arc $a\in\aaa$ could be an 
input arc in a subnetwork $\M'$ and an output arc in another subnetwork 
$\M''$, in which case we temporarily re-named it $a^{+}$ in $\M'$ and $a^{-}$
in $\M''$. In this appendix, there is no issue about computing 
typings and we can do away with this distinction between ``$a$ as input arc'' 
and ``$a$ as output arc'' in two distinct subnetworks.

Moreover, whenever convenient, we ignore arc directions.  Our concern
is to minimize the number of interface links, as we break up and re-assemble
networks, and this is not affected by arc directions.
\end{remark}

Recall the definitions of ``subnetwork'' and ``component'' of
a network $\N$ in Section~\ref{sect:sound-and-complete}, which mimic
the standard defnitions of ``subgraph'' and ``component'' except for
the presence of input/output arcs. For a precise statement of
Algorithm~\ref{alg:bindSchedule}, we need the notions of
``neighbor subnetworks'' and how to ``merge'' them.

\begin{definition}{Neighbor Subnetworks and their Merge}
\label{def:neighbors}
Let $\N = (\nn,\aaa)$ be a network, and consider two non-empty
disjoint subsets of nodes: $X, X'\subseteq\nn$ with $X\cap X'
= \varnothing$.  Let $\M$ and $\M'$ be the subnetworks of $\N$
induced by $X$ and $X'$, respectively. Let $\bbb_{\text{in,out}}$
and $\bbb'_{\text{in,out}}$ be the input/output arcs of $\M$ and $\M'$.
We say $\M$ and $\M'$ are \emph{neighbor subnetworks}, or
just \emph{neighbors}, iff 
$\size{\bbb_{\text{in,out}}\cap\bbb'_{\text{in,out}}} \geqslant 1$, \ie,
$\M$ and $\M'$ have one input/output arc or more in common.
We refer to the sequence of arcs
in $\bbb_{\text{in,out}}\cap\bbb'_{\text{in,out}}$ listed according
to some fixed (but otherwise arbitrary) ordering scheme as
the \emph{sequence of joint arcs} of $\M$ and $\M'$:
\[
   \commonA{\M,\M'}\ :=
   \ \text{a fixed ordering 
   of the arcs in $\bbb_{\text{in,out}}\cap\bbb'_{\text{in,out}}$}.
\]
Observe that we restrict the notion of ``neighbors'' to two
subnetworks $\M$ and $\M'$ induced by disjoint subsets of nodes $X$
and $X'$, but which share some input/output arcs.

To \emph{merge} $\M$ and $\M'$ means to produce
the subnetwork of $\N$ induced by $X\cup X'$, which
we denote $(\merg{\M}{\M'})$. If $\M$ and $\M'$
are not neighbors, then $(\merg{\M}{\M'})$ is
undefined.%
   \footnote{$(\merg{\M}{\M'})$ can be written in terms of the $\Bind$
   operation defined in Section~\ref{sect:disassemble-and-reassemble},
   by applying it as many times as there are arcs in
   $\commonA{\M,\M'}$. But it is more economical to just write
   ``$(\merg{\M}{\M'})$''.  }
\end{definition}

\begin{definition}{Strong Neighbors}
\label{def:strong-neighbors}
Let $\N = (\nn,\aaa)$ be a network, and  
$\M'$ and $\M''$ be neighbors in $\N$ 
induced by the disjoint subsets of nodes $X'$ and $X''$, as 
in Definition~\ref{def:neighbors}. We define the \emph{binding strength}
of the neighbors $\M'$ and $\M''$ as follows:
\[
    \bindingS{\M',\M''}\ :=\ 
    \ \size{\commonA{\M',\M''}}. 
\]
Because $\M'$ and $\M''$ are neighbors, 
$\bindingS{\M',\M''}\geqslant 1$. The external dimension of
$\merg{\M'}{\M''}$ is:
\[
  \exDim{\merg{\M'}{\M''}}
  \ =\ \exDim{\M'}+\exDim{\M''}-2\cdot\bindingS{\M',\M''} .
\]
We say $\M'$ and $\M''$ are
\emph{strong neighbors} if the following inequality is satisfied:
\begin{alignat*}{4}
  & \exDim{\merg{\M'}{\M''}}\ \leqslant
\\
  & \min\;
  \bigl(
  \Set{\,\exDim{\merg{\M'}{\M}}\;|
  \;\text{$\M$ is a neighbor of $\M'$}\,}
\\
  &\qquad\ \ \ \cup\Set{\,\exDim{\merg{\M''}{\M}}\;|
            \;\text{$\M$ is a neighbor of $\M''$}\,}\bigr) .
\end{alignat*}
Equivalently, $\M'$ and $\M''$ are \emph{strong neighbors} if:
\begin{alignat*}{4}
  & \bindingS{\M',\M''}\ \geqslant
\\
  & \max\;
  \bigl(
  \Set{\,\bindingS{\M',\M}\;|\;\text{$\M$ is a neighbor of $\M'$}\,}
\\
  &\qquad\ \ \ \cup\Set{\,\bindingS{\M'',\M}\;|
            \;\text{$\M$ is a neighbor of $\M''$}\,}\bigr) .
\end{alignat*}
In words, $\M'$ and $\M''$ are \emph{strong neighbors} if they
have at least as many external arcs in common as each has in common 
with another neighbor $\M$.
\end{definition} 

In Algorithm~\ref{alg:bindSchedule}, we use repeatedly the same group of
instructions, which we here collect together as a single ``macro''
instruction called $\MergeSym$.  Let $X_1\uplus\cdots\uplus X_p = \nn$
be a partition of the nodes of the given network $\N$. Let
$\CC = \Set{\M_1,\ldots,\M_p}$ be the subnetworks of $\N$ induced by
$X_1,\ldots,X_p$, respectively. Let $\bbb_{\text{\#}}^1,\ldots,\bbb_{\text{\#}}^p$ 
be the (necessarily disjoint) sets of \emph{internal}
arcs of $\M_1,\ldots,\M_p$, respectively. With $\N$ thus disassembled,
if we select two distinct subnetworks $\M,\M'\in\CC$ that are neighbors,
we write $\MergeSym$ with $5$ arguments as:
\[ 
   \MergeSym(\M,\M',\sigma,\delta,\CC)
\]
where the last $3$ are the following quantities:
\begin{itemize}[itemsep=1pt,parsep=2pt,topsep=2pt,partopsep=0pt] 
\item $\sigma = \text{an ordering of the arcs in
      $\bbb_{\text{\#}}^1\cup\cdots\cup\bbb_{\text{\#}}^p$}$
      \ \ (the binding schedule computed by the algorithm)
\item $\delta \geqslant 3$ \ \ (a tight upper bound on $\Index{\sigma}$)
\item $\CC = \Set{\M_1,\ldots,\M_p}$
\end{itemize}
The notions of a ``binding schedule'' and its ``index'' were defined
in Section~\ref{sect:disassemble-and-reassemble}.
The macro expansion of $\MergeSym(\M,\M',\sigma,\delta,\CC)$
is shown in Figure~\ref{fig:macroExpansion}.

{\spacing{\normalsize}{1.2}
\begin{figure}[h] 
\fbox{
\begin{minipage}[b][2.0cm][t]{0.95\linewidth}
\vspace{-.15in} 
   \begin{center}
   \small
   \begin{custommargins}{-1.5cm}{-2.5cm}
\begin{alignat*}{6}
  &1.\quad && \sigma
    \ &&:=\ \ &&\sigma\ \;\commonA{\M,\M'} 
    \quad &&\text{// append\ $\commonA{\M,\M'}$\ to $\sigma$}
\\
  &2. && \delta
    \ &&:=\ &&\max\Set{\delta,\;\exDim{\M}+\exDim{\M'}-2} 
    \quad &&\text{// new tight upper bound on $\Index{\sigma}$}
\\
  &3. && \CC
    \ &&:=\ &&(\CC - \Set{\M,\M'})\cup\Set{\merg{\M}{\M'}}
    \quad &&\text{// exclude $\M$ and $\M'$, include
                      their merge $\merg{\M}{\M'}$}
\end{alignat*}
    \end{custommargins}
    \end{center}
\end{minipage}
\vspace{-.2in} 
}
\caption{Macro expansion of $\MergeSym(\M,\M',\sigma,\delta,\CC)$.}
\label{fig:macroExpansion}
\end{figure}
}
\noindent
Instead of the three instructions shown in
Figure~\ref{fig:macroExpansion}, we can now write
a single macro instruction:
\[
   (\sigma,\,\delta,\,\CC)\ :=\ \MergeSym(\M,\M',\sigma,\delta,\CC)
\]

\noindent
We need one more classification of arcs before we define
Algorithm~\ref{alg:bindSchedule}. 
Let network $\N = (\nn,\aaa)$ be given in a good $k$-outerplanar 
embedding, with $\aaa = \aaa_{\text{in,out}}\uplus\aaa_{\text{\#}}$.
Using Lemma~\ref{lem:paths-in-good-embeddings},
we partition the internal arcs of $\N$ into two parts, 
$\aaa_{\text{\#}} = \aaa_{\text{\#,1}}\uplus\aaa_{\text{\#,2}}$, where:
\begin{alignat*}{5}
  & \aaa_{\text{\#,1}}\ &&:=\ &&\SET{\,a\in\aaa_{\text{\#}}\;|
    \;\text{$a$ is a cross arc}\,} \ \cup
\\
  & &&  &&\SET{\,a\in\aaa_{\text{\#}}\;|
    \;\text{there is $1\leqslant i\leqslant k$ such that
    such that $\Set{\head{a},\tail{a}}\cap (L_i-L_i')\neq\varnothing$}\,} ,
\\
  & \aaa_{\text{\#,2}}\ &&:=\ && \aaa_{\text{\#}} - \aaa_{\text{\#,1}} .
\end{alignat*}
In words, $\aaa_{\text{\#,1}}$ is the set of: (1) all cross arcs, and (2)
all peeling arcs on a path connecting 
two consecutive concentric cycles of the good embedding of $\N$.
See the statement of Lemma~\ref{lem:paths-in-good-embeddings}
for further explanation.

{\spacing{\normalsize}{1.3}
\begin{algorithm} 
\caption {\quad $\Schedule$: Define Optimal Binding Schedule} 
          \label{alg:bindSchedule}
\begin{algorithmic}[1]
    \Statex \textbf{input}: 
      good planar embedding of network $\N = (\nn,\aaa)$, 
      with $\aaa = \aaa_{\text{in,out}}\uplus\aaa_{\text{\#,1}}\uplus\aaa_{\text{\#,2}}$
    \Statex \textbf{output}: 
      $\sigma = b_1 b_2 \cdots b_m$, an ordering 
      of internal arcs of $\N$ (a ``binding schedule''), 
     \Statex $\phantom{\textbf{output}:}$
      where \ $\aaa_{\text{\#}} = \Set{b_1, b_2, \ldots , b_m}$,
      together with a tight upper bound $\delta$ on $\Index{\sigma}$.
    \Statex \rule[2pt]{15cm}{.9pt} 
    \Statex \underline{\textbf{initialization}}
    \State $k :=\ \text{outerplanarity of $\N$}$
    \State $\sigma :=\ \varepsilon$
    \hspace*{.32in} // $\sigma$ is initially the empty ``binding schedule''
    \State $\CC :=\ \SET{\,\M\;\bigl|
        \;\text{$\M$ subnetwork of $\N$ induced
                by $\Set{\nu}$ with $\nu\in\nn$}\,} $
    \Statex \hspace*{.72in}
      {// $\N$ is disassembled into 
          $\size{\N}$ one-node subnetworks, each of external dimension $3$}
    \Statex \underline{\textbf{first iteration}} \qquad {// pre-processing}
    \For {every arc $a\in\aaa_{\text{\#,1}}$}
    \State $(\sigma,\delta,\CC) :=
          \ \MergeSym(\M,\M',\sigma,\delta,\CC)$
    \Statex \qquad {where $\M,\M'\in\CC$ are the two
             subnetworks such that $a\in\commonA{\M,\M'}$}
    \EndFor\qquad 
     {// every arc $a\in\aaa_{\text{\#,1}}$ is now included in $\sigma$,} 
    \Statex \qquad\qquad\quad
     {// for every $\M\in\CC$ such that $\size{\M} = 1$,
         the single node of $\M$ is an input/output node,}
    \Statex \qquad\qquad\quad
     {// for every $\M\in\CC$ such that $\size{\M} \geqslant 2$,
         ignoring input/output arcs of $\N$,
         $\exDim{\M} = 4$}
    \Statex \underline{\textbf{second iteration}} 
           \qquad {// pre-processing}
    \While {there are neighbors $\M,\M'\in\CC$ 
            such that:\quad $\size{\M} = 1$
             \ \textbf{or}\ $\bindingS{\M,\M'} = 2$\quad }
       \State  $(\sigma,\delta,\CC) :=
          \ \MergeSym(\M,\M',\sigma,\delta,\CC)$
    \EndWhile \qquad  {// for every $\M\in\CC$,
          ignoring input/output arcs of $\N$,
          $\exDim{\M} = 4$}
    \Statex \underline{\textbf{main iteration}}
           \qquad {// re-assemble $\N$ from the subnetworks in $\CC$
                   and store it in $\PP$}
    \State\label{step:initialize-P}
      $\PP :=\ \M$\quad {where $\M$ is any ``outermost'' subnetwork in $\CC$}
    \State $\CC :=\ \CC - \Set{\PP}$ 
    \While {$\CC \neq \varnothing$} 
    \State\label{step:select-strong-neighbor}
       {select $\M\in\CC$ which is a strong neighbor of $\PP$}
    \State  $\sigma :=\ \sigma\ \;\commonA{\PP,\M}$
    \State  $\delta :=\ \max\Set{\delta,\;\dime{\PP}+\dime{\M}-2}$
    \State  $\PP := \merg{\PP}{\M}$
    \State  $\CC := \CC - \Set{\M}$ 
    \EndWhile \qquad {// $\N$ is now re-assembled and stored in $\PP$}
    \State \Return $\sigma$ and $\delta$
\end{algorithmic}
\end{algorithm}
}

\begin{example}
\label{ex:running-example3}
This is a continuation of the network considered in
Examples~\ref{ex:running-example1} and~\ref{ex:running-example2}.
They refer to the good $3$-outerplanar embedding on the right in
Figure~\ref{fig:thirty-node-network}, and again in
Figure~\ref{fig:thirty-node-network-on-rectangular-grid}, which we use
to illustrate the operation of Algorithm~\ref{alg:bindSchedule}.
The progress of Algorithm~\ref{alg:bindSchedule} is shown
in Figure~\ref{fig:Disassemble-and-Reassemble1}, for the
\textbf{first iteration} and the \textbf{second iteration}, 
and in Figure~\ref{fig:Disassemble-and-Reassemble2} for 
the \textbf{main iteration}.
\end{example}

\begin{Proof}{for Theorem~\ref{fact:principal-typing-of-planar}}
Let $\N_0$ be the network $\N$ in the statement of
Theorem~\ref{fact:principal-typing-of-planar}, to distinguish it from
the ``$\N$'' introduced below.  Let $\N_0 = (\nn_0,\aaa_0)$ be given
in a $k_0$-outerplanar embedding, for some $k_0 \geqslant 1$. Let $p
= \size{\aaa_{0,\text{in}}} \geqslant 1$, $q
= \size{\aaa_{0,\text{out}}} \geqslant 1$, $m_0
= \size{\aaa_{0,\text{\#}}} \geqslant 1$ and $n_0
= \size{\nn_0} \geqslant 1$. Because $\N_0$ is planar, $\N_0$ is
sparse; more specifically, $m_0\leqslant 3n_0-6$ (see, for example,
Theorem 4.2.7 and its corollaries in~\cite{diestel2012}). Hence, the
complexity bound $\bigOO{m_0+n_0}$ is the same as $\bigOO{n_0}$.

By Lemmas~\ref{lem:from-arbitrary-to-3-regular} and~\ref{lem:good-embeddings},
we can transform the $k_0$-outerplanar embedding of $\N_0$ into a 
\emph{good} $k$-outerplanar embedding of an equivalent $\N$, 
with $k\leqslant 2k_0$. The transformation is such that
$\exDim{\N_0} = \exDim{\N}= p+q$. Let 
$m = \size{\aaa_{\text{\#}}}$ and $n = \size{\nn}$. 
By Lemma~\ref{lem:from-arbitrary-to-3-regular}, 
$m\leqslant 3 m_0 + \bigOO{m_0}$ and $n\leqslant 2 m_0 + \bigOO{m_0}$,
where ``$\bigOO{m_0}$'' accounts for the arcs and nodes introduced in
the construction of Lemma~\ref{lem:good-embeddings}.

We next run Algorithm~\ref{alg:bindSchedule} on the good 
$k$-outerplanar embedding of $\N$. We first consider the correctness
of the algorithm, and then its run-time complexity. The initialization
consists in breaking up $\N$ into $n$ one-node subnetworks, each of 
external dimension $= 3$. 

The \textbf{first iteration} assembles new subnetworks $\M$ of external
dimension $= 4$, if we ignore the presence of all input/output arcs of $\N$.%
   \footnote{By ``ignoring an input/output arc $a$'', we mean that we omit
      $a$ but not the input/output node $\nu$ to which $a$ is incident. 
      The node $\nu$ is thus temporarily made to have degree $= 2$. }
See Figure~\ref{fig:Disassemble-and-Reassemble1} for an
illustration. Such a subnetwork $\M$ of external dimension $= 4$ has
two nodes -- say $\Set{\nu_1,\nu_2}$ -- that are \emph{either} on the
same peeling level \emph{or} on two consecutive peeling
levels. Specifically, there is $1\leqslant i\leqslant k$, such
that \emph{either} both $\nu_1,\nu_2\in L_i'$ \emph{or} $\nu_1\in
L_i'$ and $\nu_2\in L_{i+1}'$.

A similar conclusion applies to the \textbf{second iteration}: It
assembles new subnetworks $\M$ of external dimension $= 4$, again ignoring
the presence of all input/output arcs of $\N$. See
Figure~\ref{fig:Disassemble-and-Reassemble1} for an
illustration. Such a subnetwork has external
dimension $= 4$, with four input/output nodes (these are \emph{not}
the same as the input/output nodes of $\N$) -- say
$\Set{\nu_1,\nu_2,\nu_3,\nu_4}$ -- that are \emph{either} all on the
same peeling level \emph{or} on two consecutive peeling levels with
two nodes on each. Specifically, there is $1\leqslant i\leqslant k$,
such that \emph{either} both $\Set{\nu_1,\nu_2,\nu_3,\nu_4}\subseteq
L_i'$ \emph{or} $\Set{\nu_1,\nu_2}\subseteq L_i'$ and
$\Set{\nu_3,\nu_4}\subseteq L_{i+1}'$.

At the end of the \textbf{second iteration}, if
we ignore all input/output arcs of $\N$, every subnetwork $\M$ in
$\CC$ has external dimension  $= 4$ and is one of two kinds:
\begin{itemize}[itemsep=1pt,parsep=2pt,topsep=2pt,partopsep=0pt] 
\item  $\M$ is assembled in the \textbf{first iteration} and not
       affected by the \textbf{second iteration}. In this case,
       $\M$ straddles \emph{either}
       two opposite nodes of the same level $L_i'$ \emph{or} two 
       nodes of two consecutive levels $L_i'$ and $L_{i+1}'$.
\item  $\M$ is assembled in the \textbf{second iteration} from
       two or more networks of the previous kind. In this case,
       $\M$ straddles \emph{either} two opposite peeling arcs 
       on the same level \emph{or} two peeling arcs on two 
       consecutive levels.
\end{itemize}
At the end of the \textbf{second iteration}, for any two subnetworks
$\M_1,\M_2\in\CC$, if $\M_1$ and $\M_2$ are neighbors, then
$\bindingS{\M_1,\M_2} = 1$.

The task of the \textbf{main iteration} in
Algorithm~\ref{alg:bindSchedule} is to re-assemble the original
$\N$ from the subnetworks in $\CC$ at the end of the 
\textbf{second iteration} in such a way as to minimize the external 
dimension of the intermediate subnetwork $\PP$. We initialize $\PP$ by
selecting for it an ``outermost'' $\M$ in $\CC$ 
(\textbf{line~\ref{step:initialize-P}} of Algorithm~\ref{alg:bindSchedule}),
\ie, we choose $\M$ so that all its nodes
are \emph{either} all on level $L_1'$ \emph{or} on two
consecutive levels $L_1'$ and $L_2'$. 

The selection of the initial $\M$ in the \textbf{main iteration} is
totally arbitrary. For example, in the third assembly on the right of
Figure~\ref{fig:Disassemble-and-Reassemble1}, we choose for this
initial $\M$ the subnetwork containing the north-west corner of $\N$,
and the corresponding progress of Algorithm~\ref{alg:bindSchedule}
during the \textbf{main iteration} is shown in
Figure~\ref{fig:Disassemble-and-Reassemble2}.

To minimize the external dimension of $\PP$ at every turn of the
\textbf{main iteration}, it suffices to select any $\M\in\CC$
which is a strong neighbor of $\PP$
(\textbf{line~\ref{step:select-strong-neighbor}} of
Algorithm~\ref{alg:bindSchedule}). For every $\M\in\CC$ which is a
neighbor of $\PP$, we have $\bindingS{\PP,\M} \geqslant 1$.
Initially, $\exDim{\PP} = 4$ (ignoring all input/output arcs of $\N$),
and the maximum number of strong neighbors $\M\in\CC$ such that
$\bindingS{\PP,\M} = 1$ in consecutive turns of the \textbf{main
iteration} is $(k-1)$. It is now easy to see that
$\exDim{\PP} \leqslant 2k+2$ is an invariant of the \textbf{main
iteration}.  Figure~\ref{fig:Disassemble-and-Reassemble2} shows how
$\PP$ may be assembled during the \textbf{main iteration}. For a
schematic example, which is also easier to follow by making the
peeling levels $L_1', L_2', L_3', \ldots$ drawn as concentric
ellipses, see Figures~\ref{fig:concentric-outerplanar0}
and~\ref{fig:concentric-outerplanar}.

Consider now a subnetwork $\M$ which is obtained by merging 
subnetworks $\M'$ and $\M''$, \ie, $\M = \merg{\M'}{\M''}$, where:
\begin{itemize}[itemsep=1pt,parsep=2pt,topsep=2pt,partopsep=0pt] 
\item $\exDim{\M'} = \ell' \geqslant 2$ and 
      $\exDim{\M''} = \ell'' \geqslant 2$,
\item $\commonA{\M',\M''} = \Set{a_1,\ldots,a_j}$ where
      $j\leqslant \min \Set{\ell',\ell''}$.
\end{itemize}
The arcs in $\Set{a_1,\ldots,a_j}$ are re-connected one at a time,
so that, starting from $\Set{\M',\M''}$ and ending with $\M$, the
merge operation produces $j$ intermediate subnetworks (including $\M$)
with external dimensions:
\[
   (\ell'+\ell'' - 2),\ (\ell'+\ell'' - 4),\ \ldots\ ,\ (\ell'+\ell'' - 2j),
\]
respectively. Hence, while $\exDim{\PP}\leqslant  2k+2$, the maximum
external dimension encountered in the course of the operation of 
Algorithm~\ref{alg:bindSchedule} is --
again ignoring all input/output arcs of $\N$:
\begin{alignat*}{5}
   & \text{``maximum external dimension of $\PP$''}
   +  \text{``external dimension of all subnetworks in $\CC$''} - 2 
\\
   & \leqslant\ (2k + 2) + 4 - 2\ =\ 2k + 4,
\end{alignat*} 
Hence, if we include
the presence of the $p+q$ input/output arcs of $\N$, the maximum external
dimension of subnetworks produced during the entire operation of
Algorithm~\ref{alg:bindSchedule} cannot exceed
$2k + 4 + p + q$, which is precisely the final value assigned to
$\delta$ by Algorithm~\ref{alg:bindSchedule}, which is also
$\leqslant 4k_0 + 4 + p + q$ where $k_0$
is the outerplanarity of the original network $\N_0$.

\medskip
To conclude the proof of
Theorem~\ref{fact:principal-typing-of-planar}, we need to show that
the run-time complexity is $\bigOO{n_0} =
\bigOO{n}$. This is a straightforward consequence of the fact that:
\begin{enumerate}[itemsep=1pt,parsep=2pt,topsep=2pt,partopsep=0pt] 
\item The initial transformation from $\N_0$ to $\N$ 
      is carried out in time $\bigOO{n_0}$, according to
      Lemmas~\ref{lem:from-arbitrary-to-3-regular} and~\ref{lem:good-embeddings}.
\item Each of the four stages in Algorithm~\ref{alg:bindSchedule}
      (\textbf{initialization}, \textbf{first iteration},
       \textbf{second iteration}, and \textbf{main iteration}) runs in
      time $\bigOO{m}$, which is the same as $\bigOO{m_0} = \bigOO{n_0}$.
\end{enumerate}
Using the binding schedule $\sigma$ returned by Algorithm~\ref{alg:bindSchedule},
whose index is $\leqslant 4k_0 + 4 + p + q$ , we now invoke part 3
in Theorem~\ref{fact:principal-typing-build-up}, which is based
on Algorithm~\ref{alg:basic-compositional} in
Appendix~\ref{sect:appendix:disassemble}. We
conclude that the principal typing of the initial network $\N_0$ 
can be computed in time $m_0\cdot 2^{\bigOO{k_0+p+q}}$ or, equivalently,
in time $\bigOO{n_0}$ where the multiplicative constant depends on
$k_0, p$, and $q$ only.
\hfill \QED
\end{Proof}


\begin{figure}[t] 
%
\begin{minipage}[b][4.90cm][t]{0.95\linewidth}
   \begin{center}
      \begin{custommargins}{-1.5cm}{-2.5cm}
      \includegraphics[scale=.3,trim=0cm 9.00cm 0cm 3.60cm,clip]%
         {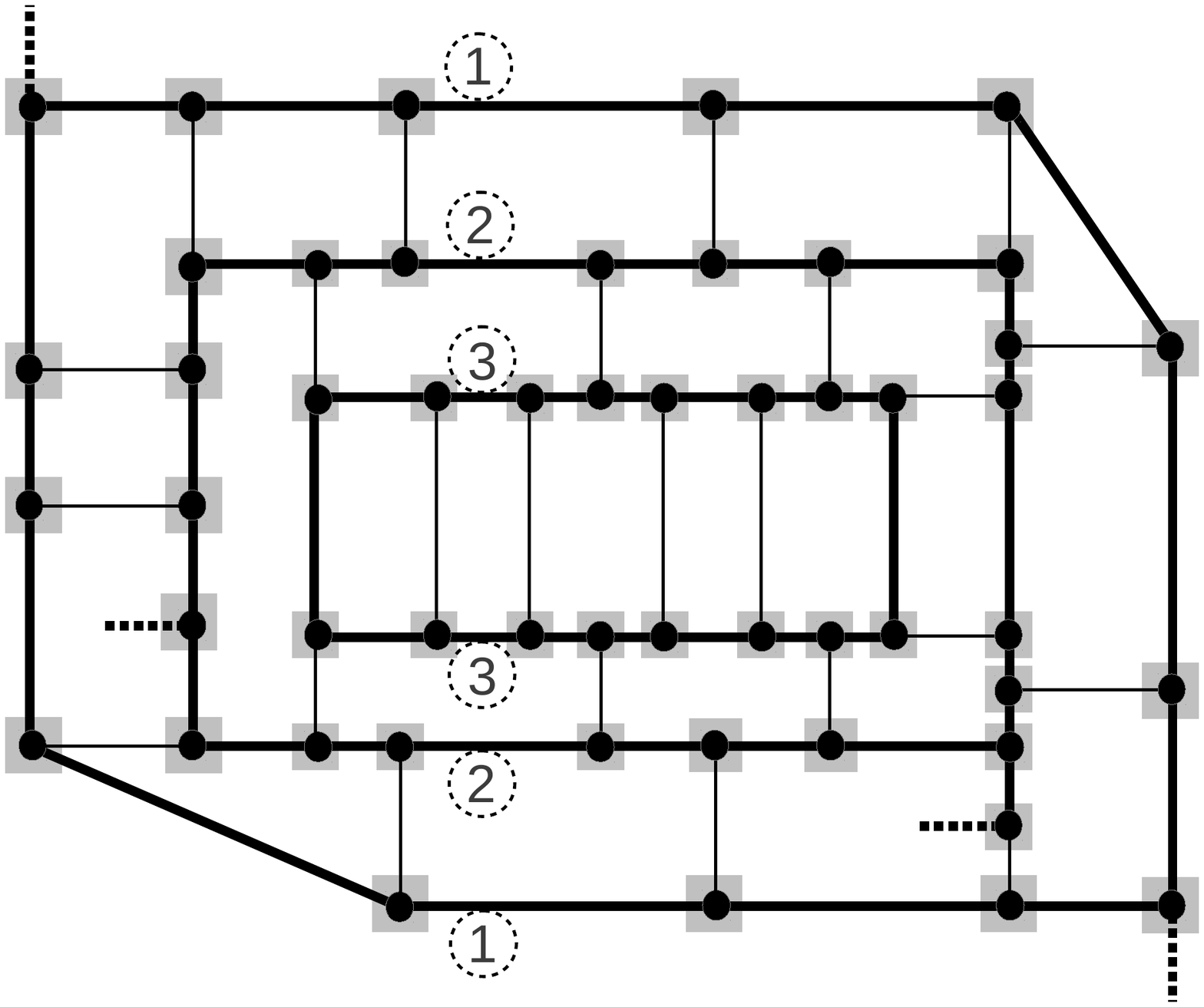}
      \includegraphics[scale=.3,trim=0cm 9.0cm 0cm 3.60cm,clip]%
         {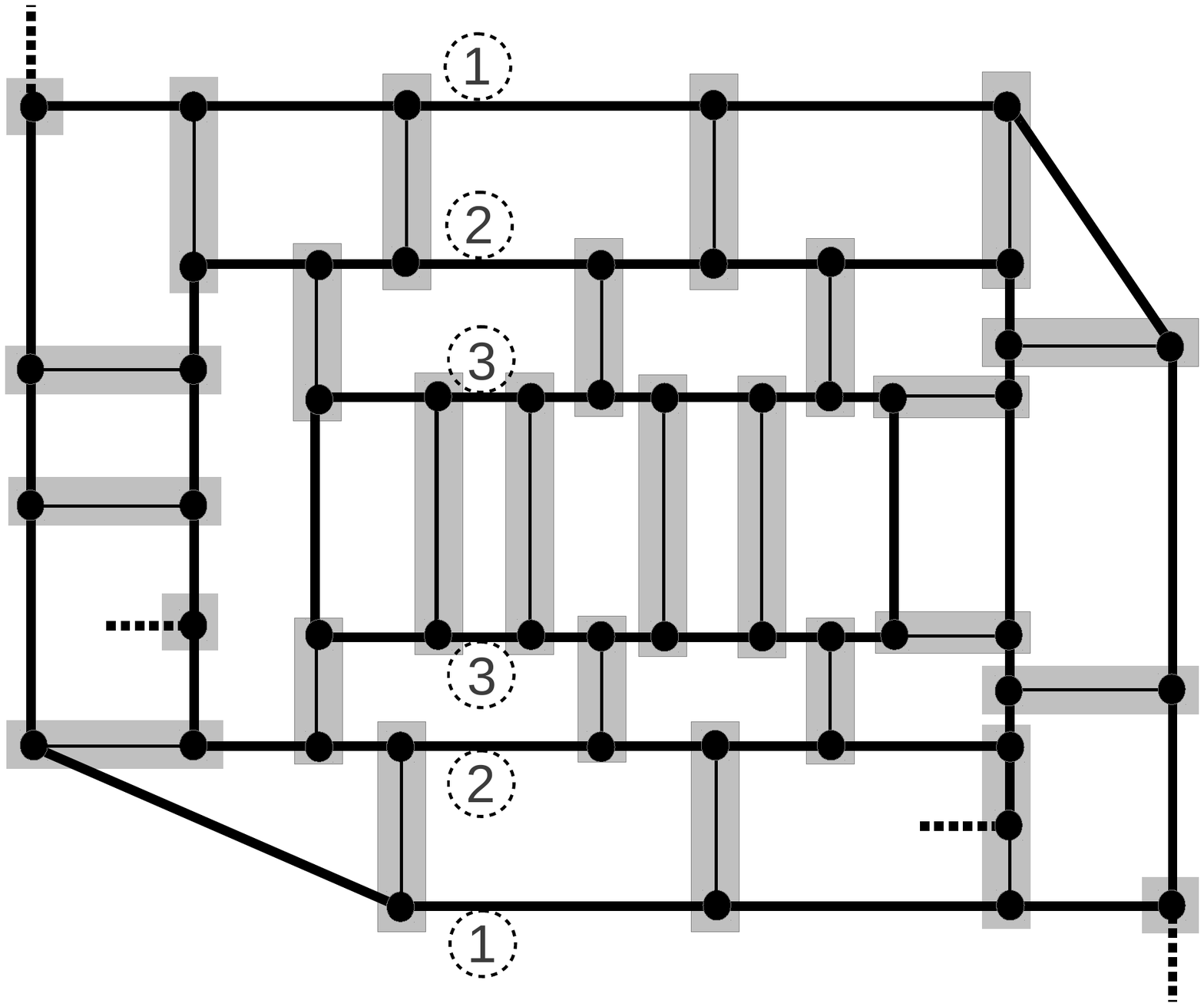}
      \includegraphics[scale=.3,trim=0cm 9.0cm 0cm 3.60cm,clip]%
         {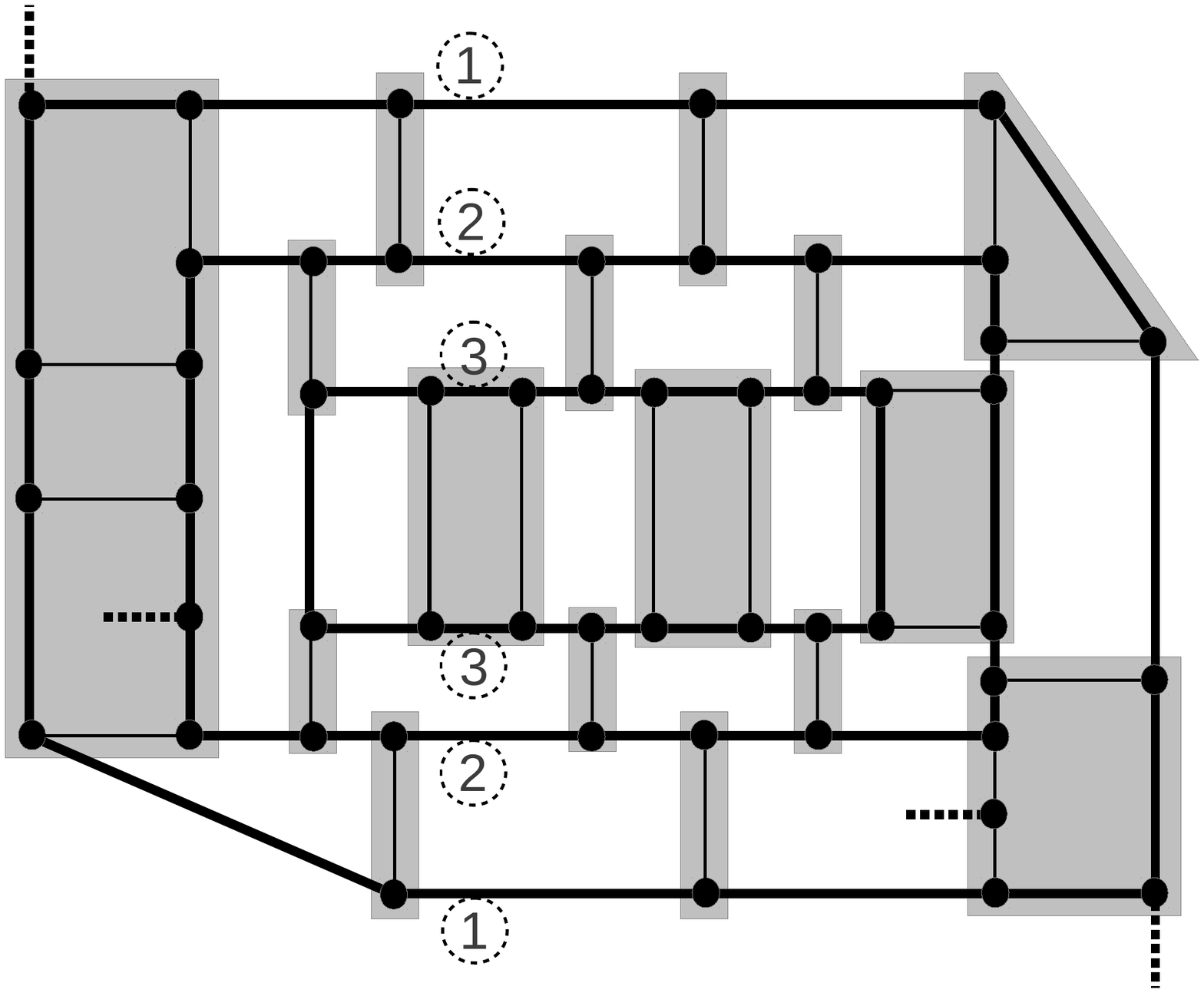}
      \end{custommargins}
    \end{center}
\end{minipage}
\caption{Progress of Algorithm~\ref{alg:bindSchedule} on 
         a good $3$-outerplanar embedding (same as in 
         Figure~\ref{fig:thirty-node-network-on-rectangular-grid}).
         The shaded areas demarcate the subnetworks already assembled. 
         Left assembly: after \textbf{initialization}, Middle assembly: 
         after \textbf{first iteration},
         Right assembly: after \textbf{second iteration}. 
         }
\label{fig:Disassemble-and-Reassemble1}
\end{figure}


\begin{figure}[h] 
%
\begin{minipage}[b][4.560cm][t]{0.95\linewidth}
   \begin{center}
      \begin{custommargins}{-1.5cm}{-2.5cm}
      \includegraphics[scale=.28,trim=0cm 9.00cm 0cm 3.65cm,clip]%
         {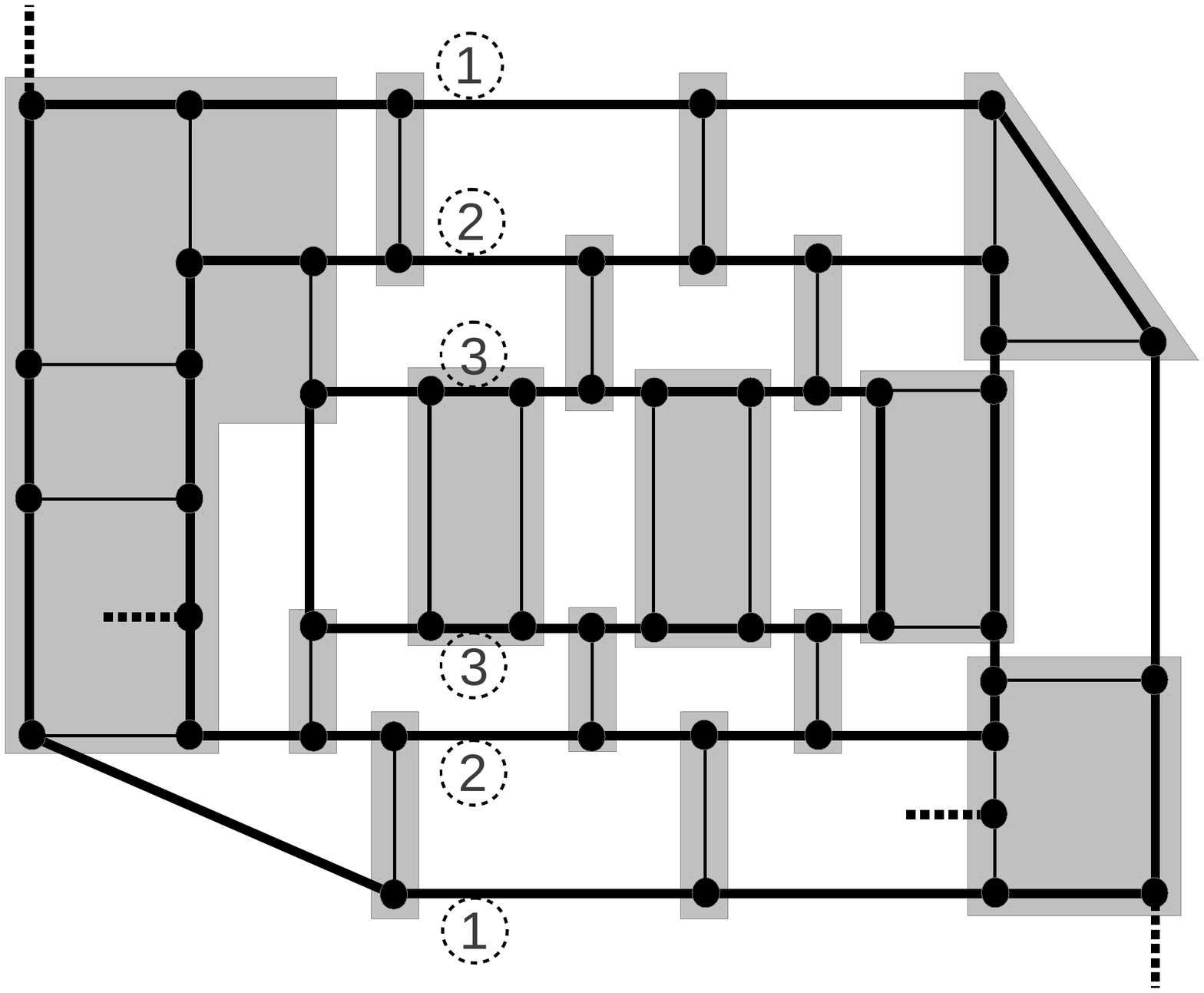}
      \includegraphics[scale=.28,trim=0cm 9.0cm 0cm 3.65cm,clip]%
         {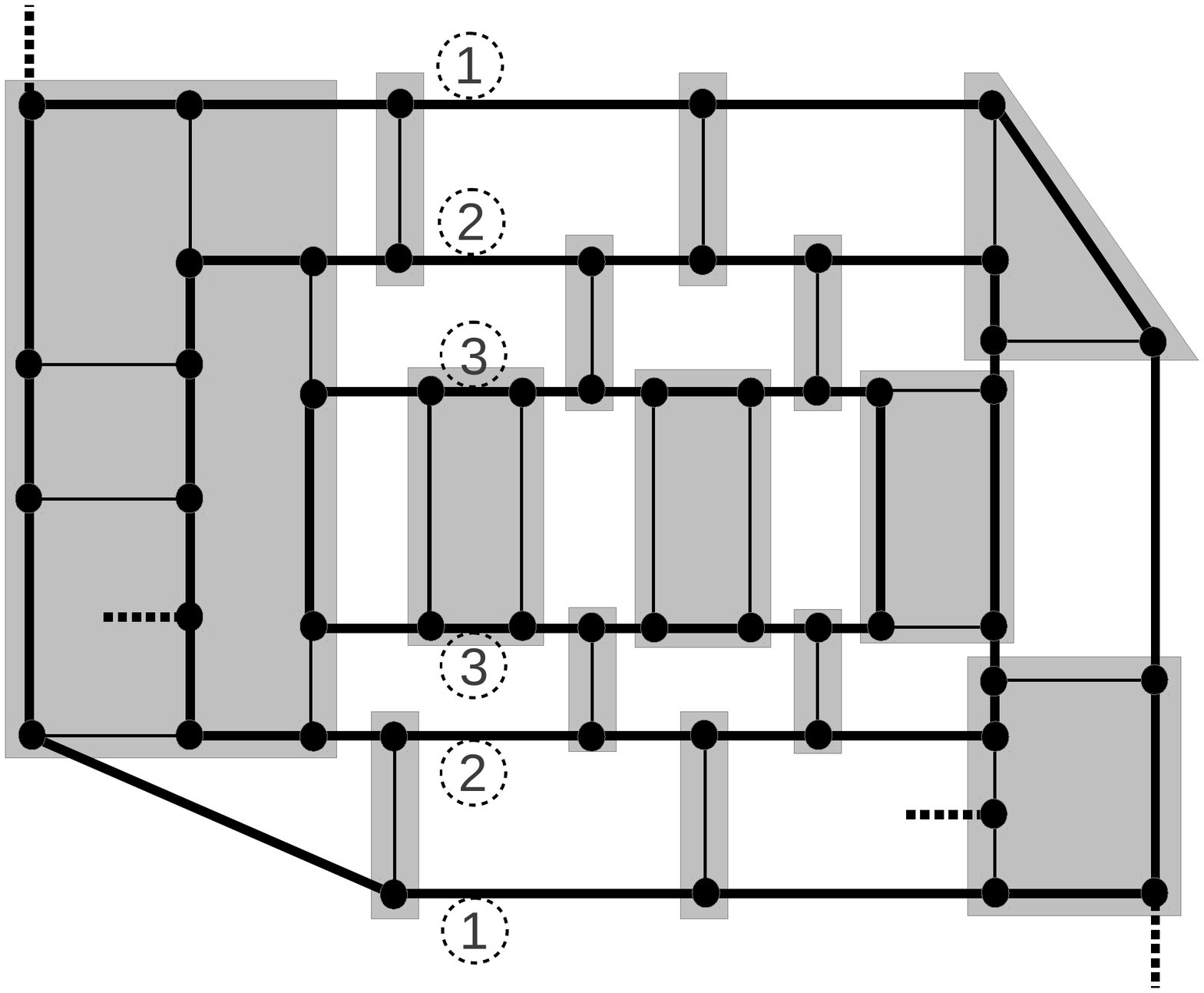}
      \includegraphics[scale=.28,trim=0cm 9.0cm 0cm 3.65cm,clip]%
         {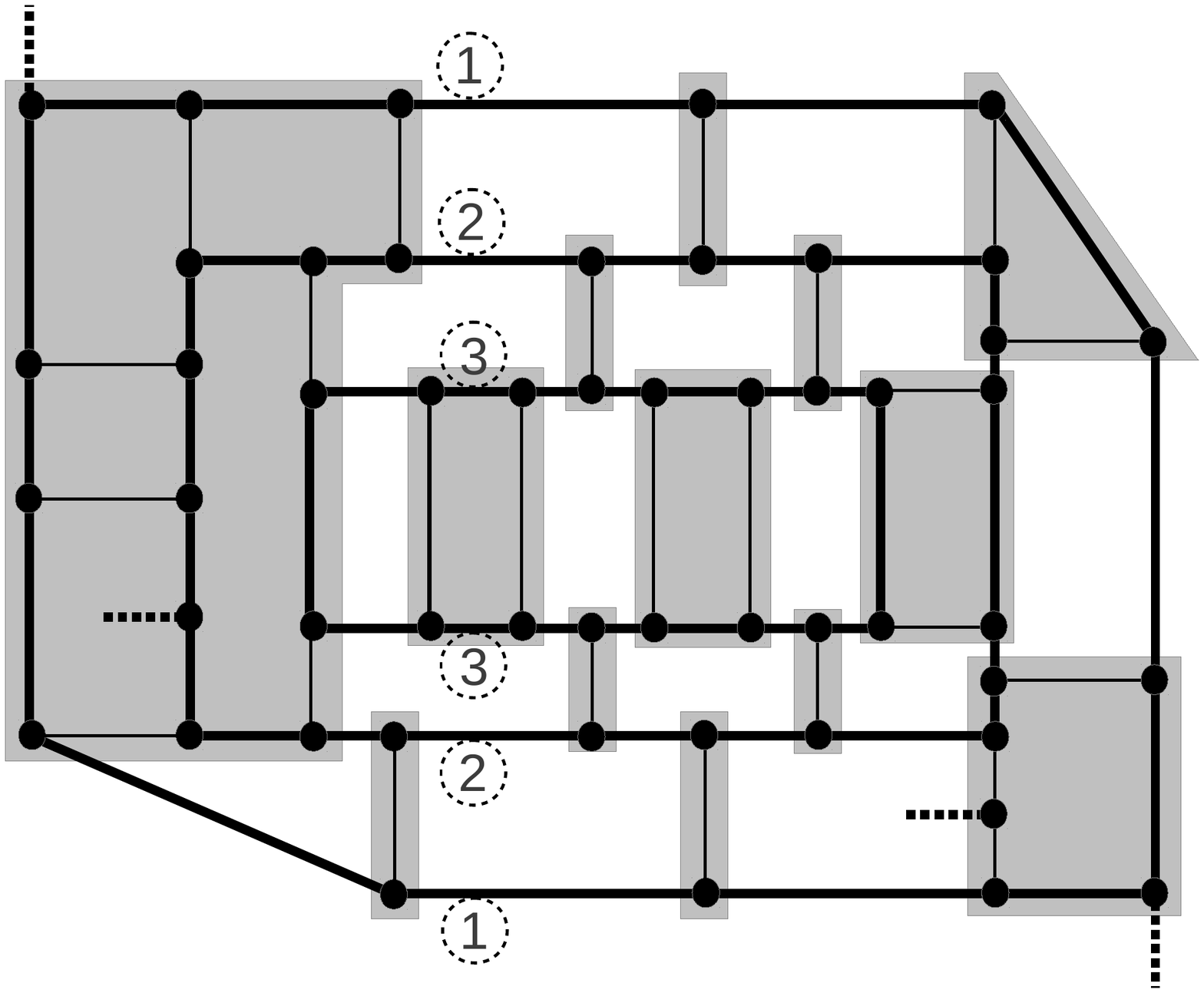}
      \end{custommargins}
    \end{center}
\end{minipage}

\begin{minipage}[b][4.56cm][t]{0.95\linewidth}
   \begin{center}
      \begin{custommargins}{-1.5cm}{-2.5cm}
      \includegraphics[scale=.28,trim=0cm 9.00cm 0cm 3.65cm,clip]%
         {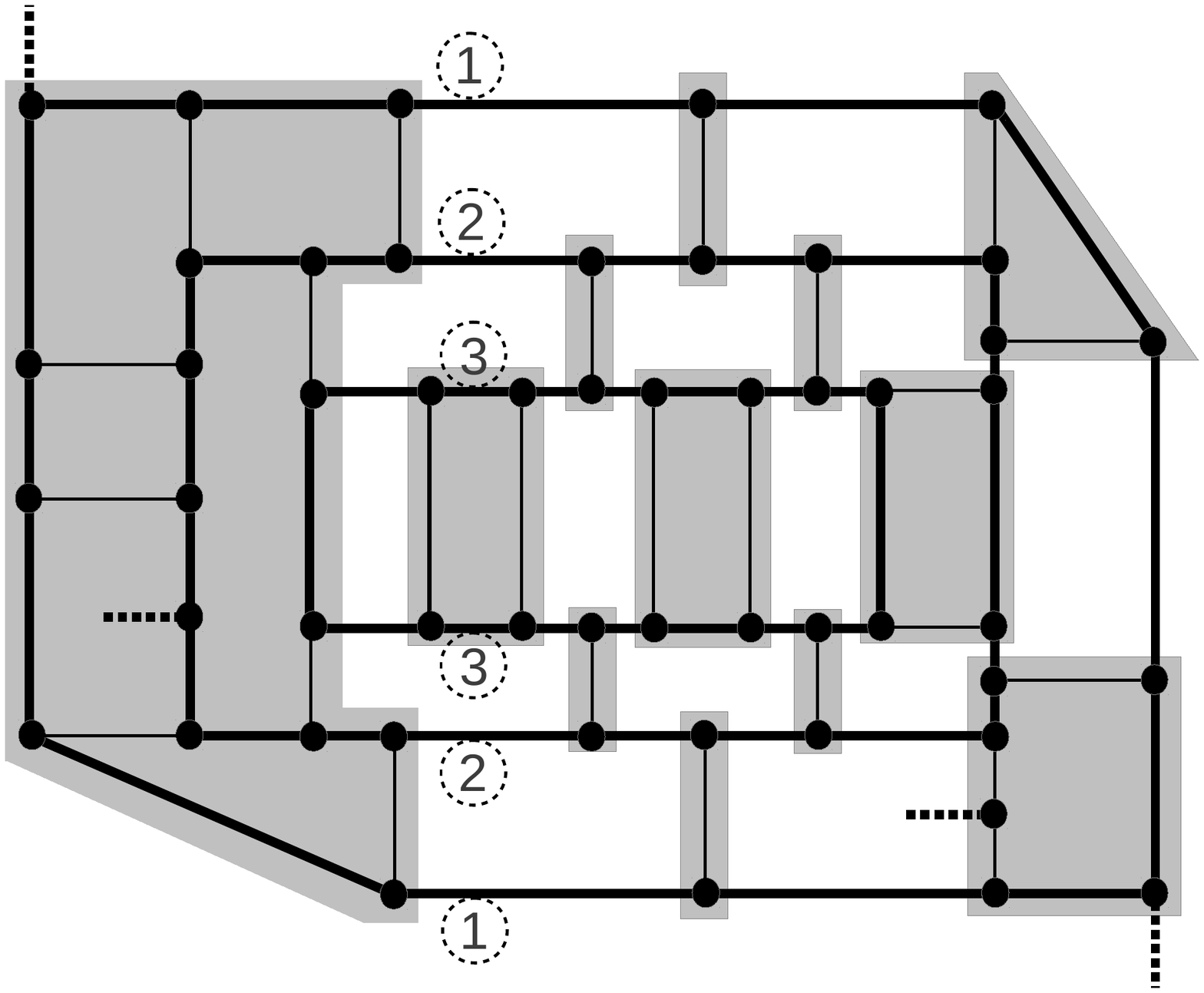}
      \includegraphics[scale=.28,trim=0cm 9.0cm 0cm 3.65cm,clip]%
         {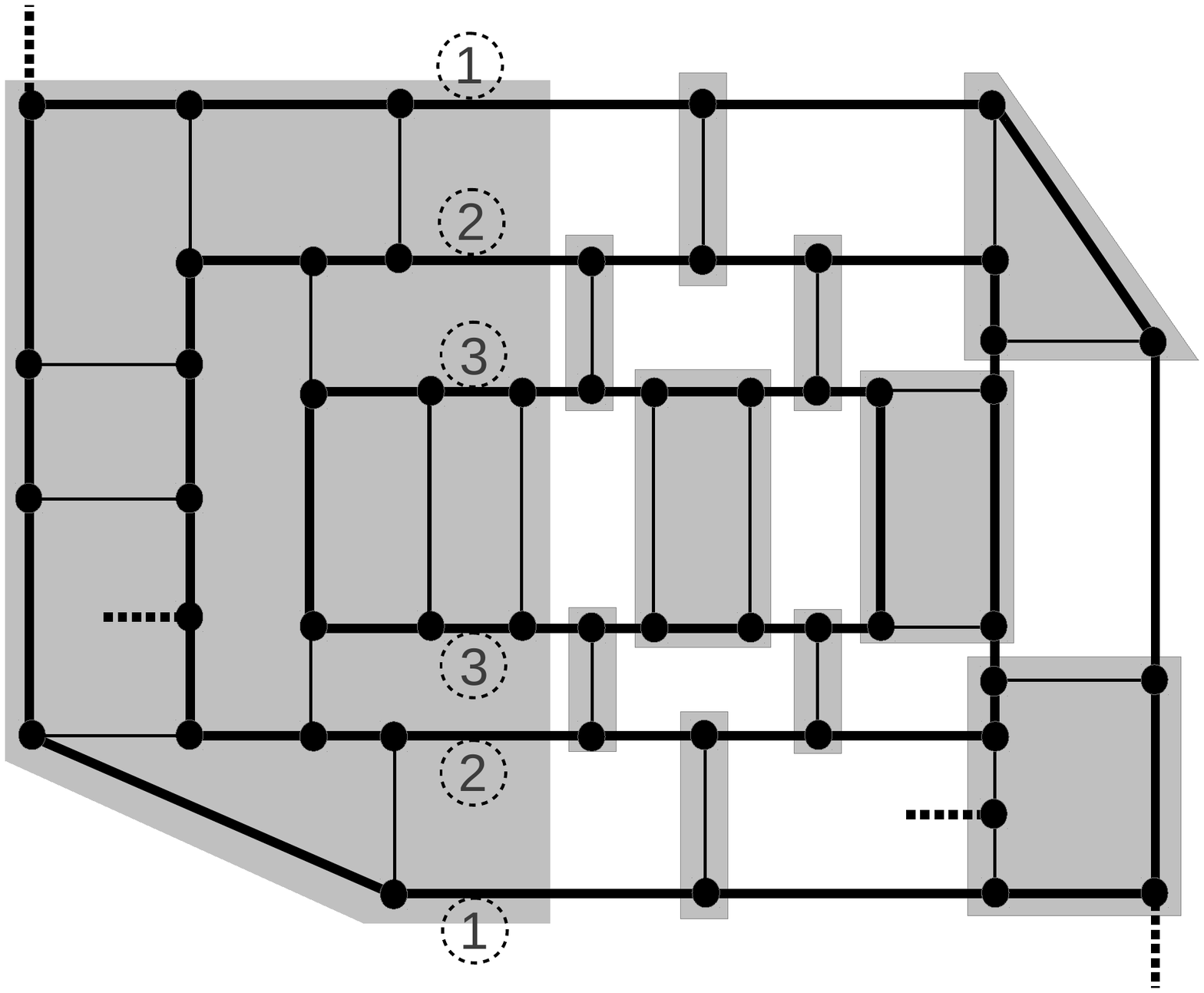}
      \includegraphics[scale=.28,trim=0cm 9.0cm 0cm 3.65cm,clip]%
         {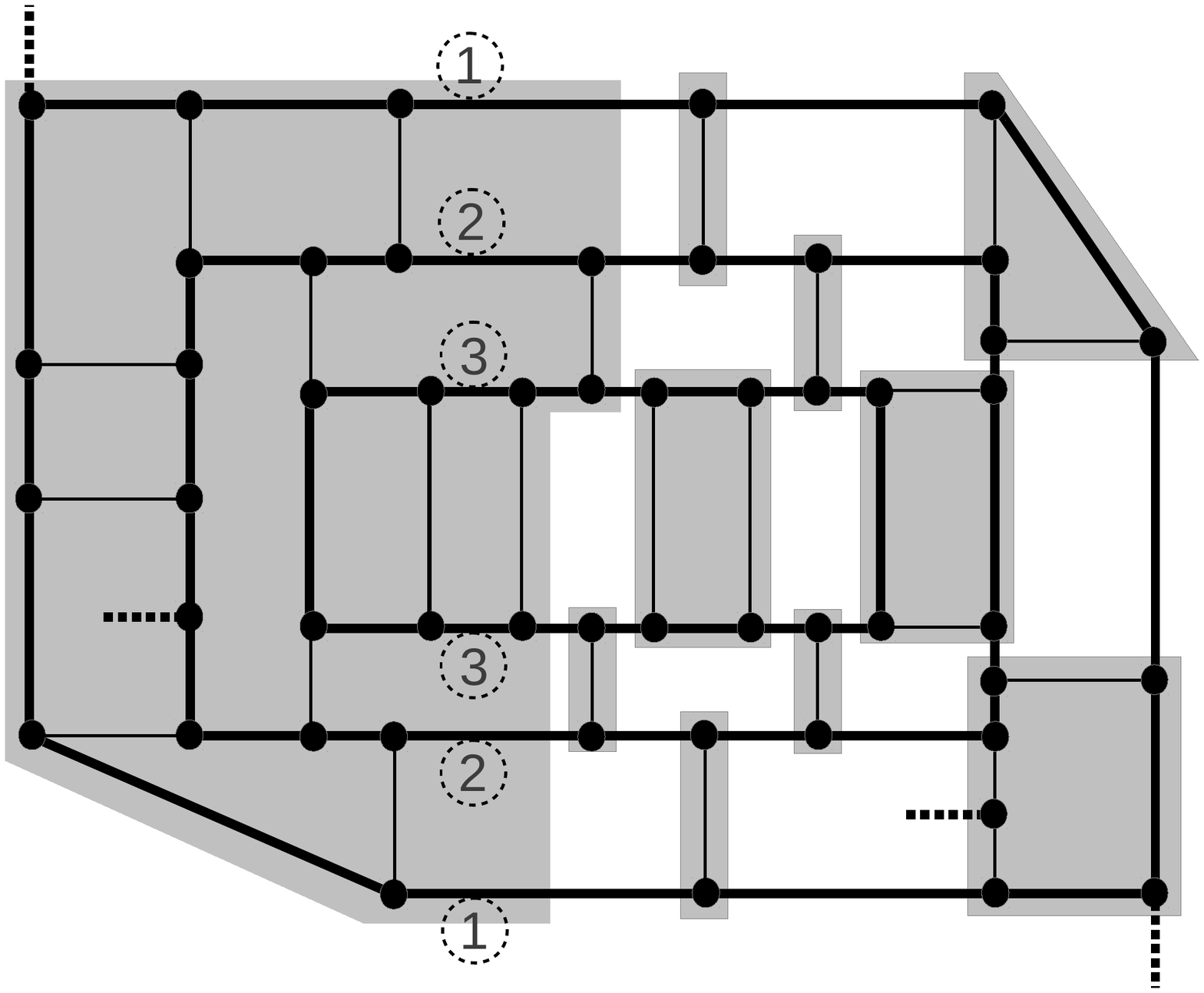}
      \end{custommargins}
    \end{center}
\end{minipage}

\begin{minipage}[b][4.56cm][t]{0.95\linewidth}
   \begin{center}
      \begin{custommargins}{-1.5cm}{-2.5cm}
      \includegraphics[scale=.28,trim=0cm 9.00cm 0cm 3.60cm,clip]%
         {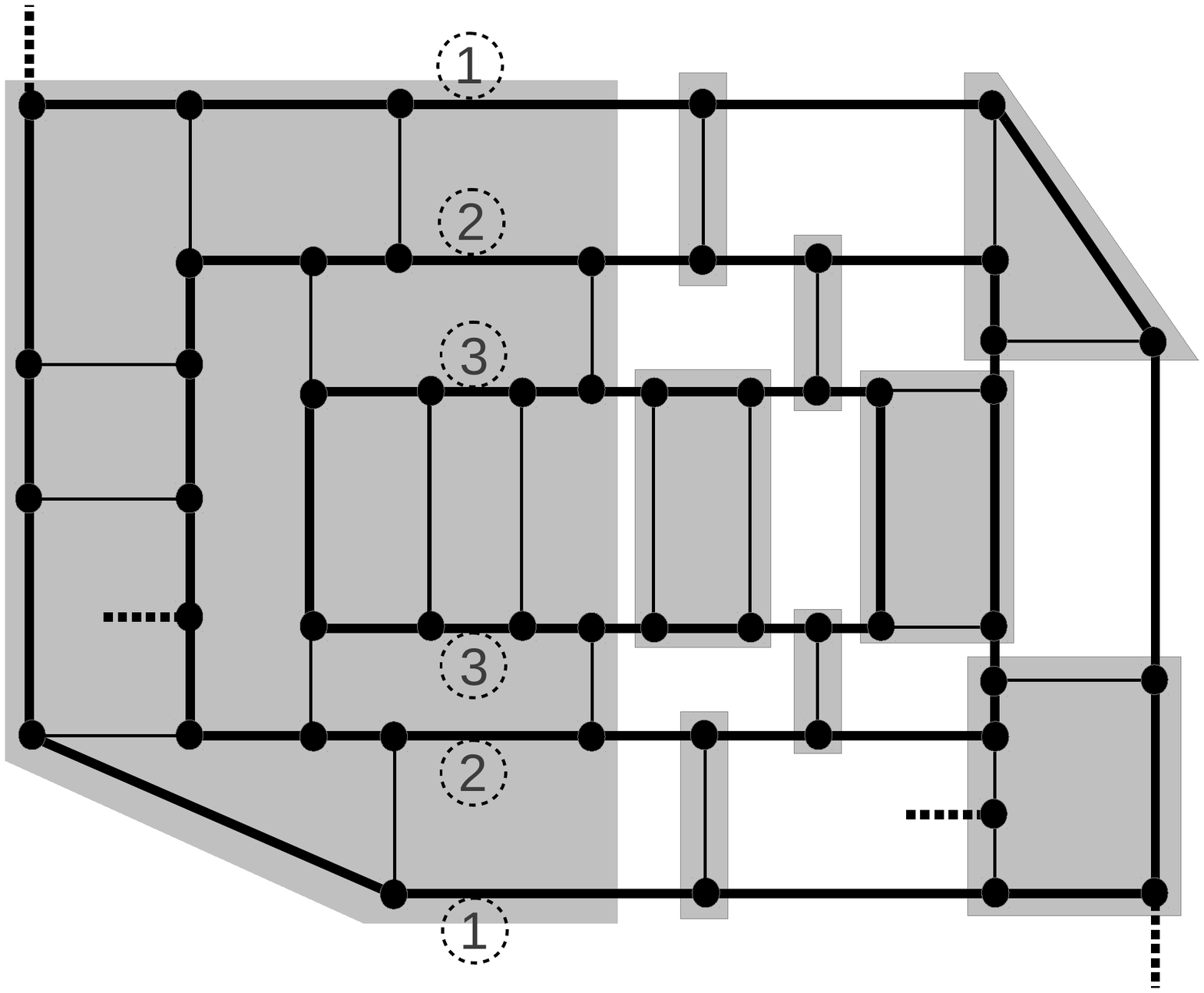}
      \includegraphics[scale=.28,trim=0cm 9.0cm 0cm 3.60cm,clip]%
         {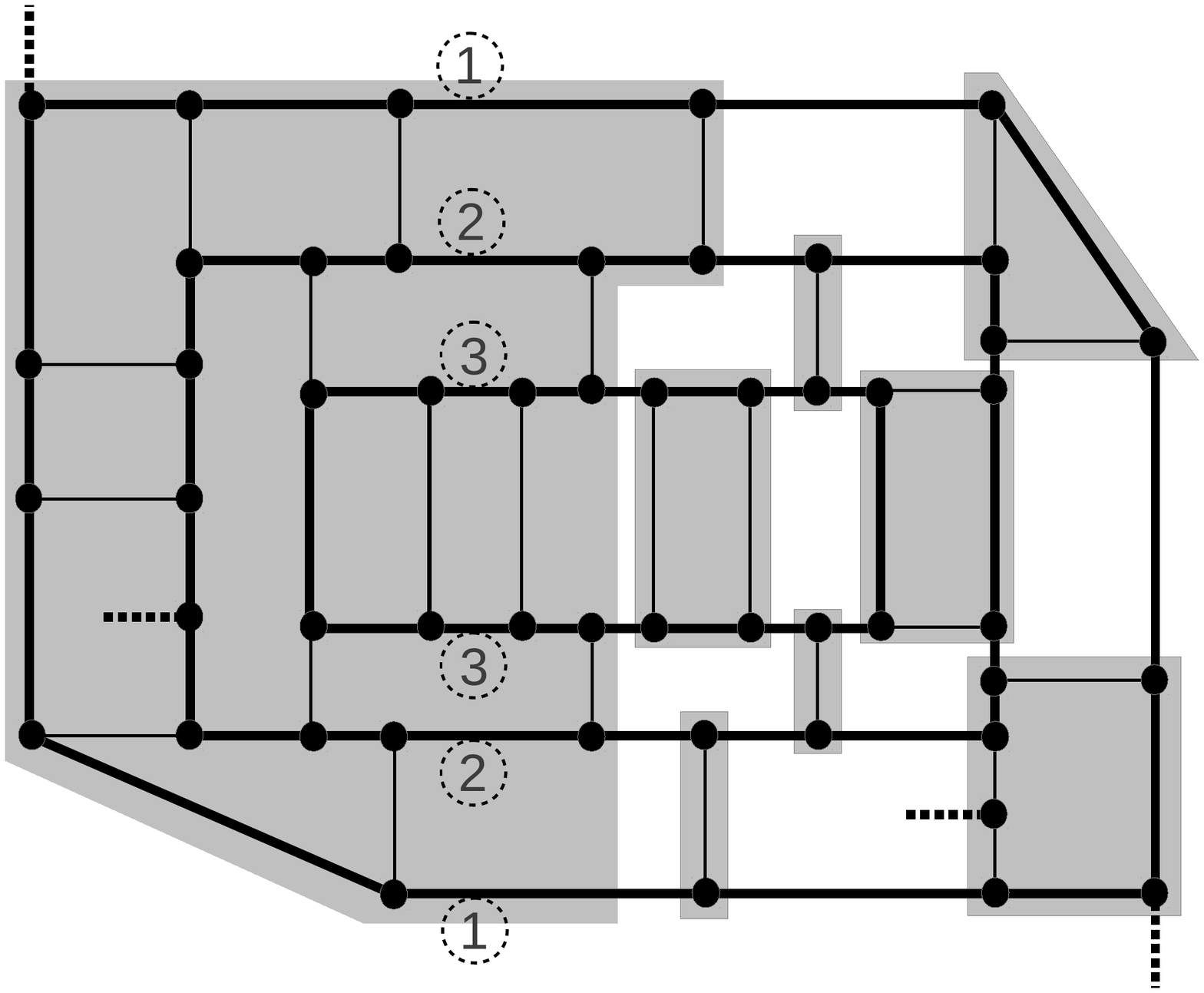}
      \includegraphics[scale=.28,trim=0cm 9.0cm 0cm 3.60cm,clip]%
         {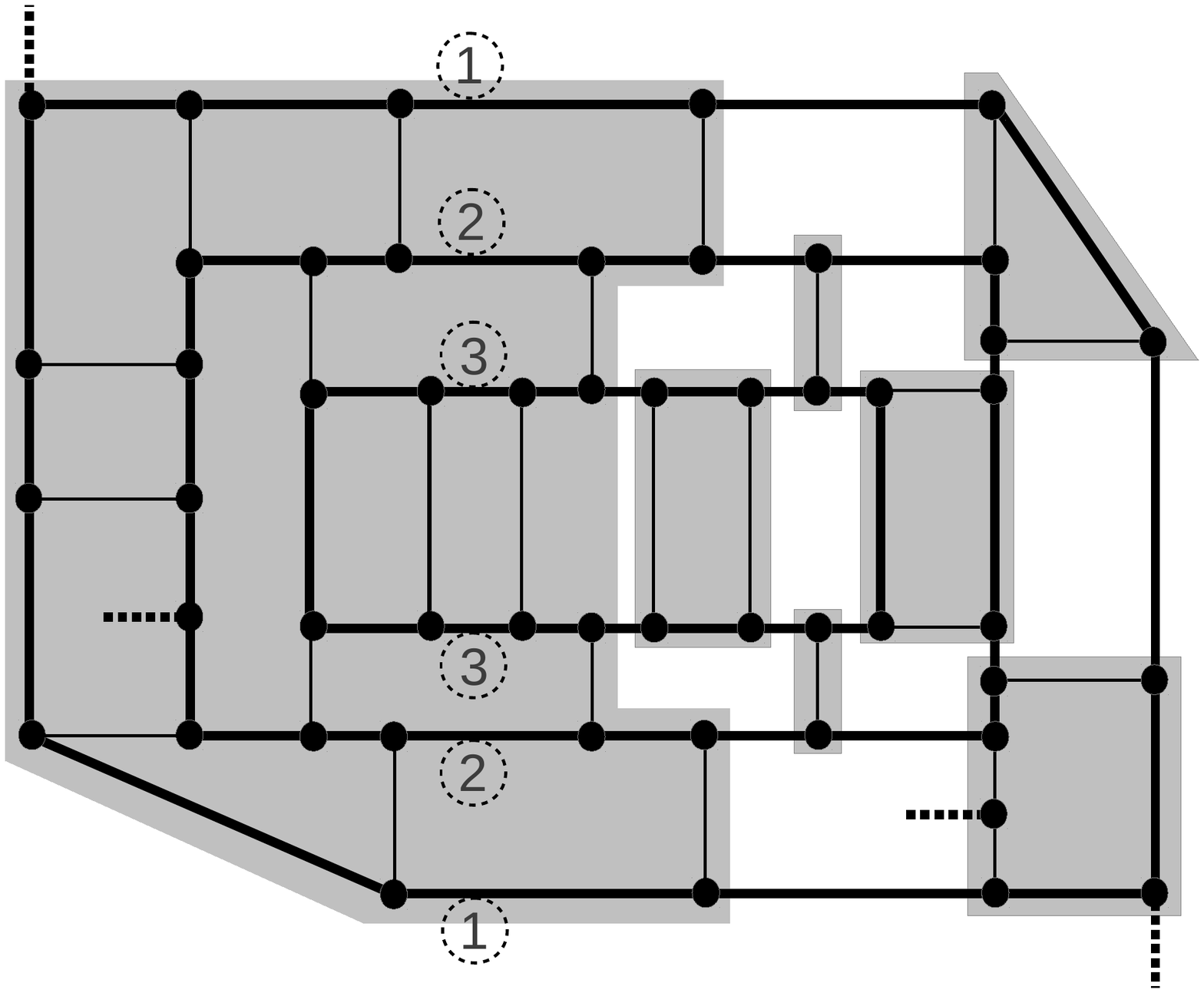}
      \end{custommargins}
    \end{center}
\end{minipage}

\begin{minipage}[b][4.56cm][t]{0.95\linewidth}
   \begin{center}
      \begin{custommargins}{-1.5cm}{-2.5cm}
      \includegraphics[scale=.28,trim=0cm 9.00cm 0cm 3.65cm,clip]%
         {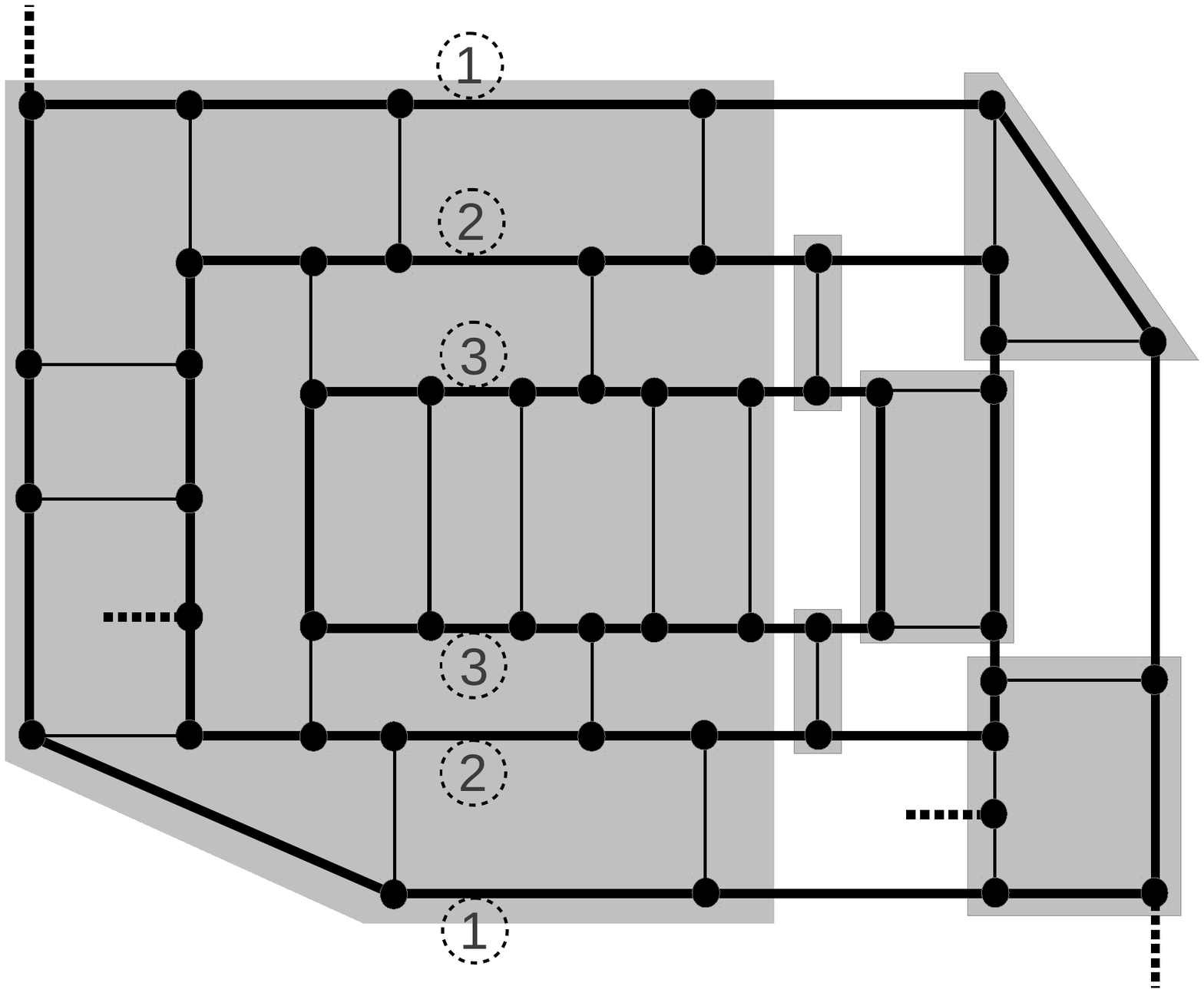}
      \includegraphics[scale=.28,trim=0cm 9.0cm 0cm 3.65cm,clip]%
         {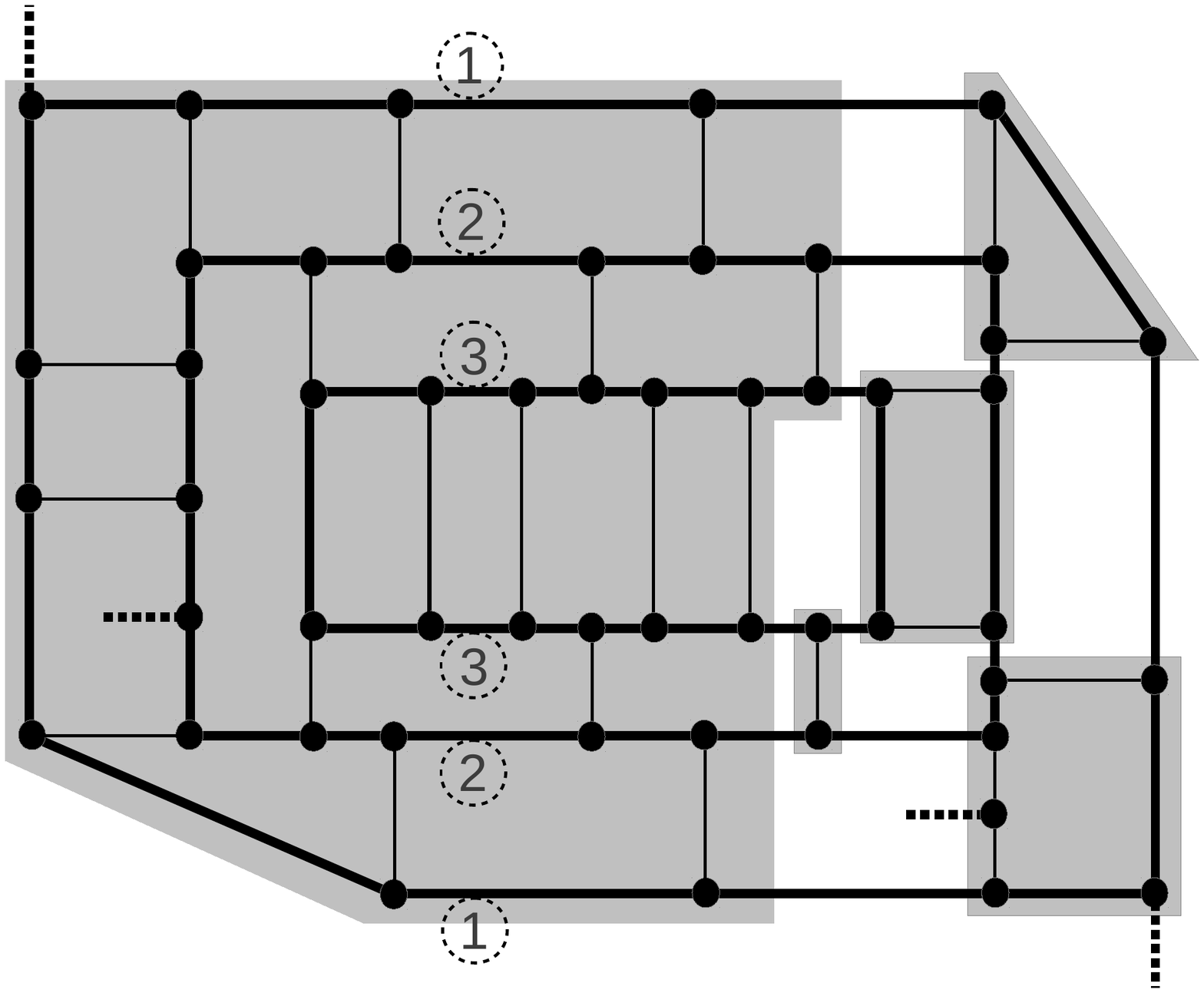}
      \includegraphics[scale=.28,trim=0cm 9.0cm 0cm 3.65cm,clip]%
         {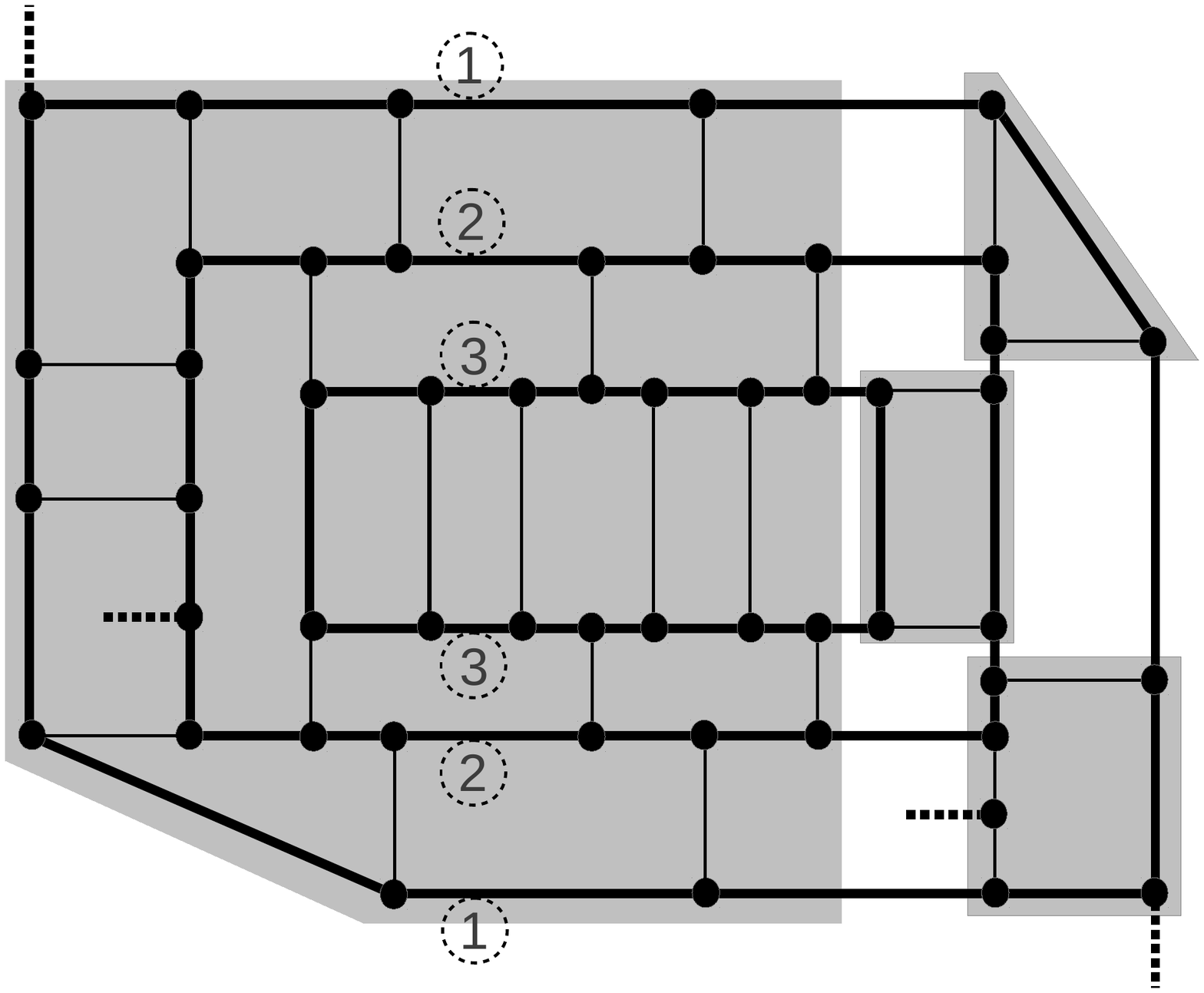}
      \end{custommargins}
    \end{center}
\end{minipage}

\begin{minipage}[b][4.34cm][t]{0.95\linewidth}
   \begin{center}
      \begin{custommargins}{-1.5cm}{-2.5cm}
      \includegraphics[scale=.28,trim=0cm 9.00cm 0cm 3.65cm,clip]%
         {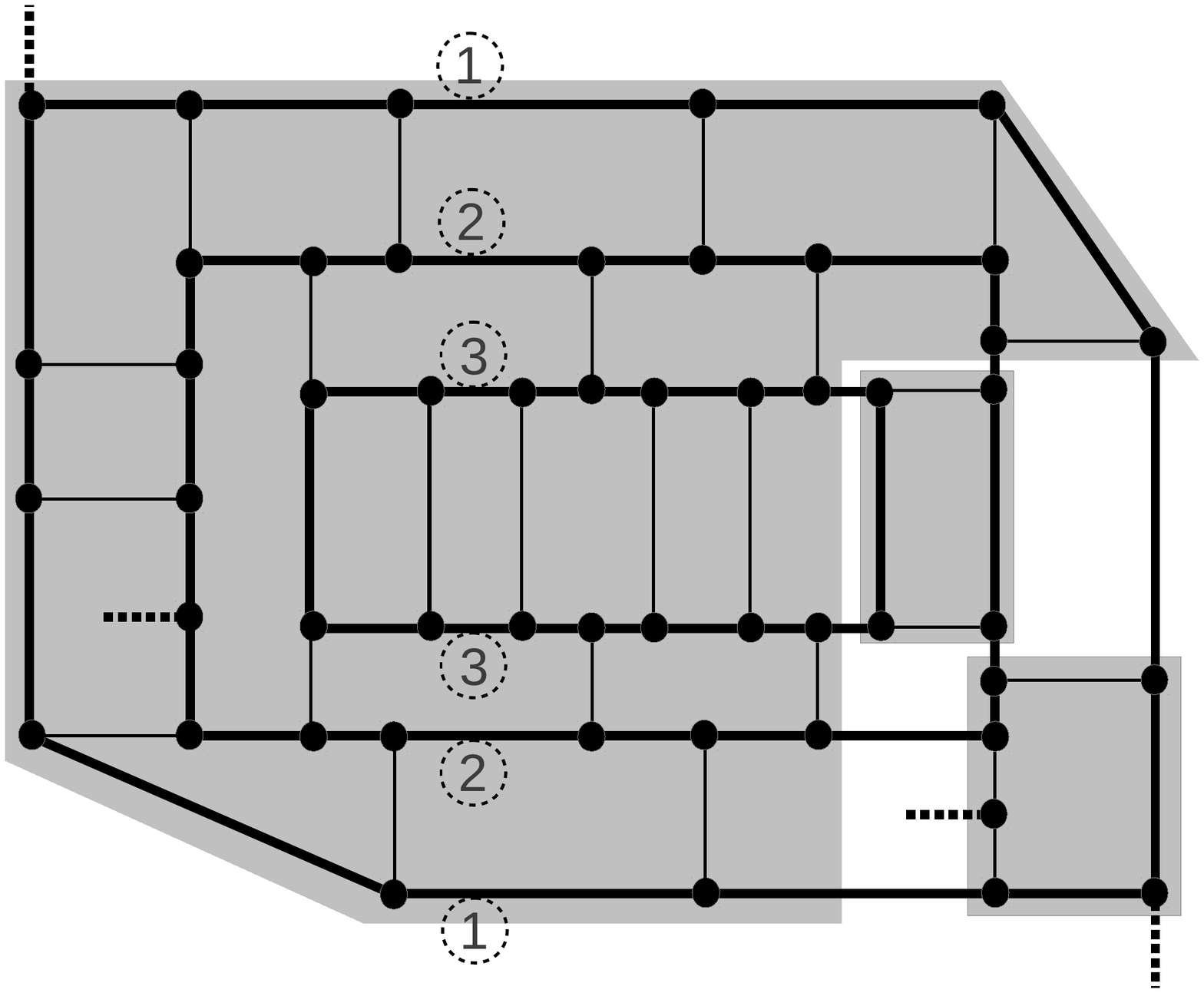}
      \includegraphics[scale=.28,trim=0cm 9.0cm 0cm 3.65cm,clip]%
         {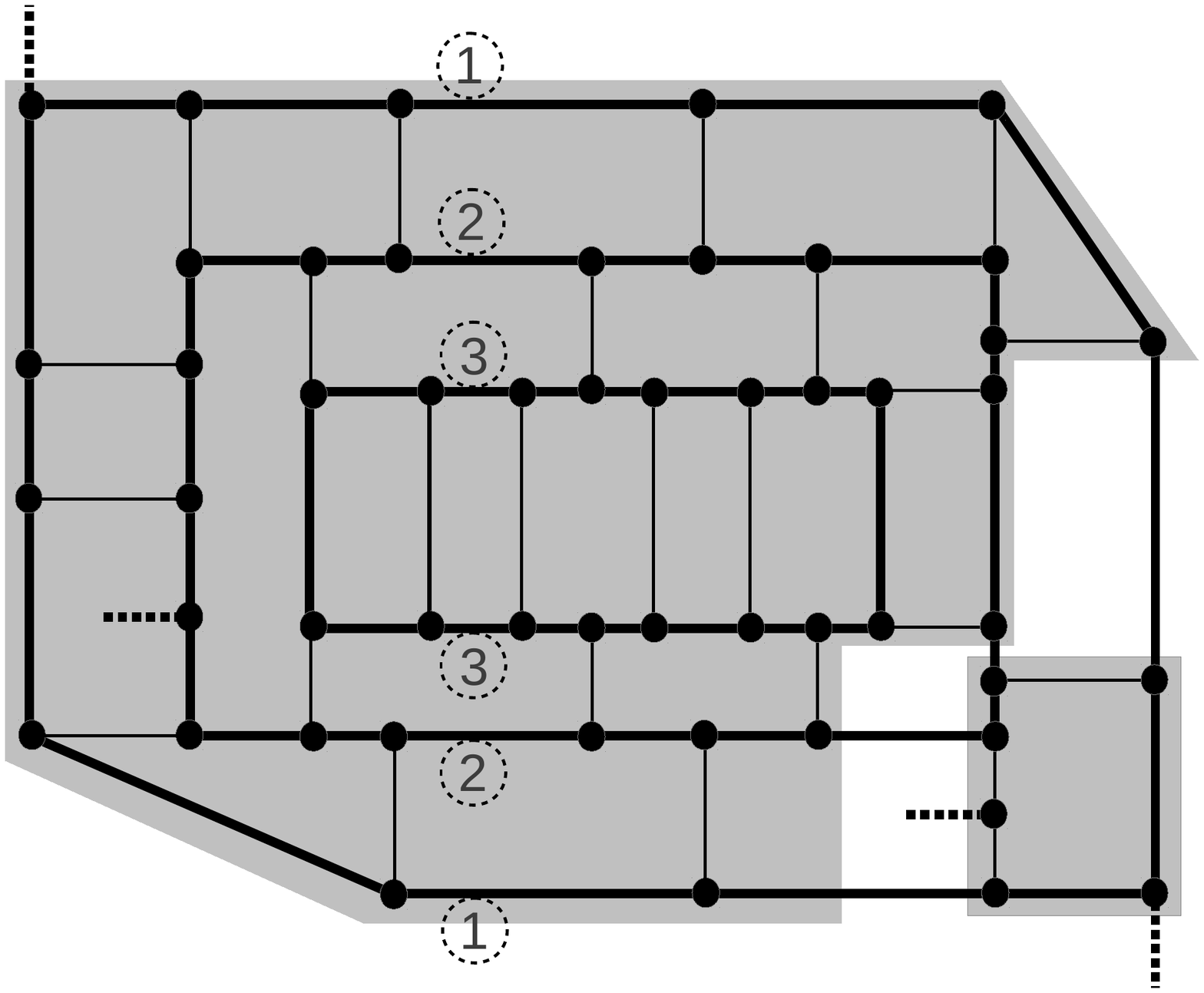}
      \includegraphics[scale=.28,trim=0cm 9.0cm 0cm 3.65cm,clip]%
         {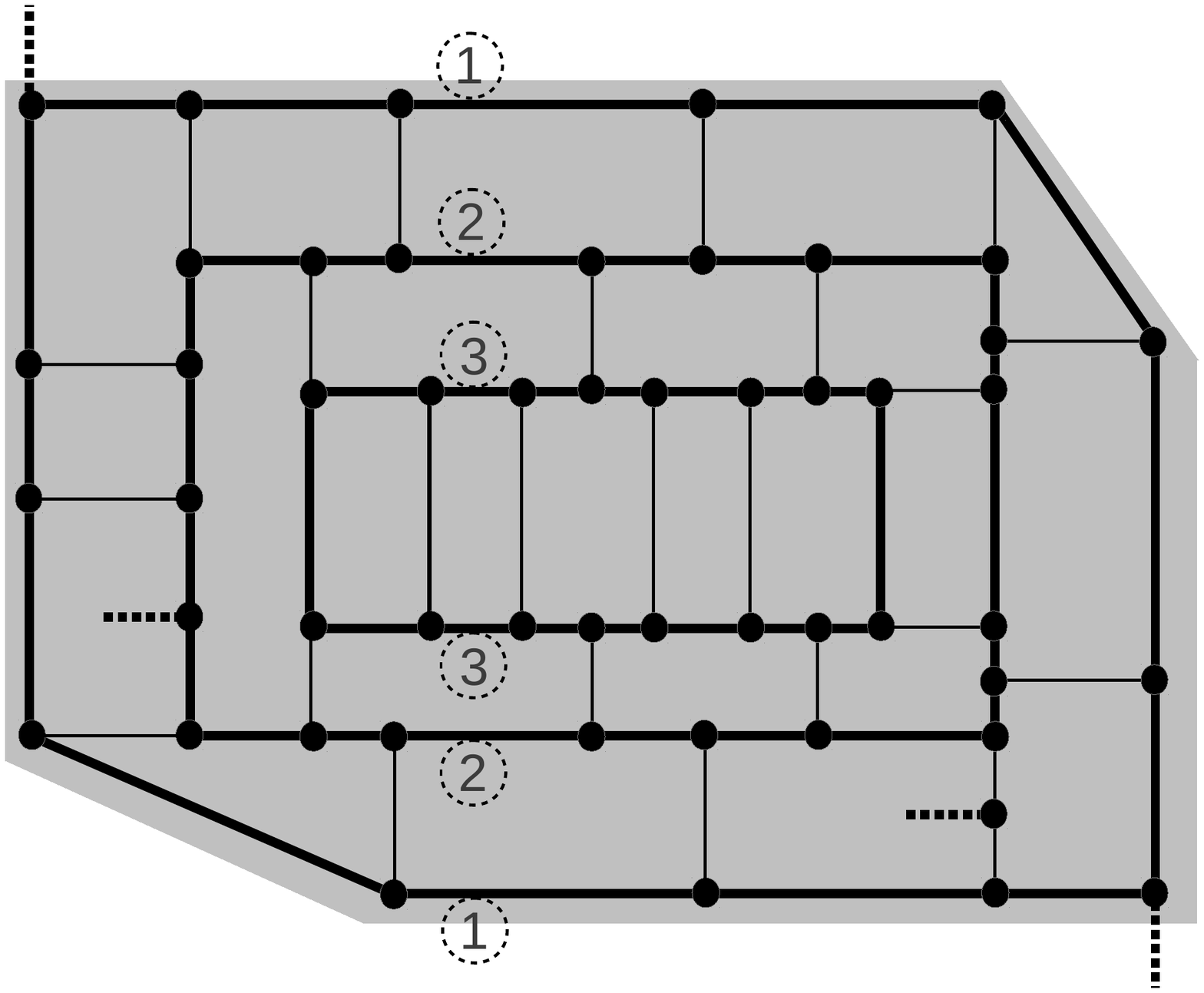}
      \end{custommargins}
    \end{center}
\end{minipage}
\caption{Progress of Algorithm~\ref{alg:bindSchedule} on 
         a good $3$-outerplanar embedding (same as in 
         Figure~\ref{fig:thirty-node-network-on-rectangular-grid})
         during the \textbf{main iteration}.}
\label{fig:Disassemble-and-Reassemble2}
\end{figure}



\begin{figure}[t] 
%
\begin{minipage}[b][4.80cm][t]{0.95\linewidth}
   \begin{center}
      \begin{custommargins}{-1.5cm}{-2.5cm}
      \includegraphics[scale=.4,trim=0cm 15.90cm 0cm 0.60cm,clip]%
         {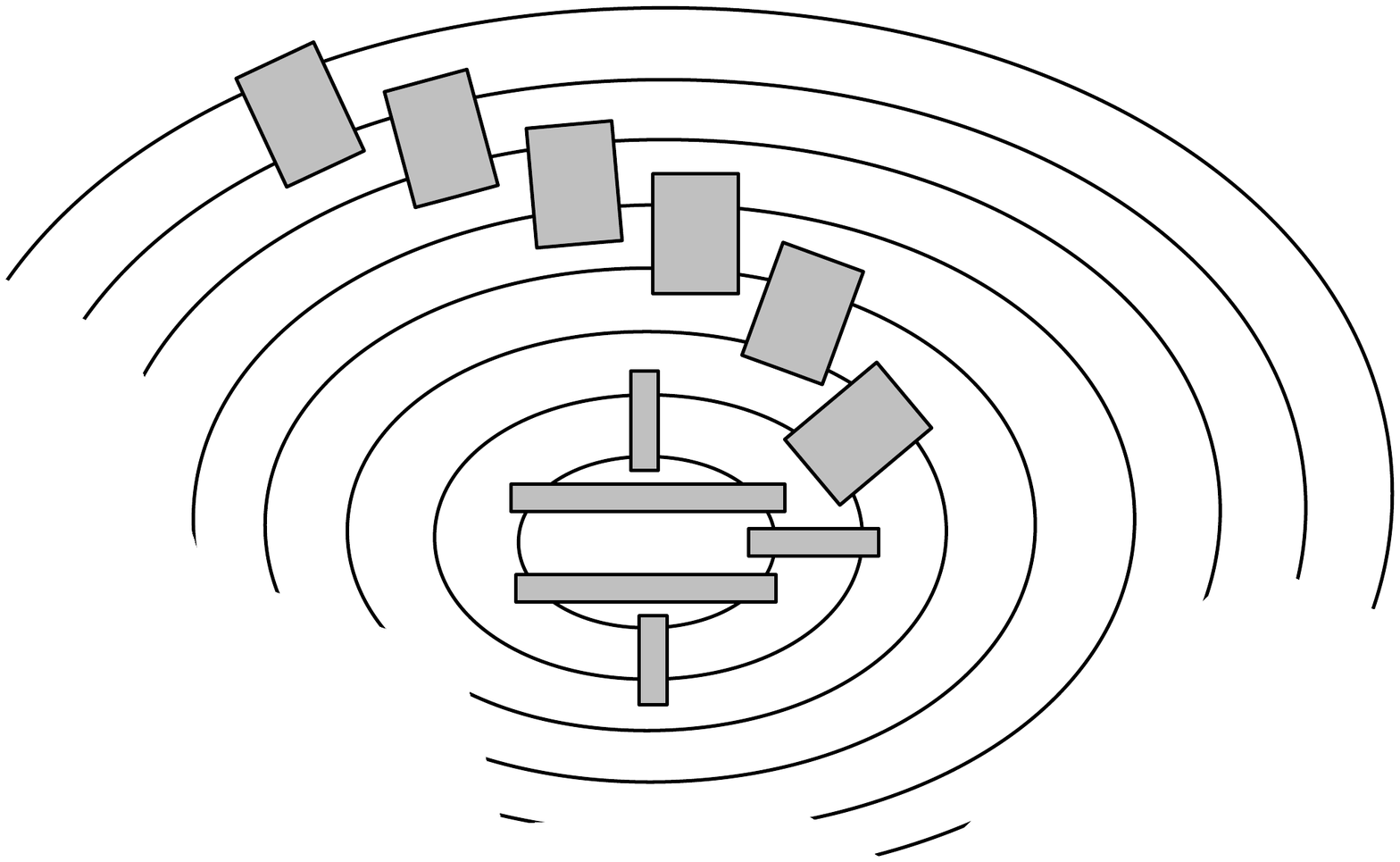}
      \includegraphics[scale=.4,trim=0cm 15.90cm 0cm 0.60cm,clip]%
         {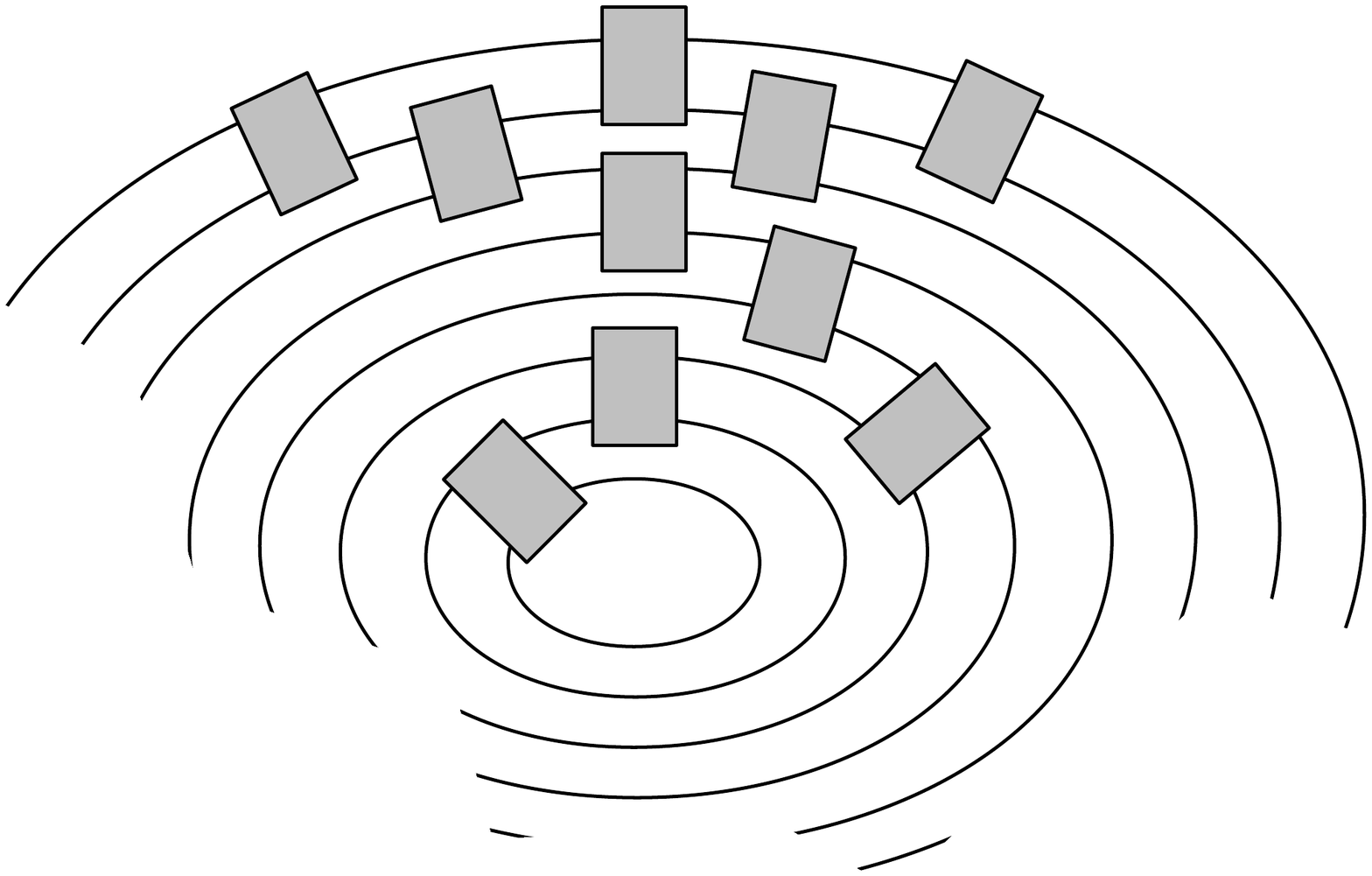}
      \end{custommargins}
    \end{center}
\end{minipage}
\caption{Two possible configurations of assembled subnetworks,
         in a good $8$-outerplanar embedding,
         at the end of the \textbf{second iteration} of 
         Algorithm~\ref{alg:bindSchedule}.}
\label{fig:concentric-outerplanar0}
\end{figure}


\begin{figure}[t] 
%
\begin{minipage}[b][3.80cm][t]{0.95\linewidth}
   \begin{center}
      \begin{custommargins}{-1.5cm}{-2.5cm}
      \includegraphics[scale=.3,trim=0cm 17.00cm 0cm 0.60cm,clip]%
         {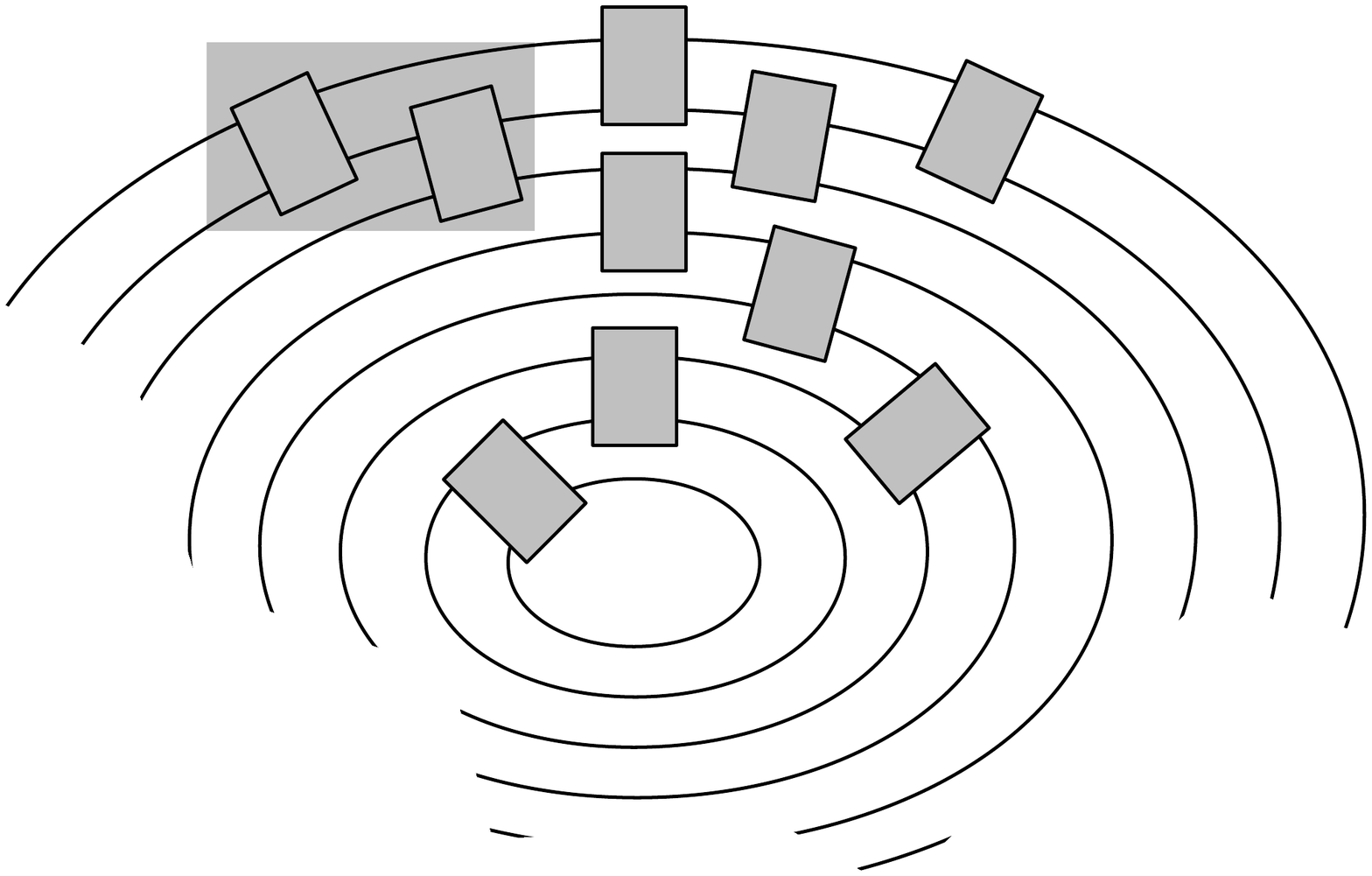}
      \includegraphics[scale=.3,trim=0cm 17.0cm 0cm 0.60cm,clip]%
         {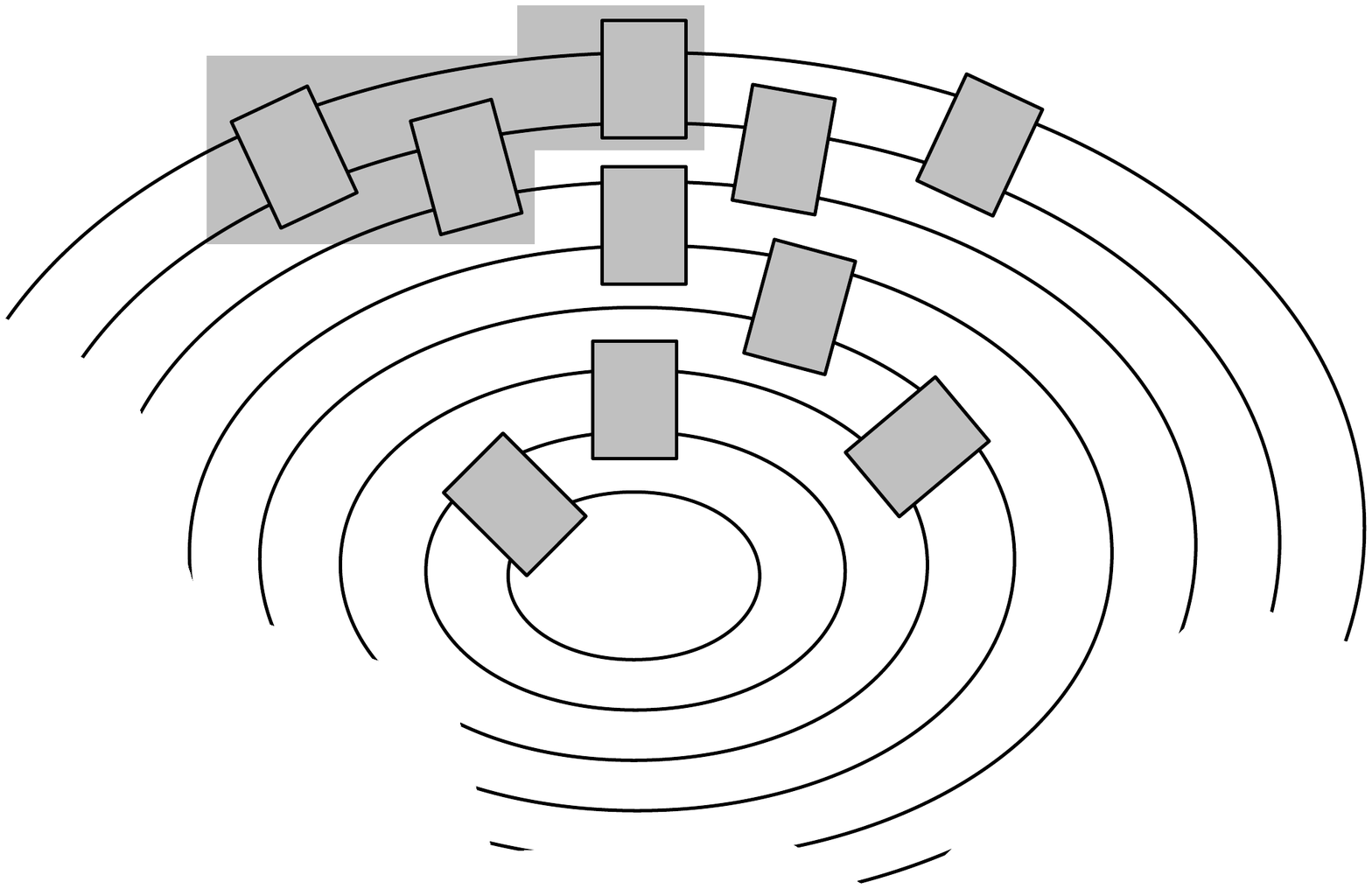}
      \includegraphics[scale=.3,trim=0cm 17.0cm 0cm 0.60cm,clip]%
         {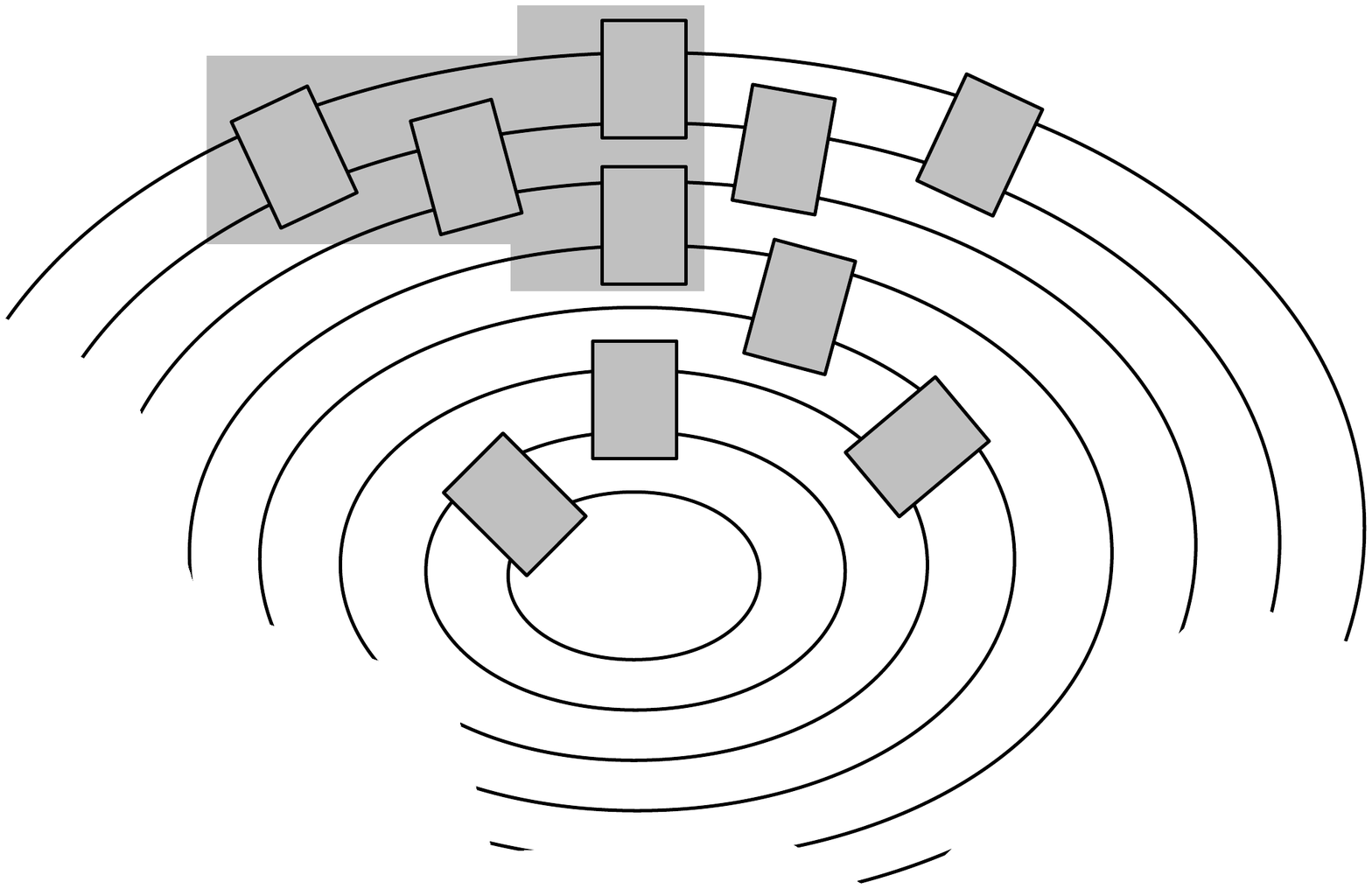}
      \end{custommargins}
    \end{center}
\end{minipage}

\begin{minipage}[b][3.80cm][t]{0.95\linewidth}
   \begin{center}
      \begin{custommargins}{-1.5cm}{-2.5cm}
      \includegraphics[scale=.3,trim=0cm 17.00cm 0cm 0.60cm,clip]%
         {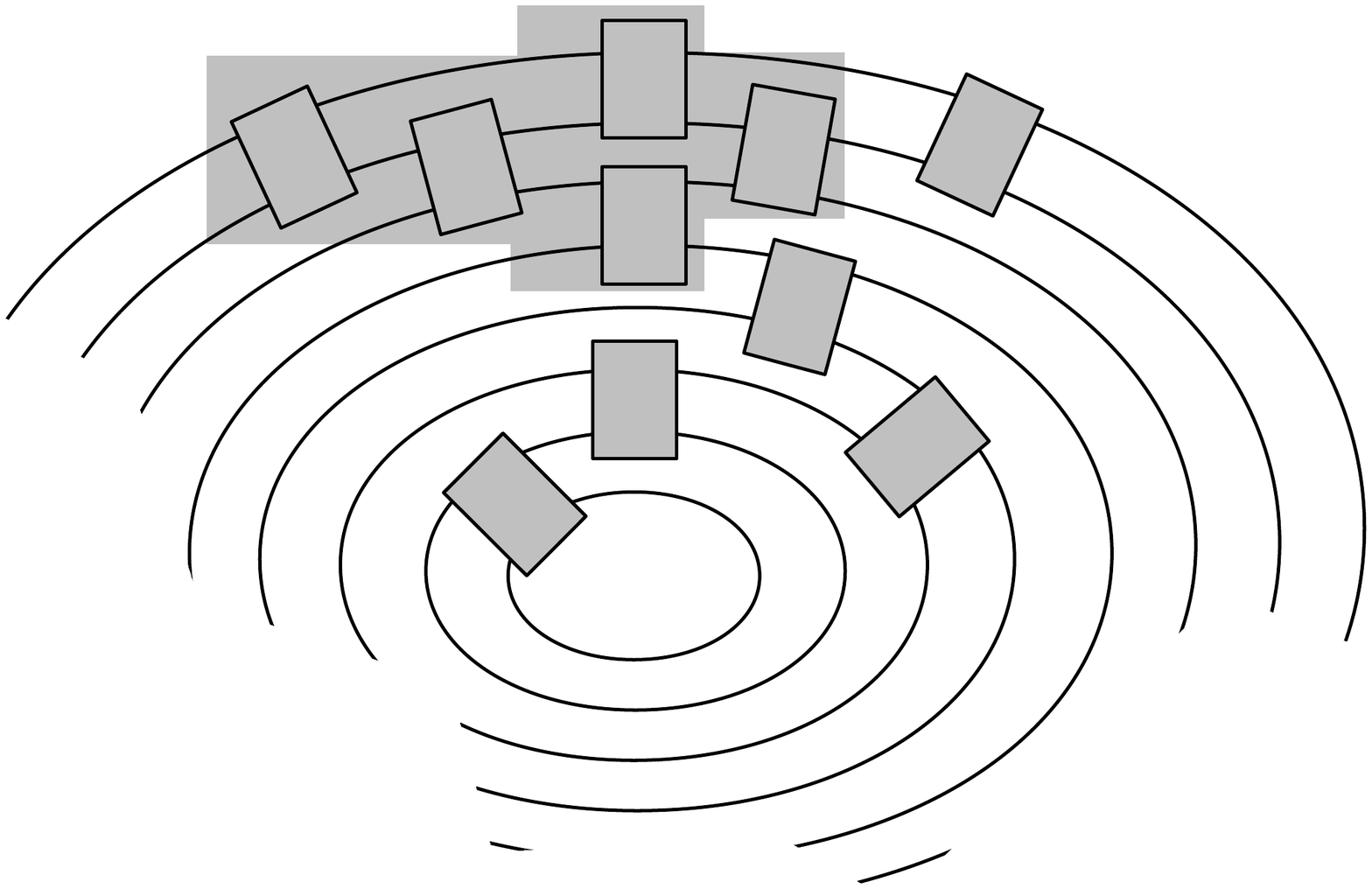}
      \includegraphics[scale=.3,trim=0cm 17.0cm 0cm 0.60cm,clip]%
         {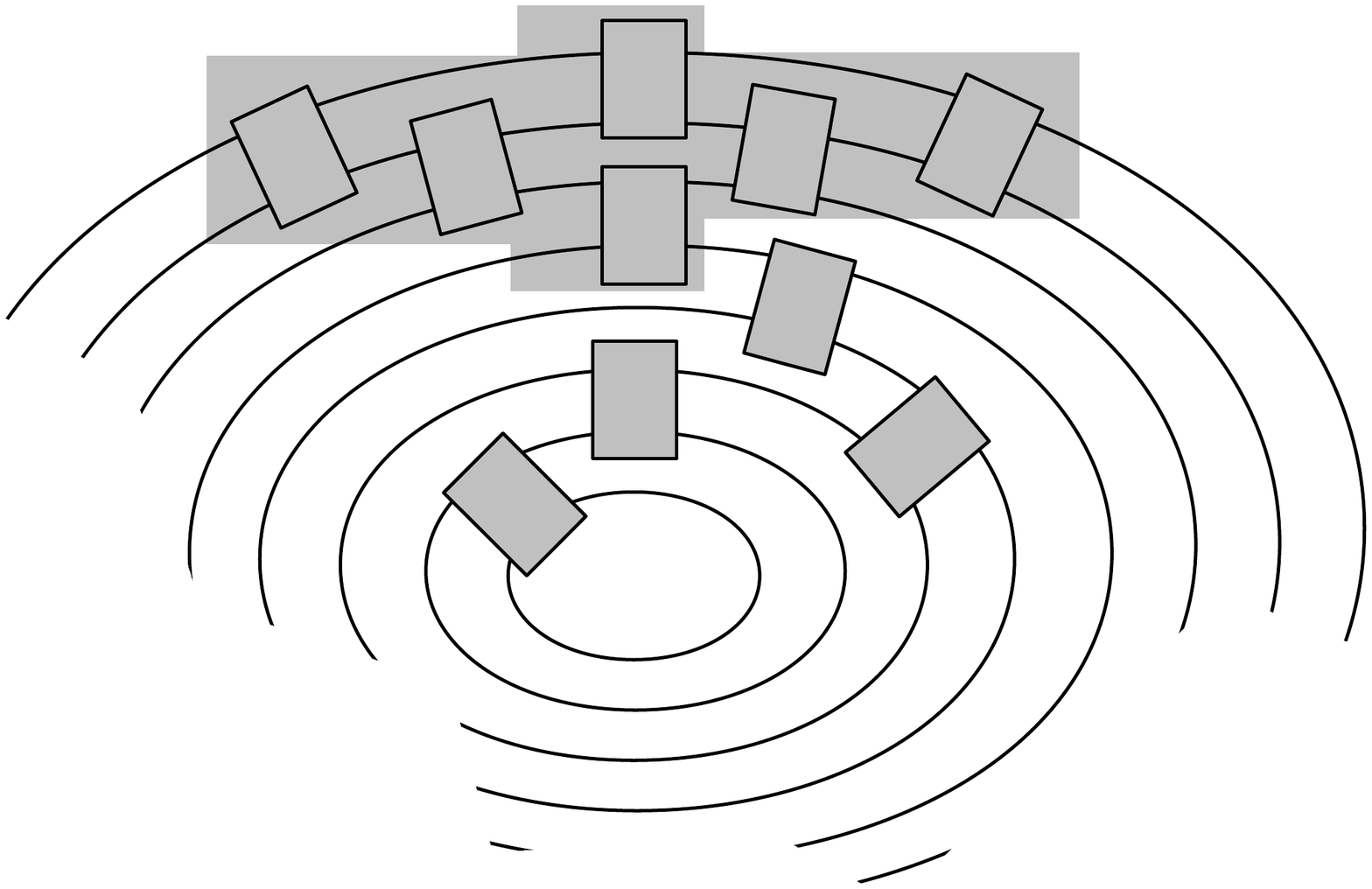}
      \includegraphics[scale=.3,trim=0cm 17.0cm 0cm 0.60cm,clip]%
         {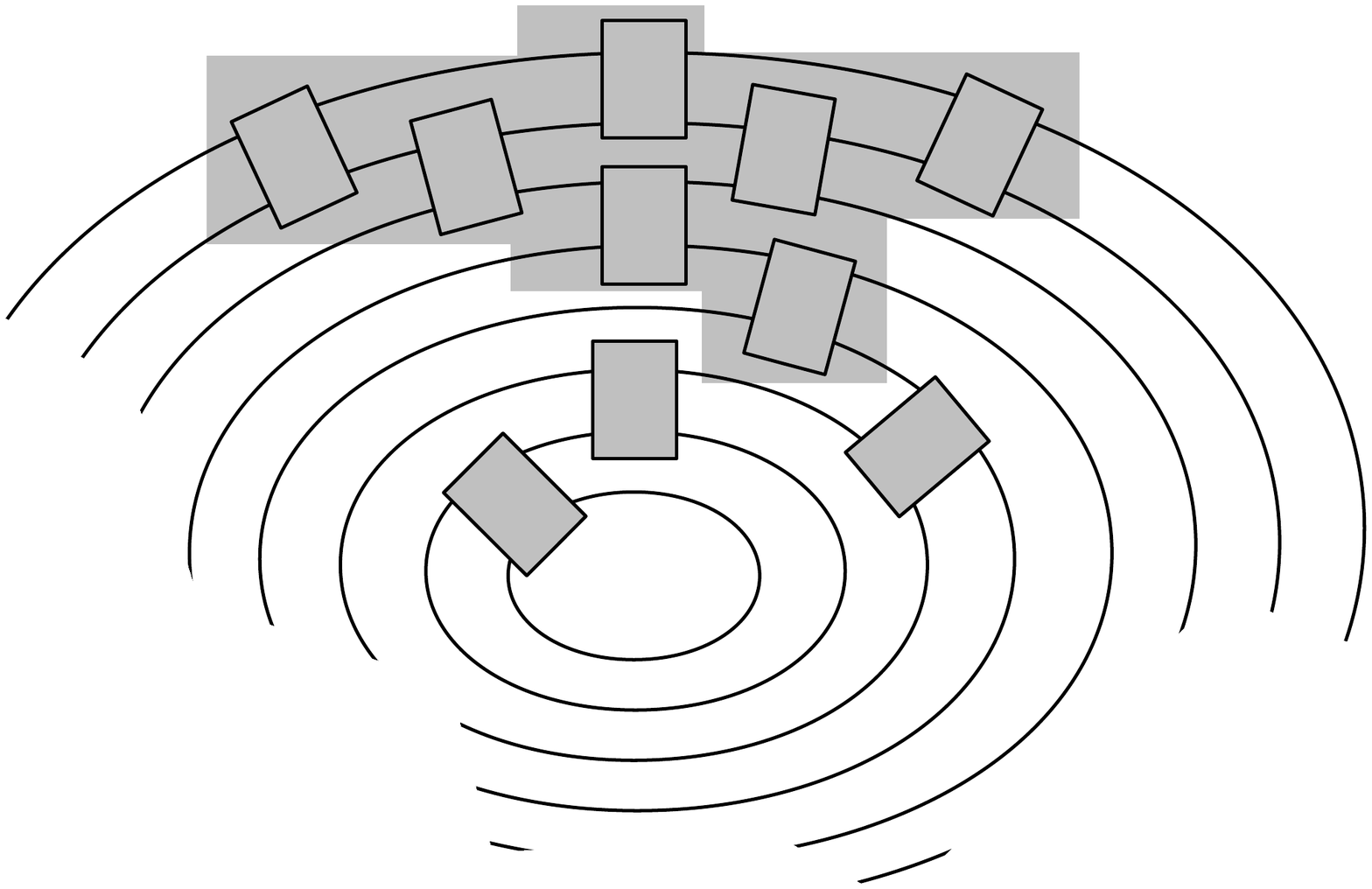}
      \end{custommargins}
    \end{center}
\end{minipage}

\begin{minipage}[b][3.40cm][t]{0.95\linewidth}
   \begin{center}
      \begin{custommargins}{-1.5cm}{-2.5cm}
      \includegraphics[scale=.3,trim=0cm 17.00cm 0cm 0.60cm,clip]%
         {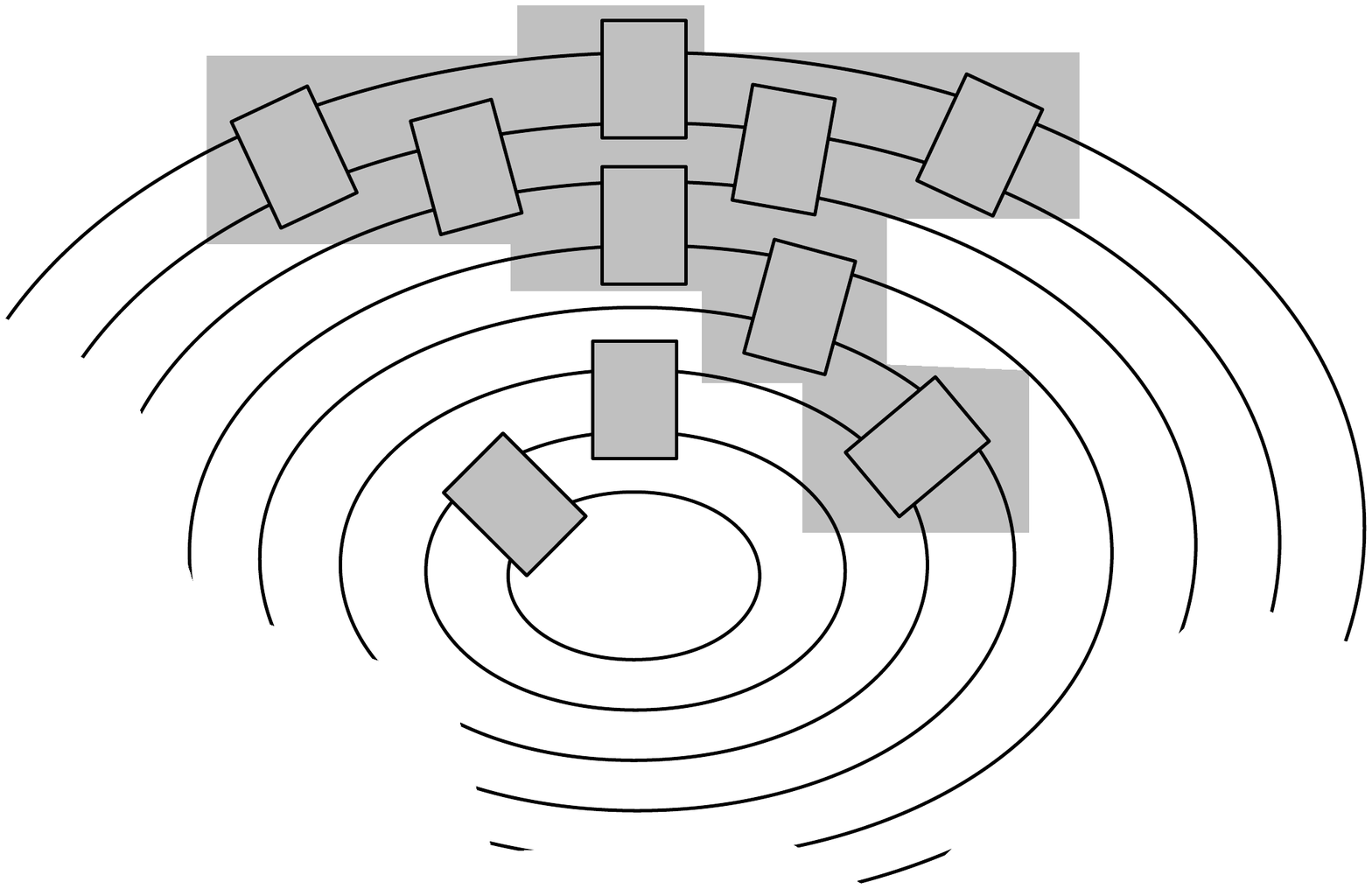}
      \includegraphics[scale=.3,trim=0cm 17.0cm 0cm 0.60cm,clip]%
         {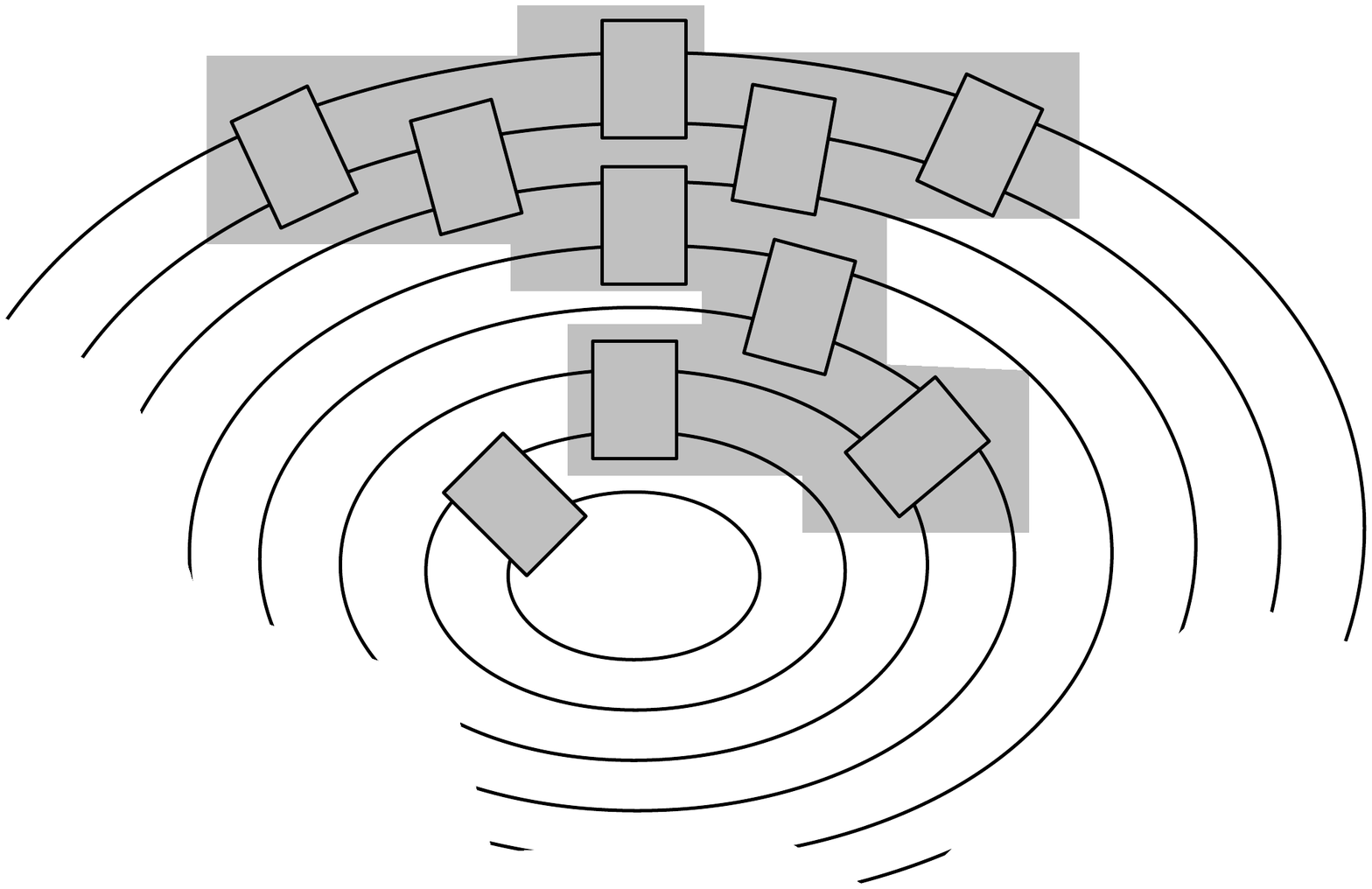}
      \includegraphics[scale=.3,trim=0cm 17.0cm 0cm 0.60cm,clip]%
         {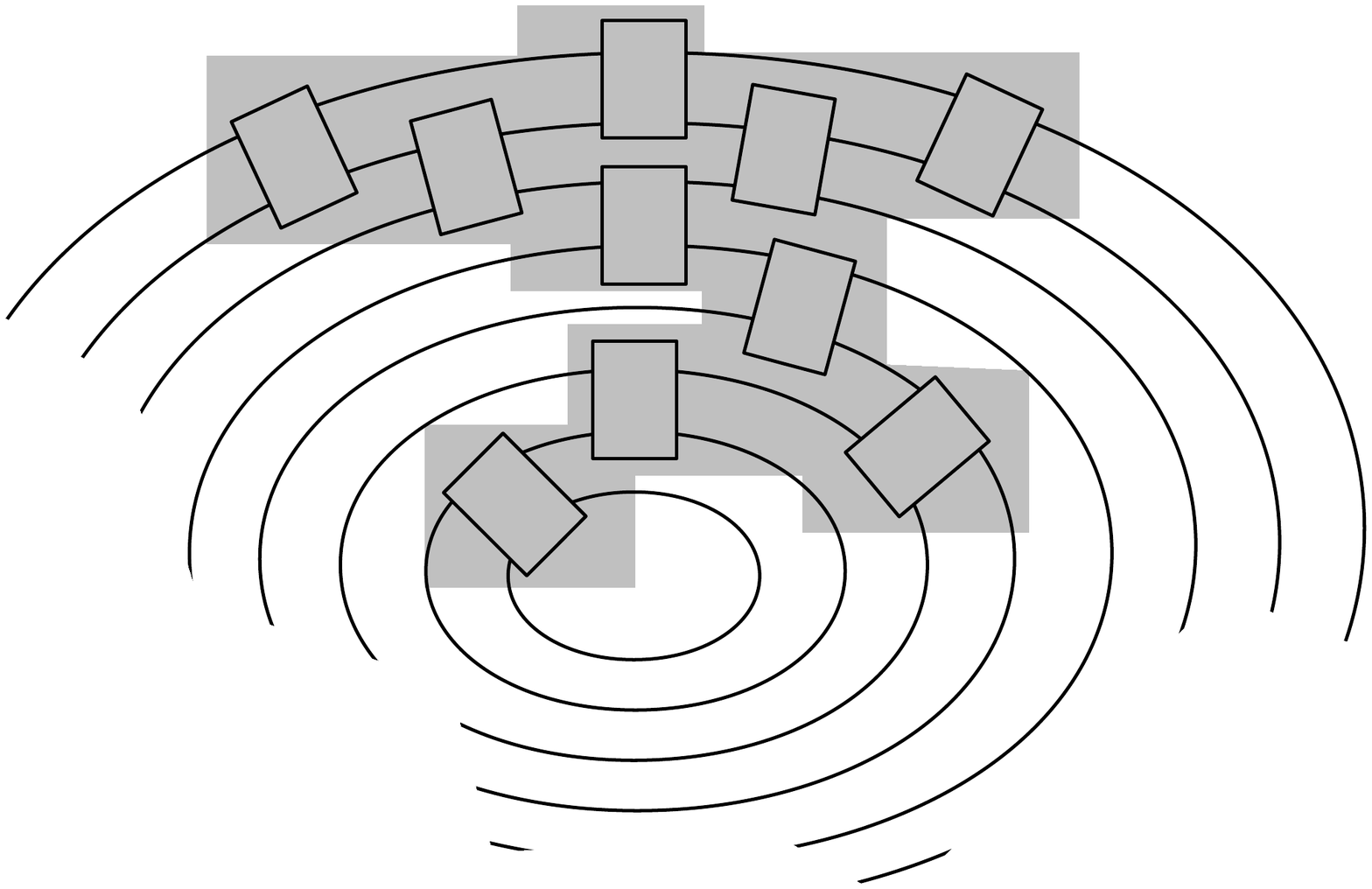}
      \end{custommargins}
    \end{center}
\end{minipage}
\caption{Progress of Algorithm~\ref{alg:bindSchedule}
         during its \textbf{main iteration}, starting from
         the configuration on the right in
         Figure~\ref{fig:concentric-outerplanar0}.}
\label{fig:concentric-outerplanar}
\end{figure}

\clearpage
\vspace{-.2in}
\section{Appendix: Further Comments for  
         \hyperref[sect:extensions-and-future]%
         {Section~\ref*{sect:extensions-and-future}}}
  \label{sect:appendix:future}

Beyond the future work mentioned in
Section~\ref{sect:extensions-and-future}, we here mention two other
areas of future research. These will build on results already obtained
and provide a wider range of useful applications in system modeling
and analysis. The second area below
(Subsection~\ref{sect:angelic-vs-demonic}) was alluded to earlier, in
footnote~\ref{foot:angelic-vs-demonic} and in
Example~\ref{ex:simplest-3}.  
They should separately open a different line of investigation.

\subsection{Algebraic Characterization of Principality}
\label{sect:algebraic-characterization}

A typing $T$ is a function of the form
$T: \power{\aaa_{\text{in,out}}} \to \intervals{\reals}$, but
not every function of this form 
is a typing of some network.
To be a network typing, such a function must satisfy certain 
conditions. For example, it must always be such that
$T(\varnothing) = T(\aaa_{\text{in,out}}) = [0,0] = \Set{0}$. 
Another necessary condition is expressed by the conclusion of
Lemma~\ref{lem:necessity} (there are simple examples, with 
$\size{\aaa_{\text{in,out}}}\geqslant 4$, showing this
condition is not sufficient to make $T$ a principal typing). 
Tasks ahead include the following:
\begin{enumerate}[itemsep=1pt,parsep=2pt,topsep=2pt,partopsep=0pt] 
\item
     Define an \emph{algebraic characterization}, preferably in the form of 
     necessary and sufficient conditions, such that a partial function
     $T: \power{\aaa_{\text{in,out}}} \to\intervals{\reals}$ satisfies
     these conditions iff there exists a network $\N$ of which 
     $T$ is the principal typing.
\item
     Once such an algebraic characterization is established, develop an
     \emph{implementation methodology} which, given a $T:
     \power{\aaa_{\text{in,out}}} \to\intervals{\reals}$ satisfying
     it, can be used to implement $T$ in the form of a network
     $\N$. More precisely, given such a $T$, develop a methodology to
     construct a network $\N$ such that $T$ is the principal typing of
     $\N$.
\item
     Refine this implementation methodology so that it constructs
     a \emph{smallest-size} network $\N$ for which $T$ is the principal 
     typing. Such a network $\N$ can be viewed as the ``best'' implementation
     of the given $T$.
\end{enumerate}
When the external dimension $\size{\aaa_{\text{in,out}}} = 2$, with 
one input arc $a_1$ and one output arc $a_2$, these questions are trivial.
In such a case, $T$ is the principal typing of some network iff there are
numbers $0\leqslant r\leqslant s$ such that
\[
   T(\varnothing) = T(\Set{a_1,a_2}) = [0,0],
   \quad
   T(\Set{a_1}) = [r,s],
   \quad
   T(\Set{a_2}) = [-s,-r].
\]
For such a $T$, there is always a one-node implementation. 

When the external dimension $\size{\aaa_{\text{in,out}}} = 3$, with,
say, input arcs $\Set{a_1,a_2}$ and output arc $a_3$, these questions
are again easy. In such a case, $T$ is the principal typing of some
network iff there are numbers $0\leqslant r_i\leqslant s_i$, for every
$i\in\Set{1, 2, 3}$, such that:
\begin{alignat*}{5}
   & T(\varnothing)\ &&=\ &&\ \ T(\Set{a_1,a_2,a_3})\ &&=\ &&[0,0],
   \\[.6ex]
   & T(\Set{a_1})\ &&=\ && -T(\Set{a_2,a_3})\ &&=\ &&[r_1,s_2],   
   \\[.6ex]
   & T(\Set{a_2})\ &&=\ && -T(\Set{a_1,a_3})\ &&=\ &&[r_2,s_2],
   \\[.6ex]
   &T(\Set{a_3})\ &&=\ && -T(\Set{a_1,a_2})\ &&=\ &&[-s_3,-r_3], 
\end{alignat*}
where $r_1+r_2 = r_3$ and $\max \Set{s_1,s_2} \leqslant s_3\leqslant s_1+s_2$.
For such a $T$, there is always a one-node implementation.

The problem becomes interesting and non-trivial when
$\size{\aaa_{\text{in,out}}} \geqslant 4$. For a sense of the
difficulty in such a case, consider the network $\N_2$ in
Example~\ref{ex:simplest-2}. It is not the smallest-size
implementation of the typing $T_2$ in Example~\ref{ex:simplest-2},
as illustrated by the next example.

\begin{example}
\label{ex:smallest-size-implementation}
The network $\N_4$ in Figure~\ref{fig:equivalent-to-N2} was obtained
by brute-force trial-and-error. It is equivalent to $\N_2$ in
Example~\ref{ex:simplest-2}, and qualifies as a better implementation
of $T_2$, because $\N_4$ has fewer nodes than $\N_2$, with $6$ nodes
in $\N_4$ against $8$ nodes in $\N_2$. 
(We can also compare network sizes by counting both
nodes and arcs: $6+12 = 18$ in $\N_4$ against $8+16 = 24$ in $\N_2$.)
\end{example}

\begin{figure}[ht] 
\begin{center}
\includegraphics[scale=.28,trim=0cm 16.50cm 0cm 1.5cm,clip]%
    {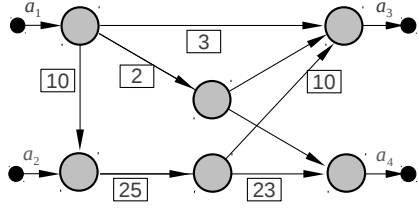}
\caption{For Example~\ref{ex:smallest-size-implementation}:
         Network $\N_4$ is equivalent to $\N_2$ in 
         Example~\ref{ex:simplest-2}, but is a better
         implementation of the same principal typing $T_2$. 
         Missing lower bounds are $0$, missing upper bounds are 
         a ``very large number'' $K$.}
\label{fig:equivalent-to-N2}
\end{center}
\end{figure}

We conjecture that, if $\size{\aaa_{\text{in,out}}} = 4$ and
$T: \power{\aaa_{\text{in,out}}} \to\intervals{\reals}$ is the
principal typing of some network with external arcs
$\aaa_{\text{in,out}}$, then there is a smallest-size implementation
of $T$ requiring at most $6$ nodes.


\subsection{Angelic Non-Determinism \emph{versus} Demonic Non-Determinism}
\label{sect:angelic-vs-demonic}

Suppose $\AAA$ is a large assembly of networks containing 
network $\M$ as a subnetwork. Under what conditions can we 
safely substitute another network $\N$ for $\M$?
A minimal requirement is that $\M$ and $\N$ are similar, \ie, they
have the same number of input arcs and the same number of output
arcs. If we are given principal typings $T$ and $U$ for $\M$ and
$\N$, respectively, we should have enough information to decide
whether the substitution is safe.
To simplify a little, let the input and output arcs of $\M$ and $\N$
be $\aaa_{\text{in}} =\Set{a_1,a_2}$ and $\aaa_{\text{out}}
=\Set{a_3,a_4}$.  If the substitution of $\N$ for $\M$ is safe, then
$\N$ should be able to consume every input flow that $\M$ is able to
consume, \ie, if an input assignment
$f_{\text{in}}: \Set{a_1,a_2}\to\nreals$ satisfies
$\rest{T}{\power{\Set{a_1,a_2}}}$, then it must also satisfy
$\rest{U}{\power{\Set{a_1,a_2}}}$. Hence, the following inclusions are a
reasonable requirement for safe substitution:
\begin{itemize}
\item[$(\dag)$]
\quad $T(\Set{a_1})\subseteq U(\Set{a_1}),\quad
   T(\Set{a_2})\subseteq U(\Set{a_2}), \quad\text{and}
   \quad T(\Set{a_1,a_2})\subseteq U(\Set{a_1,a_2}). $
\end{itemize}
Symmetrically, for safe substitution, every output flow produced by
$\N$ should not exceed the limits of an output flow produced by
$\M$, \ie, if an input assignment
$f_{\text{out}}: \Set{a_3,a_4}\to\nreals$ satisfies
$\rest{U}{\power{\Set{a_3,a_4}}}$, then it must also satisfy
$\rest{T}{\power{\Set{a_3,a_4}}}$.  Hence, another reasonable requirement
consists of the following reversed inclusions:
\begin{itemize}
\item[$(\ddag)$]
\quad $T(\Set{a_3})\supseteq U(\Set{a_3}),\quad
   T(\Set{a_4})\supseteq U(\Set{a_4}), \quad\text{and}
   \quad T(\Set{a_3,a_4})\supseteq U(\Set{a_3,a_4}). $
\end{itemize}
If $U$ satisfies both $(\dag)$ and $(\ddag)$, is
the substitution of $\N$ for $\M$ in $\AAA$ safe? It depends.
Conditions $(\dag)$ and $(\ddag)$ are necessary, but
there are other issues which we elaborate in the next 
example.%
    \footnote{In a different context
    (strongly-typed programming languages),
    conditions $(\dag)$ and $(\ddag)$ resemble the
    conditions for making $U$ a subtyping of $T$ and, accordingly, an
    object of typing $U$ to be safely substituted for an object of
    typing $T$.  Specifically, $(\dag)$ mimics
    the \emph{contravariance} in the domain $\tau_1$ of a function
    type $\tau_1\to\tau_2$, and $(\ddag)$ mimics the \emph{covariance}
    in the co-domain $\tau_2$ of the same type $\tau_1\to\tau_2$, in a
    strongly-typed functional language. However, $(\dag)$ and $(\ddag)$ are not
    sufficient for safe substitution here, because there are
    dependences between input types and output types in our networks
    that do not occur in a strongly-typed functional language.}

\begin{example}
\label{ex:angelic-vs-demonic}
In the larger assembly $\AAA$ described above, let 
$\M = \N_{3}$ from Example~\ref{ex:simplest-3} and
$\N = \N_2$ from Example~\ref{ex:simplest-2}.
We have the following relationship $\subT{T_2}{T_3}$, where
``$\subTsym$'' is the subtyping relation, defined in
Examples~\ref{ex:simplest-2} and~\ref{ex:simplest-3}. More, in fact, 
$T= T_{3}$ and $U=T_2$ satisfy both conditions $(\dag)$ and $(\ddag)$, 
which are therefore not sufficient to prevent
the unsafe situation we now describe.

As we explain below, if $\N_2$ operates in a way to preserve the
feasibility of flows in $\AAA$, \ie, if it operates \emph{angelically}
and tries to keep $\AAA$ in good working order, then replacing
$\N_{3}$ by $\N_2$ is safe. However, if $\N_2$ makes choices
that disrupt $\AAA$'s good working order, maliciously or
unintentionally, \ie, if it operates \emph{demonically} and violates
the feasibility of flows in $\AAA$, then the substitution is unsafe.
This can happen because for the same assignment $f_{\text{in}}$ to the
input arcs (resp., the same assignment $f_{\text{out}}$ to the output arcs),
corresponds several possible output assignments $f_{\text{out}}$
(resp., input assignments $f_{\text{in}}$), without violating any
of $\N_2$'s internal constraints.

Suppose $\N_{3}$ in $\AAA$ is prompted to \emph{consume} some flow
entering at input arcs $a_1$ and $a_2$. (A similar and symmetric
argument can be made when $\N_{3}$ is asked to \emph{produce} some
flow at output arcs $a_3$ and $a_4$.) Suppose the incoming flow is
given by the assignment $f_{\text{in}}(a_1) = 15$ and
$f_{\text{in}}(a_2) = 0$. Flow is then pushed along the internal arcs
of $\N_{3}$, respecting capacity constraints and flow conservation at
nodes.  There are many different ways in which flow can be pushed
through.  By direct inspection, relative to the given $f_{\text{in}}$,
the largest possible quantity exiting at output arc $a_4$ is $10$. So,
relative to the given $f_{\text{in}}$, the output assignment which is
most skewed in favor of $a_4$ is $f_{\text{out}}(a_3) = 5$ and
$f_{\text{out}}(a_4) = 10$. 

Under the assumption that $\AAA$ works safely with $\N_{3}$ inserted,
we take this conclusion to mean that any output quantity exceeding
$10$ at arc $a_4$, when $f_{\text{in}}(a_1) = 15$ and
$f_{\text{in}}(a_2) = 0$, disrupts $\AAA$'s overall operation.
For a concrete situation,  when $f_{\text{in}}(a_1) = 15$ and
$f_{\text{in}}(a_2) = 0$, it can occur that the $10$ units exiting from $a_4$ 
enter some node $\nu$ in $\AAA$ and cannot be increased without violating 
a capacity constraint on an arc exiting $\nu$.

Next, suppose we substitute $\N_2$ for $\N_{3}$ and examine
$\N_2$'s behavior with the same $f_{\text{in}}(a_1) = 15$ and
$f_{\text{in}}(a_2) = 0$. By inspection, the flow that is most skewed
in favor of $a_4$ gives rise to the output assignment
$f_{\text{out}}(a_3) = 3$ and $f_{\text{out}}(a_4) = 12$. In this
case, the output quantity at $a_4$ exceeds $10$, which, as argued
above, is disruptive of $\AAA$'s overall operation.  Note that the
presumed disruption occurs in the enclosing context that is part of
$\AAA$, not inside $\N_2$ itself, where flow is still directed by
respecting flow conservation at $\N_2$'s nodes and
lower-bound/upper-bound capacities at $\N_2$'s arcs. Thus, $\N_2$'s
harmful behavior is not the result of violating its own internal
constraints, but of its malicious or (unintended) faulty interaction
with the enclosing context.

Consider now a slight adjustment of $\N_{3}$, call it $\N_3'$,
where we make a single change in $\N_{3}$, namely, in the
upper-bound capacity of input arc $a_1$:
\emph{Decrease $\uc(a_1)$ from $K$ (``very large number'') to $10$.}
The typing $T_{3}$ is no longer principal for $\N_3'$.
We compute a new principal typing
$T_3'$ for $\N_3'$ which, in addition to the type
assignments $T_3'(\varnothing) = T_3'(\Set{a_1,a_2,a_3,a_4}) = [0,0]$, 
makes the following type assignments:
\begin{alignat*}{5}
  & \underline{a_1:[0,10]}\quad && a_2:[0,25]\quad 
  && -a_3:[-15,0]\quad && -a_4:[-25,0]
\\[.8ex] 
  & a_1+a_2:[0,30]\quad && a_1-a_3:[-10,10]\quad
  && \underline{a_1-a_4:[-23,10]}
\\[.8ex] 
  & \underline{a_2-a_3:[-10,23]}\quad && a_2-a_4:[-10,10]
  && -a_3-a_4:[-30,0]\quad 
\\[.8ex] 
  & a_1+a_2-a_3: [0,25]\qquad&& a_1+a_2-a_4:[0,15] \qquad
   && a_1-a_3-a_4: [-23,0]\qquad && \underline{a_2-a_3-a_4: [-10,0]}\quad
\end{alignat*}
The underlined type assignments here are those that differ from
the corresponding type assignments made by $T_{3}$. 
It is easy to check that, however demonically $\N_2$ chooses to push
flow through its internal arcs, the substitution of $\N_2$ for
$\N_3'$ is ``input safe''; \ie, for every input assignment
$f_{\text{in}} : \Set{a_1,a_2}\to\nreals$ satisfying 
$\rest{T_3'}{\power{\Set{a_1,a_2}}}$, 
and every extension $g : \Set{a_1,a_2,a_3,a_4}\to\nreals$ of $f_{\text{in}}$, 
the IO assignment $g$ satisfies $T_2$ iff $g$ satisfies $T_3'$.

Similarly, consider an outgoing flow in $\N_3$ given by the 
assignment $f_{\text{out}}(a_3) = 0$ and $f_{\text{out}}(a_4) = 25$.
Relative to this $f_{\text{out}}$, consider the entering flow at $\Set{a_1,a_2}$
which is most skewed in favor of $a_1$. By inspection, this is
the input assignment $f_{\text{in}}(a_1) = 10$ and $f_{\text{in}}(a_2) = 15$.
By contrast in $\N_2$, if $f_{\text{out}}(a_3) = 0$ and $f_{\text{out}}(a_4) = 25$,
then the corresponding input assignment which is most skewed
in favor of $a_1$ is $f_{\text{in}}(a_1) = 12$ and $f_{\text{in}}(a_2) = 13$.

We can adjust $\N_{3}$, to define
another network $\N_3''$, for which the substitution of $\N_2$
is ``output safe''. $\N_3''$ is obtained by making a single change:
\emph{Decrease $\uc(a_4)$ from $K$ (``very large number'') to $10$.}
The principal typing $T_3''$ for $\N_3''$ makes the type
assignments $T_3''(\varnothing) = T_3''(\Set{a_1,a_2,a_3,a_4}) = [0,0]$
in addition to:
\begin{alignat*}{5}
  & a_1:[0,15]\quad && \underline{a_2:[0,20]}\quad 
  && -a_3:[-15,0]\quad && \underline{-a_4:[-10,0]}
\\[.8ex] 
  & \underline{a_1+a_2:[0,25]}\quad && a_1-a_3:[-10,10]\quad
  && \underline{a_1-a_4:[-10,15]}
\\[.8ex] 
  & \underline{a_2-a_3:[-15,10]}\quad && a_2-a_4:[-10,10]
  && -a_3-a_4:[-25,0]\quad 
\\[.8ex] 
  & \underline{a_1+a_2-a_3: [0,10]}\qquad&& a_1+a_2-a_4:[0,15] \qquad
   && \underline{a_1-a_3-a_4: [-20,0]}\qquad && a_2-a_3-a_4: [-15,0]\quad
\end{alignat*}
Finally, we can make both of the preceding adjustments in $\N_3$:
\emph{Decrease both $\uc(a_1)$ and $\uc(a_4)$
from $K$ (``very large number'') to $10$}, so that the substitution
of $\N_2$ is both ``input safe'' and ``output safe''.
\end{example}

Based on the discussion in Example~\ref{ex:angelic-vs-demonic},
in the presence of \emph{demonic non-determinism}, we need a 
notion of subtyping more restrictive than ``$\subTsym$'',
which we call ``strong subtyping'' and denote by ``$\SubTsym$''.

\begin{definition}{Strong Subtyping}
\label{def:strong-subtyping}
Let $T, U : \power{\aaa_{\text{in,out}}}\to\intervals{\reals}$ be
principal typings for similar networks $\M$ and $\N$, respectively, both
with the same set $\aaa_{\text{in,out}}$ of input/output arcs. 
We say $T$ is \emph{input-safe} for $U$ iff:
\begin{itemize}[itemsep=2pt,parsep=2pt,topsep=5pt,partopsep=0pt] 
\item For every $f_{\text{in}}:\aaa_{\text{in}}\to\nreals$ satisfying
      $\rest{T}{\power{\aaa_{\text{in}}}}$, and for every 
      $g:\aaa_{\text{in,out}}\to\nreals$ extending $f_{\text{in}}$,
      it holds that:\\
      $g$ satisfies $T$ $\Leftrightarrow$ 
      $g$ satisfies $U$.
\end{itemize}
We say $T$ is \emph{output-safe} for $U$ iff:
\begin{itemize}[itemsep=2pt,parsep=2pt,topsep=5pt,partopsep=0pt] 
\item For every $f_{\text{out}}:\aaa_{\text{out}}\to\nreals$ satisfying
      $\rest{T}{\power{\aaa_{\text{out}}}}$, and for every 
      $g:\aaa_{\text{in,out}}\to\nreals$ extending $f_{\text{out}}$,
      it holds that:
      $g$ satisfies $T$ $\Leftrightarrow$ 
      $g$ satisfies $U$.
\end{itemize}
We say $T$ is \emph{safe} for $U$, or say $U$ is a 
\emph{strong subtyping} of $T$ and write $\SubT{U}{T}$, iff $T$ is
both input-safe and output-safe for $U$. Strong subtyping expresses
the condition for the safe substitution of $\N$ (whose principal typing
is $U$) for $\M$ (whose principal typing is $T$) in the presence of
demonic non-determinism.
\end{definition}

We state without proof some simple properties of ``$\SubTsym$'',
the starting point of an investigation of how to extend our
typing theory to handle demonic non-determinism.

\begin{fact}[Strong Subtyping is a Partial Order]
Let $S, T, U : \power{\aaa_{\text{\rm in,out}}}\to\intervals{\reals}$ be
principal typings (of some similar networks) over the same input/output set
$\aaa_{\text{\rm in,out}}$.
\begin{enumerate}[itemsep=2pt,parsep=2pt,topsep=5pt,partopsep=0pt] 
\item $\SubT{T}{T}$ (reflexivity).
\item If $\SubT{T}{U}$ and $\SubT{U}{T}$, then $T = U$
      (anti-symmetry).
\item If $\SubT{S}{T}$ and $\SubT{T}{U}$, then $\SubT{S}{U}$
      (transitivity). 
\item If $\SubT{T}{U}$, then $\subT{T}{U}$, but not the other way around
      in general. (A counter-example for the converse is in 
       Example~\ref{ex:angelic-vs-demonic},
       where $\subT{T_2}{T_{3}}$ but $\NotSubT{T_2}{T_{3}}$.)
\end{enumerate}
Points 1-3 say that ``$\SubTsym$'' is a partial order, 
just as ``$\subTsym$'' is, and
point 4 says that this partial order can be embedded in the
partial order of ``$\subTsym$''.
\end{fact}

\end{document}